\documentclass[11pt]{article}

\input{preamble.tex}

\usepackage{todonotes}
\usepackage[T1]{fontenc}
\usepackage{amsthm}
\usepackage{amsmath}
\usepackage{amsfonts}
\usepackage{thmtools,thm-restate}
\usepackage{comment}
\usepackage[normalem]{ulem}

\usepackage{enumitem,kantlipsum}

\date{}
\begin{document}

\title{A Parameterized Approximation Scheme for {\sc Min $k$-Cut}}
\author{
Daniel Lokshtanov\thanks{University of California, Santa Barbara, USA. Emails: \texttt{daniello@ucsb.edu}, \texttt{vaishali@ucsb.edu}}
 \and Saket Saurabh\thanks{The Institute of Mathematical Sciences, HBNI, Chennai, India, and University of Bergen, Norway. \texttt{saket@imsc.res.in}}
 \and Vaishali Surianarayanan\footnotemark[1]
 % \thanks{University of California, Santa Barbara, USA. \texttt{vaishali@ucsb.edu}}
}
\maketitle

\begin{abstract}
%!TEX root = main.tex
In the {\sc Min $k$-Cut} problem, input is an edge weighted graph $G$ and an integer $k$, and 
the task is to partition the vertex set into $k$ non-empty sets, such that the total weight of the edges with 
endpoints in different parts is minimized.  When $k$ is part of the input, the problem is NP-complete and hard to approximate within any factor less than $2$. Recently, the problem has received significant attention from the perspective of parameterized approximation.  Gupta {\em et al.} [SODA 2018] initiated the study of FPT-approximation for  the {\sc Min $k$-Cut} problem and  gave an $1.9997$-approximation algorithm running in time $2^{\cO(k^6)}n^{\cO(1)}$. Later, the same set of authors~[FOCS 2018]  designed an $(1 +\epsilon)$-approximation algorithm that runs in time $(k/\epsilon)^{\cO(k)}n^{k+\cO(1)}$, and a $1.81$-approximation algorithm  running in time $2^{\cO(k^2)}n^{\cO(1)}$. More, recently, Kawarabayashi and Lin~[SODA 2020] gave a $(5/3 + \epsilon)$-approximation for {\sc Min $k$-Cut} running in time  $2^{\cO(k^2 \log k)}n^{\cO(1)}$. 

In this paper we give a parameterized approximation algorithm with best possible approximation guarantee, and best possible running time dependence on said guarantee (up to Exponential Time Hypothesis (\ETH) and constants in the exponent). In particular, for every 
$\epsilon > 0$, the algorithm obtains a $(1 +\epsilon)$-approximate solution in time $(k/\epsilon)^{\cO(k)}n^{\cO(1)}$. The main ingredients of our algorithm are: a simple sparsification procedure, a new polynomial time algorithm for decomposing a graph into highly connected parts, and a new exact algorithm with running time $s^{\cO(k)}n^{\cO(1)}$ on 
%connected 
unweighted (multi-) graphs. Here, $s$ denotes the number of edges in a minimum $k$-cut. The latter two are of independent interest.

\end{abstract}

\newpage
\pagestyle{plain}
\setcounter{page}{2}
%!TEX root = main.tex

\section{Introduction}
\label{section:Introduction}

In this article we focus on the \mkk problem from the perspective of parameterized approximation. Here input is an edge weighted graph $G$ and an integer $k$, and  the task is to partition the vertex set into $k$ non-empty sets, say $\tilde{P}$, such that the total weight of the edges with 
endpoints in different parts is minimized.  We call such a partition as {\em optimal $k$-cut},  or {\em minimum $k$-cut} or simply a {\em $k$-cut}. 
 For $k=2$, this is just the classic {\sc Global Min-Cut} problem, which can be solved in polynomial time. In fact, for every fixed $k$, the problem is known to be polynomial time solvable~\cite{GoldschmidtH94}. However, when $k$ is part of the input, the problem is \NPC~\cite{GoldschmidtH94}.
 %\todo{define k-cut}

%$\tilde{O}(n^{2(k-1)}\log^{O(1)}n)$
%Different directions the problem has been studied traditionally 
Traditionally, the problem has been extensively studied from three perspectives: (a) exact algorithms for small values of $k$;  (b) 
combinatorial upper bound on the number of minimum $k$-cuts; and (c) approximation algorithms.  
In particular, already in early 90's, Goldschmidt and Hochbaum~\cite{GoldschmidtH94} gave the first polynomial-time algorithm for fixed $k$, with running time $\cO(n^{(0.5-o(1))k^2})$. In 1996, Karger and Stein~\cite{KargerS96} obtained their ingenious {\em contraction algorithm}, a randomized algorithm with running time $\tilde{\cO}(n^{2(k-1)})$\footnote{$\tilde{\cO}$ hides the poly-logarithimic factor in the running time.}. The same algorithm immediately yields a proof that the number of $k$-cuts of minimum weight in any undirected graph on $n$ vertices is upper bounded by  $\tilde{\cO}(n^{2(k-1)})$. In 2008, Thorup~\cite{Thorup08} obtained a deterministic algorithm that essentially matched the running time of Karger and Stein~\cite{KargerS96} and in 2019, Chekuri et al.~\cite{DBLP:journals/siamdm/ChekuriQX20} used a LP based approach to improve the running time by a factor of $n$. 
%improved the deterministic algorithm of Goldschmidt and Hochbaum~\cite{GoldschmidtH94}, by designed an algorithm with running time  $\tilde{O}(mn^{2(k-2)})$. 
%
Thorup's algorithm is based on a tree-packing theorem which allows to efficiently find a tree that crosses the optimal $k$-cut at most $2k-2$ times. Then the algorithm simply enumerates over all possible sets of $2k-2$ edges of this tree and all ways of merging the components of the tree minus these $2k-2$ edges into $k$ non-empty groups. 
From the perspective of approximation algorithms there are several $2(1-\frac{1}{k})$-approximation algorithms, that run in time polynomial in $n$ and $k$~\cite{NaorR01,ravi2008approximating,SaranV95}. This approximation ratio cannot be improved assuming the Small Set Expansion Hypothesis (\SSE)~\cite{Manurangsi18}. 

Thus, as of early 2018, the status was as follows. The direction of polynomial time approximation algorithms was essentially settled, with factor $2(1-\frac{1}{k})$ approximation algorithms and matching lower bounds. The fastest algorithm for small $k$ had running time $\tilde{\cO}(n^{2(k-1)})$, and a $f(k)n^{o(k)}$ lower bound on the running time~\cite{DBLP:journals/entcs/DowneyEFPR03,DBLP:books/sp/CyganFKLMPPS15} was known under the \ETH. It remained an interesting research direction to find an algorithm with an $n^{ck}$ running time and prove that the constant $c$ is optimal under reasonable complexity assumptions. The best upper bound on the number of minimum $k$-cuts was  $\tilde{\cO}(n^{2(k-1)})$, while the best lower bound was essentially $\Omega((n/(k-1)^{k-1})$, what you obtain from the complete $(k-1)$-partite graph. Since early 2018 the problem has seen a remarkable level of activity~\cite{DBLP:journals/corr/abs-1810-06864,DBLP:conf/focs/GuptaLL18,DBLP:conf/soda/GuptaLL18,DBLP:conf/stoc/GuptaLL19,DBLP:conf/focs/Li19,DBLP:journals/corr/abs-1906-00417,DBLP:conf/soda/KawarabayashiL20}, with a new research direction of parameterized approximation being initiated, and the two remaining established research directions (exact algorithms and combinatorial bounds) essentially being completely resolved.

For exact algorithms for small values of $k$, Gupta et al.~\cite{DBLP:conf/focs/GuptaLL18} made the first improvement in two decades, by designing an  algorithm with running time $\cO(k^{\cO(k)}n^{(\frac{2\omega}{3}+o(1))k})$ for graphs with polynomial integer weights.  Subsequently, for unweighted graphs, Li~\cite{DBLP:conf/focs/Li19} obtained an algorithm with running time $\cO(k^{\cO(k)}n^{(1+o(1)k})$. 
For the original formulation of the problem (with general integer weights) Gupta et al.~\cite{DBLP:conf/stoc/GuptaLL19} showed an $\cO(n^{(1.98+o(1))k})$-time algorithm. Their algorithm also implied a new and improved $\cO(n^{(1.98+o(1))k})$ upper bound on the number of minimum $k$-cuts. Finally Gupta et al.~\cite{DBLP:journals/corr/abs-1906-00417} fully resolved the problem by showing that for every fixed $k\geq 2$, the Karger-Stein algorithm~\cite{KargerS96}  outputs any fixed minimum $k$-cut with probability at least $\widehat{\cO}(n^{-k})$, where $\widehat{\cO}(\cdot)$ hides a $2^{\cO(\ln \ln n)^2}$ factor. This immediately gives an extremal bound of $ \widehat{\cO}(n^{k})$, on the number of minimum $k$-cuts in an $n$-vertex graph and an algorithm for \mkk in similar running time. Both the extremal bound and the running time of the algorithm are essentially tight (under reasonable assumptions). Indeed the extremal bound matches known lower bounds up to $\widehat{\cO}(1)$ factors, while any further improvement to the exact algorithm would imply an improved algorithm for {\sc Max-Weight $k$-Clique}~\cite{AbboudWW14,DBLP:conf/icml/BackursT17}, which has been conjectured not to exist.

The lower bound of $(2 - \epsilon)$ on polynomial time approximation algorithms~\cite{Manurangsi18}, coupled with the $f(k)n^{o(k)}$ running time lower bound for exact algorithms~\cite{DBLP:journals/entcs/DowneyEFPR03,DBLP:books/sp/CyganFKLMPPS15} naturally leads to the question of parameterized approximation: how good approximation ratio can we achieve if we allow the algorithm to have running time $f(k)n^{\cO(1)}$? In 2018, Gupta et al.~\cite{DBLP:conf/soda/GuptaLL18} were the first to ask this question, and showed that relaxing the running time requirement from polynomial to $f(k)n^{\cO(1)}$ does lead to an improved approximation guarantee, by giving an $1.9997$-approximation algorithm for \mkk running in time $2^{\cO(k^6)}n^{\cO(1)}$. 
%
%
% initiated the study of \mkk   from the perspective of parameterized approximation. That is, could we get an approximation algorithm with factor $2-\delta$, for some fixed constant $\delta >0$, running in time $f(k)\cdot n^{O(1)}$, for some function $f$ depending only on $k$.
%
%
%
%ut a natural question arise before this is whether \mkk is fixed parameter tractable (\FPT) with respect to the parameter $k$. That is, could be get an algorithm solving  \mkk with running time $g(k) n^{O(1)}$, for some function $g$ depending  on $k$ alone. Unfortunately, the problem is known to be \WOH~\cite{DBLP:journals/entcs/DowneyEFPR03}. In fact, assuming \ETH, even for unweighted graphs there is no algorithm that solves \mkk in time $h(k) n^{o(k)}$, for some computable function $h$~\cite{DBLP:journals/entcs/DowneyEFPR03,DBLP:books/sp/CyganFKLMPPS15}.
%
%In 2018, Gupta et al.~\cite{DBLP:conf/soda/GuptaLL18} initiated the study of \mkk   from the perspective of parameterized approximation. That is, could we get an approximation algorithm with factor $2-\delta$, for some fixed constant $\delta >0$, running in time $f(k)\cdot n^{O(1)}$, for some function $f$ depending only on $k$.
%
%Thus, if we desire to design an algorithm with running time $f(k)\cdot n^{O(1)}$ for \mkk, we must sacrifice accuracy. 
%
%
%
%As a proof of concept, Gupta et al.~\cite{DBLP:conf/soda/GuptaLL18}  gave an $1.9997$-approximation algorithm for \mkk running in time $2^{O(k^6)}n^{O(1)}$. 
Subsequently the same set of authors~\cite{DBLP:conf/focs/GuptaLL18} designed an $(1 +\epsilon)$-approximation algorithm that runs in time $(k/\epsilon)^{\cO(k)}n^{k+\cO(1)}$, and a $1.81$-approximation algorithm  running in time $2^{\cO(k^2)}n^{\cO(1)}$. More recently, Kawarabayashi and Lin~\cite{DBLP:conf/soda/KawarabayashiL20} gave a $(5/3 + \epsilon)$-approximation for {\sc Min $k$-Cut} running in time  $2^{\cO(k^2 \log k)}n^{\cO(1)}$. In  this paper we bring the direction of parameterized approximation to its natural conclusion by giving an algorithm with best possible approximation guarantee, and best possible running time dependence on said guarantee (up to \ETH and constants in the exponent).  In particular, for every 
$\epsilon > 0$, the algorithm obtains a $(1 +\epsilon)$-approximate solution in time $(k/\epsilon)^{\cO(k)}n^{\cO(1)}$.

\begin{theorem}
\label{thm:mainFPTAS}
There exists a randomized algorithm that given a graph $G$, weight function $w : V(G) \rightarrow \mathbb{Q}_{\geq 0}$, integer $k$, and an $\epsilon > 0$, runs in time $(k/\epsilon)^{\cO(k)}n^{\cO(1)}$ and outputs a partition $\tilde{P}$ of $V(G)$ into $k$ non-empty parts. With probability at least $\prup$ the weight of $\tilde{P}$ is at most $(1+\epsilon)$ times the weight of an optimum $k$-cut of $G$. By weight of $\tilde{P}$, we mean the total weight of the edges with endpoints in different parts. 
\end{theorem}

%\todo[inline]{Here state our theorem.}

\paragraph{Overview of the Algorithm.}
A preliminary step of our algorithm is to run the $2$-approximation algorithm~\cite{NaorR01,ravi2008approximating,SaranV95} to get an estimate of the value of $\opt(G,k)$ (the weight of an optimal $k$-cut in $G$). A standard rounding procedure (similar to the well-known Knapsack PTAS~\cite{DBLP:books/daglib/0015106}) reduces the problem in a $(1+\epsilon)$-approximation preserving manner to unweighted multi-graphs with at most $m^2/\epsilon$ edges. From now on, {\em throughout this paper}, we will assume that our input graph is an {\em unweighted multigraph}. As long as the graph has a non-trivial $2$-cut (non-trivial means that there is at least one edge crossing the cut) of weight at most 
$\frac{\epsilon \cdot  \opt(G,k)}{k-1}$, we remove all edges of this $2$-cut. This step will be repeated less than $k-1$ times, since if it is repeated $k-1$ times we have cut the graph into at least $k$ connected components using at most  $\epsilon \cdot \opt(G,k)$ edges. Thus the procedure stops before this, and at that point we know that (i) we have included at most  $\epsilon \cdot \opt(G,k)$ edges in our solution, and (ii) the remaining graph has no non-trivial $2$-cut with at most $\epsilon  \cdot \opt(G,k)$ edges. We now solve the problem independently on each connected component, after guessing how many parts each component is split into (this introduces a $4^k$ overhead in the running time). Thus from now on we will assume that our input graph is connected and has no non-trivial $2$-cut of weight at most $\frac{\epsilon \cdot  \opt(G,k)}{k-1}$.

Having ensured that every non-trivial $2$-cut has weight at least $\epsilon \cdot \frac{ \opt(G,k)}{k-1}$ we proceed as follows. We set a probability $p = \frac{100 \log n}{\opt(G,2) \cdot \epsilon^2}$. Then for every edge of $G$ we flip a biased coin and keep the edge with probability $p$ and delete it with probability $1-p$. Call the resulting subgraph $G'$. A relatively simple probability computation (very similar to that of randomized cut sparsifiers~\cite{BenczurK15,Karger94}) combining Chernoff bounds with Karger and Stein's~\cite{KargerS96} bound on the number of minimum $2$-cuts shows that with high probability it holds that for {\em every} partition $\tilde{P}$ of $V(G) = V(G')$ into $k$ parts, the number of edges of $G'$ crossing the partition $\tilde{P}$ is within a $1 \pm \epsilon$ factor of $p$ times the number of edges crossing it in $G$. Thus, for purposes of $(1+\epsilon)$-approximation  one may just as well find a $(1+\epsilon)$-approximate solution in $G'$. However, in $G'$ we know that 
$$\opt(G',k) \simeq p \cdot \opt(G,k) =  \frac{100 \log n}{\opt(G,2) \cdot \epsilon^2} \cdot \opt(G,k) \leq \frac{k \cdot 100 \log n}{\epsilon^3}.$$
In the last transition we used the assumption that $\opt(G,2) \geq \epsilon \cdot \opt(G,k)/k-1$. This is a pretty small number of cut edges, which naturally leads to the idea of trying to apply parameterized algorithms with parameter $s$, the number of edges in the optimal $k$-cut. 

The \mkk problem is also quite well studied with $s = \opt(G,k)$ as a parameter. A brute force algorithm trying all edge sets of size $s$ solves the problem in time $\cO(n^{2s+\cO(1)})$. Marx~\cite{DBLP:journals/tcs/Marx06} posed as an open problem in 2004 whether \mkk is {\em fixed parameter tractable} when parameterized by $s$, that is whether the problem admits an algorithm with running time $f(s)n^{\cO(1)}$ for any function $f$. Kawarabayashi and Thorup~\cite{KT11} resolved this problem in the affirmative, and obtained an algorithm with running time $s^{s^{\cO(s)}}n^2$. Chitnis et al.~\cite{randcontr} gave an algorithm with running time $s^{\cO(s^2)}n^{\cO(1)}$, improving the dependence on $s$ from double exponential to single exponential.  Finally, Cygan et al.~\cite{DBLP:journals/corr/abs-1810-06864} showed that the problem is solvable in time $s^{\cO(s)} n^{\cO(1)}$. 
    Unfortunately, applying the algorithm of Cygan et al.~\cite{DBLP:journals/corr/abs-1810-06864} directly on $G'$ yields an algorithm with running time $(\frac{k\log n}{\epsilon^3})^{\cO(\frac{k\log n}{\epsilon^3})}$ which grows super-polynomially with $n$, even for constant $\epsilon$ and $k$. Even if one obtains a parameterized algorithm for \mkk with running time $2^{\cO(s)} n^{\cO(1)}$ (which is an interesting open problem in  itself), this would only lead to a $2^{\cO(\frac{k\log n}{\epsilon^3})} = n^{\cO(\frac{k}{\epsilon^3})}$ time algorithm, which is slower than the classic exact $n^{2(k-1)}$ time algorithm of Karger and Stein~\cite{KargerS96}. Thus, at a glance, the parameterized approach seems to come close, but hit a dead end. The key insight of our algorithm is that even though $s$ is ``small'', that $k$ is much smaller, and that the hard instances for parameter $s$ seem to have a value of $k$ close to $s$. This leads to the natural problem of whether it is possible to design a parameterized algorithm with parameters $k$ and $s$ that substantially outperforms the algorithm of Cygan et al.~\cite{DBLP:journals/corr/abs-1810-06864} when $k \ll s$. Our main technical contribution is precisely such an algorithm. We state this algorithm as a separate theorem.

\begin{theorem}
\label{thm:paraAlg}
%\label{DP-main-theorem}
There exists an algorithm that given an unweighted multigraph $G$ and integers $k$ and $s$, determines whether $G$ has a $k$-cut of weight at most $s$ in times $s^{\cO(k)}n^{\cO(1)}$.%determines whether $G$ has a $k$-cut $\tilde{P}$ having weight at most $s$ and returns $\tilde{P}$ if it exists in time $s^{\cO(k)}n^{\cO(1)}$. %determines whether G has a $k$-cut of
\end{theorem}

%For unweighted graph, by optimal $k$-cut or simply $k$-cut we mean the number of edges with endpoints in different parts. 
The algorithm of Theorem~\ref{thm:paraAlg} can easily be extended using self reduction to also produce a partition 
$\tilde{P}$ of  $G$ with weight at most $s$, if it exists. The algorithm of Theorem~\ref{thm:paraAlg} is based on the $s^{\cO(s)} n^{\cO(1)}$ time algorithm of Cygan et al.~\cite{DBLP:journals/corr/abs-1810-06864}.

Just as the algorithm of Cygan et al.~\cite{DBLP:journals/corr/abs-1810-06864}, we compute a tree decomposition of the input graph $G$, such that the adhesions (cuts) of the decomposition are small, and the bags (the building blocks of the decomposition) are highly edge-connected. Unfortunately we are not able to use their decomposition theorem as a black box, because their running time to compute the decomposition is already  $s^{\cO(s)} n^{\cO(1)}$. We make our own decomposition theorem with weaker properties than the one in~\cite{DBLP:journals/corr/abs-1810-06864}, but computable in polynomial time.  The construction of the decomposition theorem follows the template of Cygan et al.~\cite{DBLP:journals/corr/abs-1810-06864},  replacing their most computationally intensive step with a polynomial time approximation algorithm. 

The most technically challenging part of our algorithm is how to compute an optimal $k$-cut in time 
$s^{\cO(k)}n^{\cO(1)}$, when the decomposition into highly connected pieces is provided as  input. The only known way how to exploit such a decomposition is using dynamic programming. However, the dynamic programming seems to require at least $2^{\Omega(s)}$ states, even for the stronger decomposition of Cygan et al.~\cite{DBLP:journals/corr/abs-1810-06864}, let alone for our weaker structural decomposition theorem.  Here, the fact that $k$ is much smaller than $s$ comes to rescue. We prove that the dynamic programming algorithm can be ``trimmed'' to only populate $s^{\cO(k)}$ states of the dynamic programming table. This trimming procedure is done by inspecting how an optimal $k$-cut can interact with a tree obtained from Throup's~\cite{Thorup08} tree-packing theorem.  A substantial amount of technical weight lifting is required to show that the trimmed dynamic programming table for a bag of a tree decomposition can be computed efficiently from the trimmed dynamic programming tables of its children. 

The algorithm of Theorem~\ref{thm:paraAlg} immediately completes the proof of our $(1+\epsilon)$-approximation  algorithm (Theorem~\ref{thm:mainFPTAS}). Specifically, applying this algorithm to the graph $G'$ obtained by our sparsification procedure yields a 
$(1+\epsilon)$-approximation algorithm with running time   
\begin{eqnarray*}
\bigg(\frac{k \cdot 100 \log n}{\epsilon^3}\bigg)^{\cO(k)}n^{\cO(1)}& = & (k/\epsilon)^{\cO(k)} (\log n)^{\cO(k)}n^{\cO(1)}
 \leq  (k/\epsilon)^{\cO(k)} (k^{\cO(k)}+n) n^{\cO(1)}\\
& \leq &2^{\cO(k \log (\frac{k}{\epsilon}) ) }}n^{\cO(1). 
\end{eqnarray*}
 Setting $\epsilon=\frac{1}{n^{\cO(1)}}$, shows that a $2^{o(k \log(\frac{k}{\epsilon}) )} n^{\cO(1)}$ time algorithm would yield an exact algorithm with running time $n^{o(k)}$, violating \ETH. 

\paragraph{Guide to the paper.} We start by giving the notations and preliminaries that we use throughout the paper in Section~\ref{section:notation}. This section is best used as a reference, rather than being read linearly.
In Section~\ref{section:sparsification} we prove the $(1+\epsilon)$-approximation preserving reduction to instances with optimum $k$-cuts of size $\cO(\frac{k \log n}{\epsilon^3})$. Since the proofs here are so similar to those for graph sparsifiers, a reader eager to get to the key ideas of the paper should look at the statement and proof of Lemma~\ref{RandLemma1}, as well as the statement of Lemma~\ref{RandLemma2}, and skip further. 
In Section~\ref{section:decomposition}, we state and prove Theorem~\ref{thm:decomposition}, the polynomial time edge-unbreakable decomposition theorem. In the proof of Theorem~\ref{thm:decomposition}, Subsection~\ref{subsec:flw} contains a new approximation algorithm for finding small balanced edge-separators, while Subsections~\ref{sec:refine} and~\ref{sec:decoAlgo} are 
almost word-by-word identical to proofs of 
%are minor modifications of the proofs of 
Cygan et al~\cite{DBLP:journals/corr/abs-1810-06864} and may be skipped. 
Section~\ref{section:DP} contains our main technical thrust: a new exact algorithm for unweighted multi graphs with running time $s^{\cO(k)}n^{\cO(1)}$, the algorithm computes and uses a decomposition such as the one provided by Theorem~\ref{thm:decomposition}. 
Finally, in Section~\ref{section:combining}, we combine all the results obtained in the previous sections and complete the proof of our main result (Theorem~\ref{thm:mainFPTAS}). We conclude the paper with some interesting open problems in Section~\ref{section:conclusion}.

\section{Notation and Preliminaries}
\label{section:notation}
In this section we give notations, and definitions that we use throughout the paper. Unless specified we will be using all general graph terminologies from the book of Diestel~\cite{DBLP:books/daglib/0030488}.  

%The notations $G,n,m,k$ have fixed meanings in all the sections.

\subsection{Graph and Set Theoretic Definitions and Notations}
Given a graph $G$, we use $V(G)$ and $E(G)$ to denote the set of vertices and edges, respectively. We use $\cc(G)$ to denote the number of connected components of $G$.  
In this paper our graph could be a multigraph, that is, we allow parallel edges between vertices, even though we might start with a simple graph. 
Given a subset $X$ of vertices in $V(G)$, we denote by $\delta(X)$ the number of edges in $G$ that have exactly one end point in $X$. By the term {\em smoothing} a vertex $v$ of degree $2$ in a graph $G$, we mean the operation of removing $v$ from the graph and adding an edge between two of its neighbors. 

For a set $U$, an $\ell$-{\em partition} of $U$ is a family $\tilde{P} = \{P_1, P_2, \ldots, P_\ell\}$ of non-empty, pairwise disjoint sets whose union $\bigcup_{i=1}^{\ell} P_i$ is equal to $U$. A {\em partition} of $U$ is an $\ell$-partition for some positive integer $\ell$. We refer to $P_1, \ldots, P_\ell$ as the {\em parts} of $\tilde{P}$ and refer to $\ell$ as the size of the partition and denote it by $|\tilde{P}|$. 
%\todo{No empty set?} %We remark that one or more of the parts may be equal to the empty set. 
%
An edge $uv$ with endpoints in $U$ {\em crosses} $\tilde{P}$ if $u$ and $v$ are in different parts of $\tilde{P}$. The number of times an edge set $X$ crosses $P$ is defined as the number of edges in $X$ that cross $P$ and denoted by $\delta_X(\tilde{P})$. That is, $X'\subseteq X$ is the set containing all edges in $X$ that crosses $\tilde{P}$ and $\delta_X(\tilde{P})=|X'|$.
The {\em projection} of a partition $\tilde{P}$ onto a subset $S$ of $U$ is the partition $\tilde{P}' = \{P_1 \cap S, P_2 \cap S, \ldots, P_\ell \cap S\}$ with the empty sets removed. A partition $\tilde{P}'$ is a {\em refinement} of $\tilde{P}$ if every part of $\tilde{P}'$ is a subset of some part of $\tilde{P}$, see Figure~\ref{fig:refinementTtree}(a). We say that a partition $\tilde{P}$ is {\em refined} by $\tilde{P}'$ if $\tilde{P}'$ is a refinement of $\tilde{P}$. The $\emph{combining}$ of a set of parts $\mathcal{P}$ in $\tilde{P}$ is the operation of removing all the parts in $\mathcal{P}$ from $\tilde{P}$ and adding their union $\underset{P\in \mathcal{P}}{\cup}P$ as a part in $\tilde{P}$.
\begin{figure}[t]
    \centering
    \includegraphics[width=0.9\textwidth]{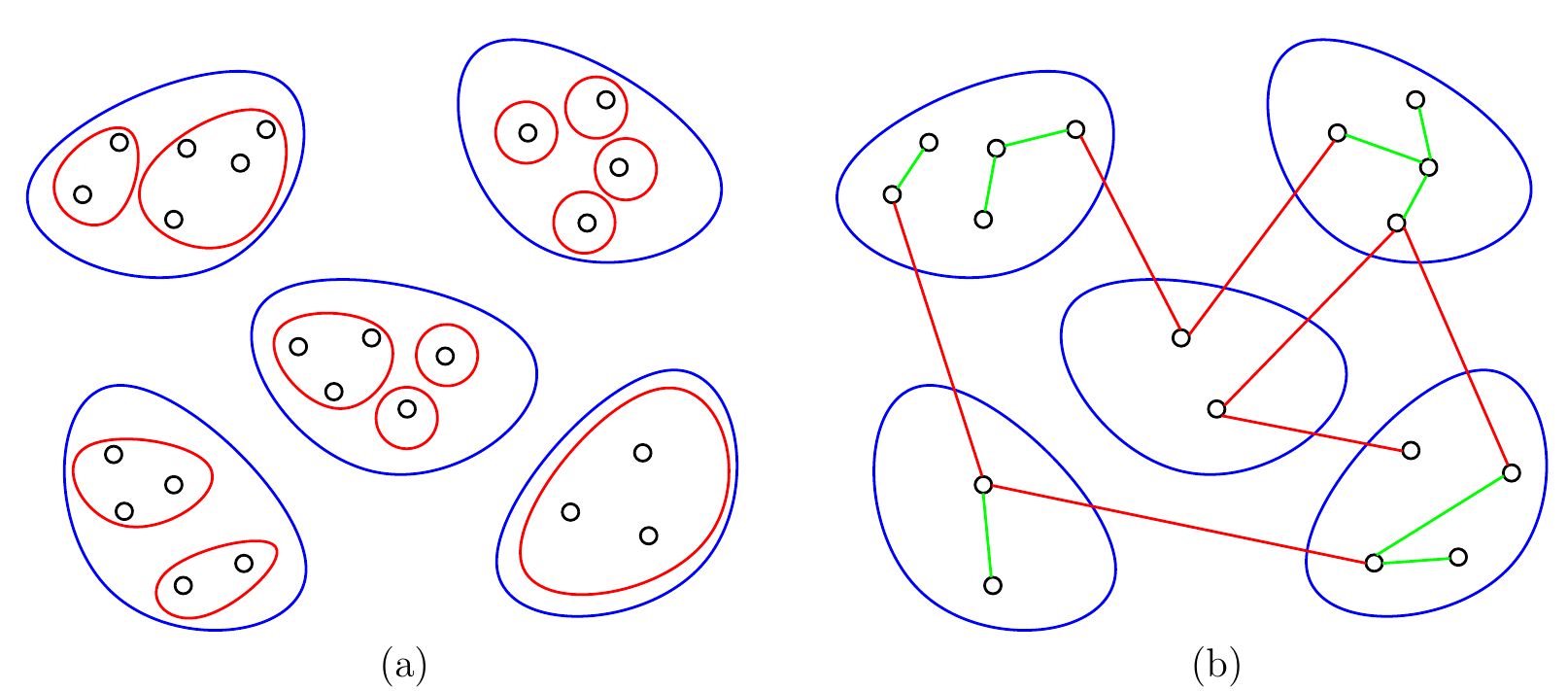}
    \caption{(a) Example of a partition $\tilde{P}$' that is a refinement of a partition $\tilde{P}$. Partition $\tilde{P}'$ is depicted in red and parition $\tilde{P}$ is depicted in blue. (b) Example of a tree $T$ that crosses a partition $\tilde{P}$ at most $2k-2$ times. The red edges are the edges from $T$ that cross the partition and the green edges are the remaining edges of $T$.}
    \label{fig:refinementTtree}
\end{figure}
%

%Let $U_1,\ldots ,U_\ell$ be $\ell$ subsets of a set $U$ and $X$ be a superset of the union of pair wise common intersections of $U_j$'s, that is $\underset{j,j'\in \{1,\ldots,\ell\}, j\neq j'}\bigcup U_j\cap U_{j'} \subseteq X$. We define the {\em union} of $\ell$ partitions $\tilde{P}_1$ of $V_1$, $\tilde{P}_2$ of $V_2$, $\ldots$, and $\tilde{P}_l$ of $V_l$ with respect to a partition $\tilde{P}_X$ of $X$, as the partition $\tilde{P}$ of $\underset{1\leq j\leq l}\bigcup V_i$ where the projection of $\tilde{P}$ on $X$ is the same as $\tilde{P}_X$ and the projection of $\tilde{P}$ on $V_j$ is the same as $\tilde{P}_{j}$ for all $1\leq j\leq l$.
%\todo[inline]{What does this mean? I do not understand the above paragraph at all.}
        
A {\em $k$-cut} $\tilde{S}$ of a graph $G$ is a $k$-partition of $G$. Let $G$ be an edge weighted graph $G$, then the {\em weight of the  $k$-cut}, denoted by $w(G,\tilde{S})$ is simply the sum of the weights of edges with endpoints in different parts of $\tilde{S}$. When $G$ is an unweighted multigraph, then 
%such that $G[S]$ is connected for each part $S\in \tilde{S}$.
the weight of $\tilde{S}$ is the number of  edges with endpoints in different parts of $\tilde{S}$ (if there are $q$ edges between a pair of vertices then it adds $q$ to the sum). We denote this also by $w(G,\tilde{S})$. In the cases where the graph $G$ is clear from the context, we just use $w(\tilde{S})$ instead of $w(G,\tilde{S})$. An {\em optimal $k$-cut} or a {\em minimum $k$-cut}  of $G$ is a $k$-cut of $G$ that has the {\em minimum weight} among all $k$-cuts of $G$. We denote the weight of an optimal $k$-cut by 
$\opt(G,k)$.   A {\em non-trivial cut} means that there is at least one edge crossing the cut. We remark that in some parts of the paper, we refer to a $k$-cut of a graph $G$ as a $k$-partition of $V(G)$ but it will be clear from the context. 

For a graph $G$, {\em an edge cut} is a pair $A, B \subseteq V(G)$ such that $A\cup B=V(G)$ and $A\cap B=\emptyset$. The order of an edge cut $(A,B)$ is $|E(A,B)|$, that is, the number of edges with one endpoint in $A$ and the other in $B$. Observe that edge cut and $2$-cut are the same. 
\begin{definition}[unbreakability]
A set $X \subseteq V(G)$ is \emph{$(q,s)$-edge-unbreakable} if every edge cut $(A,B)$ of order at most $s$
satisfies $|A \cap X| \leq q$ or $|B \cap X| \leq q$.
\end{definition}
A \emph{separation} is a pair $A, B \subseteq V(G)$ such that $A \cup B = V(G)$ and $E(A \setminus B, B \setminus A) = \emptyset$. The order of a separation $(A,B)$ is $|A \cap B|$.

\subsection{$T$-tree and Treewidth}
An important notion that underlies our algorithm is a notion of {\em $T$-trees } introduced in the seminal work of 
 Thorup~\cite{Thorup08}. 
 %improved the deterministic algorithm of Goldschmidt and Hochbaum~\cite{GoldschmidtH94}, by designed an algorithm with running time  $\tilde{O}(mn^{2(k-2)})$. 
 For an instance $(G,k)$ of \mkk, a $T$-{\em tree} of $G$ is a spanning tree of $G$ such that there exists an optimal $k$-cut $S^\star$ of $G$ such that $T$ crosses $S^\star$ at most $2k-2$ times. For an illustration of $T$-tree we refer to Figure~\ref{fig:refinementTtree}(b). 
%\begin{figure}[t]
 %   \centering
  %  \includegraphics[width=0.4\textwidth]{Ttree.pdf}
   % \caption{Example of a tree $T$ that crosses a $k=5$ partition at most $2k-2$ times. The red edges are the edges from $T$ that cross the partition and the green edges are the remaining edges of $T$.}  %\label{fig:Ttree}
%\end{figure}
%A graph is said to be $(s,q)$-unbreakable if for every pair $(V_1,V_2)$ of disjoint subsets of $V$ having size greater than $q$ each, the number of edges having one end point in $S_1$ and the other in $S_2$ is at least $s+1$.

Another important notion that we need is of tree decomposition where bags are ``highly connected''. Towards this we first define tree decomposition, treewidth and associated notions and notations that we make use of. For a rooted tree $\tau$ and vertex $v \in V(\tau)$ we denote by $\tau_v$ the subtree of $\tau$ rooted at $v$. We refer to the vertices of $\tau$ as nodes.

A {\em tree decomposition} of a graph $G$ is a pair $(\chi,\tau)$ where $\chi$ is a rooted tree and $\tau$ is a function from $V(\tau)$ to $2^{V(G)}$ such that the following three conditions hold. 
\begin{description}
\setlength{\itemsep}{-2pt}
\item[(T1)] $\underset{t\in V(\tau)}{\bigcup}\chi(t) = V(G)$;
\item[(T2)] For every $uv\in E(G)$, there exists a node $t\in \tau$ such that $\chi(t)$ contains both $u$ and $v$; and 
\item[(T3)] For every vertex $u\in V(G)$, the set $T_u=\{t\in V(\tau): u\in \chi(t)\}$, i.e., the set of nodes whose corresponding bags contain $u$, induces a connected subtree of $\tau$. 
\end{description}
For every $t \in V(\tau)$ a set $\chi(t) \subseteq V(G)$, is called a \emph{bag}.  For a tree decomposition $(\tau,\chi)$ fix an edge $e = tt' \in E(\tau)$. The deletion of $e$ from $\tau$ splits $\tau$ into two trees $\tau_1$ and $\tau_2$, and naturally induces a separation $(X_1,X_2)$ in $G$
with $X_i := \bigcup_{t \in V(\tau_i)} \chi(t)$, which we henceforth call \emph{the separation associated with $e$}.
The set $\adh_{\tau,\chi}(e) := X_1 \cap X_2 = \chi(t) \cap \chi(t')$ is called the \emph{adhesion} of $e$.
We suppress the subscript if the decomposition is clear from the context.

For $s,t \in V(\tau)$ we say that \emph{$s$ is a descendant of $t$}
or that \emph{$t$ is an ancestor of $s$} if $t$ lies on the unique path from $s$ to the root;
note that a node is both an ancestor and a descendant of itself.
For a node $t$ that is not a root of $\tau$, by $\adh_{\tau,\chi}(t)$ we mean
the adhesion $\adh_{\tau,\chi}(e)$ for the edge $e$ connecting $t$ with its parent in $\tau$. 
We extend this notation to $\adh_{\tau,\chi}(r) = \emptyset$ for the root $r$.
Again, we omit the subscript if the decomposition is clear from the context.  For brevity, for $t\in V(\tau)$, we use $A_t$ to denote $\adh_{\tau,\chi}(t)$.

We define the following functions for convenience:
\begin{align*}
\gamma(t) &= \bigcup_{q\textrm{\ descendant\ of\ }t} \chi(q), \\
\alpha(t) &= \gamma(t) \setminus A_t, \\
G_t &= G[\gamma(t)]\\
%G_t &= G[\gamma(t)] - E(G[\adh(t)])
\end{align*}
We say that a rooted tree decomposition $(\tau,\chi)$ of $G$ is \emph{compact}
if for every node $t \in V(\tau)$ for which $\adh_t \neq \emptyset$ we have
that $G[\alpha(t)]$ is connected and $N_G(\alpha(t)) = A_t$. We can extend the function $\chi$ to subsets of $V(\tau)$ in the natural way: for a subset $X \subseteq V(\tau)$, $\chi(X) = \bigcup_{x \in X} \chi(x)$. 
%The adhesion of a node $t\in V(\tau)$ that is not the root is the intersection of $\chi(t)$ and $\chi(t')$ where $t'$ is the parent of $t$ in $\tau$ and is denoted by $A_t$. If $t$ is the root, then the adhesion $A_t = \emptyset$.
%For a rooted tree decomposition $(\tau, \chi)$ of $G$ and 
%\sout{For each node $t \in V(\tau)$ we define $G_t = G[\chi(V(\tau_t))]$ and denote }
By $\chd(t)$, we denote the set of children of $t$ in 
$\tau$. For each subset $X$ of nodes in $V(\tau)$, we define $G_X = G[\cup_{t\in X}\chi(V(\tau_t))]$.

%A \emph{tree decomposition} of a graph $G$ is a pair $(T,\beta)$ where $T$ is a tree and $\beta$ is a mapping that assigns to every $t \in V(T)$
%a set $\beta(t) \subseteq V(G)$, called a \emph{bag}, such that the following holds: (i) for every $e \in E(G)$ there exists $t \in V(T)$ with $e \subseteq \beta(t)$, and (ii)
%for every $v \in V(G)$ the set $\beta^{-1}(v) := \{t \in V(T) : v \in \beta(t)\}$ induces a connected nonempty subgraph of $T$.

%%%%%%%%%%%%%%%%%%%%%%%%%%%%%%%%%%%%%%%%%%%%%%%

\subsection{Splitters and Chernoff Bound}
Let $[n]$ denote the set of integers $\{1,\ldots,n\}$. 
We start by defining the notion of splitters. 
\begin{definition}[\cite{DBLP:conf/focs/NaorSS95}]
\label{splitterdefn}
An $(n,k,\ell)$ splitter $\mathcal{F}$ is a family of functions from $[n]\rightarrow[\ell]$ such that for all  $S\subseteq [n]$, $|S|=k$, there exists a fucntion $ f\in \cal F$ that splits $S$ evenly. That is, for all $1\leq j,j'\leq \ell$, $|f^{-1}(j')\cap S|$ and $|f'^{-1}(j)\cap S|$ differ by at most one.
\end{definition}

We will need following algorithm to compute splitters with desired parameters. 
\begin{theorem}[\cite{DBLP:conf/focs/NaorSS95}]\label{splitter2}
For all $\ n,k\geq 1$ one can construct an $(n,k,k^2)$ splitter family of size $k^{\cO(1)}\log n$ in time $k^{\cO(1)}n\log n$.
\end{theorem}  

The next lemma is a simple application of Theorem~\ref{splitter2} and is used as a subroutine in our new algorithm for \mkk. 
\begin{lemma}
\label{splittersetlemma}
There exists an algorithm that takes as input a set $S$, two positive integers $s_1$ and $s_2$ that are less than $|S|$, and outputs a family $\mathcal{S}$ of subsets of $S$ having size $\cO((s_1+s_2)^{\cO(s_1)}\log |S|)$ such that for any two disjoint subsets $X_1$ and $X_2$ of $S$ of size at most $s_1$ and $s_2$, $\mathcal{S}$ contains a subset $X$ that satisfies $X_1\subseteq X$ and $X_2\cap X =\emptyset$ in time $\cO((s_1+s_2)^{\cO(s_1)}|S|^{\cO(1)})$.
\end{lemma}

\begin{proof}
The algorithm carries out the following steps.
Order the elements in $S$ arbitrarily, let $v_j$ denote the $j^{th}$ element in this ordering. Initialize $\mathcal{S}=\emptyset$. For each pair $s_1'$ and $s_2'$ such that $s_1'\leq s_1$ and $s_2'\leq s_2$ repeat the following procedure. 
\begin{enumerate}\setlength\itemsep{-.7mm}
\item[(i)] Let $s'=s_1'+s_2'$, construct a $(|S|,s',s'^2)$ splitter $\mathcal{F}$.
\item [(ii)] Construct a set $\mathcal{L}$ of subsets of elements of $S$ by the following procedure. For each function $f$ in $\mathcal{F}$ and for each set $X \subseteq [s'^2]$ having size $s_1'$ construct a subset $C$ of $S$ that contains an element $s_j\in S$ if $f(j)\in X$, that is $C=\{s_j: f(j)\in X\}$ and add it to $\mathcal{L}$.
\item [(iii)]$\mathcal{S}=\mathcal{S}\cup \mathcal{L}$.
Finally output $\mathcal{S}$
\end{enumerate}

Consider two disjoint subsets $X_1,X_2$ of $S$ of size at most $s_1$ and $s_2$ and the splitter $\mathcal{F}$ constructed for the values of $s_1'=|X_1|$ and $s_2'=|X_2|$. By Definition~\ref{splitterdefn}, for all $X'\subseteq [|S|]$ of size $s'=s_1'+s_2'$, there exists a function $f\in \mathcal{F}_1$ that maps each $x\in X'$ to a distinct integer in $[s'^2]$. Thus, there is a function $f\in \mathcal{F}$, that maps each $j$ to a distinct integer in $[s'^2]$ for each $v_j\in X_1\cup X_2$. Since the size of $X_1$ is $s_1'$, we can infer that there exists a function $f\in \mathcal{F}$ and a set $X\subseteq [s'^2]$ of size $s_1'$ such that $X=\{f(j):v_j\in X_1\}$. Therefore, by the construction of $\mathcal{L}$, there exists a set in $C$ in $\mathcal{L}$ and hence in $\mathcal{S}$ such that $X_1\subseteq C$ and $X_2\cap C=\emptyset$.  

For each $s_1'$ and $s_2'$ such that $s_1'\leq s_1$ and $s_2'\leq s_2$, it follows from Theorem~\ref{splitter2} that the size of $\mathcal{F}$ is $s'^{\cO(1)}\log |S|$. From the construction of $\mathcal{L}$ it is easy to see that its size is $|\mathcal{F}|{{s'^2}\choose{s_1'}}$. 

Thus, $\mathcal{S}$ is of size $(s_1s_2)(s_1+s_2)^{\cO(1)}\log |S|
{{(s_1+s_2)^2}\choose{s_1}}=\cO((s_1+s_2)^{\cO(s_1)}\log |S|)$ and can be computed in time $\cO((s_1+s_2)^{\cO(s_1)}|S|^{\cO(1)})$. This concludes the proof. 
\end{proof}

To prove the correctness of our sparsification procedure we need standard Chernoff bounds to bound the tail probability. 

\begin{lemma}[Chernoff Bound,\cite{DBLP:books/daglib/0012859}]
\label{lem:chernoff}
Let $X_1, \ldots, X_n$ be independent Poisson trials such that $\Pr[X_i]=p_i$. Let $X=\sum_{i=1}^n X_i$ and $\mu=E[X]$. For  
$0<\delta <1 $. 
\[\Pr[|X-\mu|\geq \delta\mu ] \leq 2\cdot e^{-\mu \delta^2/3}.\]
\end{lemma}
%\begin{itemize}
%\item define $\cc(G)$
%\end{itemize}

%!TEX root = main.tex
\section{Reduction to Instances with Logarithmic Solution Size}
\label{section:sparsification}
%via Graph Sparsification}
%\todo{In proof of final theorem, we assume the graph is unweighted multi graph - there write the redn from  wt to unwt informally}
In this section we give a polynomial time algorithm that  given an unweighted graph $G$, an integer $k$ and an 
$\epsilon >0$, outputs a subgraph $G'$ of $G$, such that $V(G)=V(G')$ and a real number $r$ such that with high probability it holds that for {\em every} partition $\tilde{P}$ of $V(G) = V(G')$ into $k$ parts, we have that 
$w(G,\tilde{P})\simeq r \cdot w(G',\tilde{P})$. Furthermore, every  non-trivial $2$-cut  of $G_1$ has weight at least $\frac{\epsilon \cdot  \opt(G,k)}{k-1}$. Recall, that a non-trivial cut means that there is at least one edge crossing the cut. 
We start by a greedy procedure that will ensure that weight of each $2$-cut of the sparsifier is sufficiently large.

\begin{lemma}\label{RandLemma1}
There exists an algorithm that takes as input  an integer $k>0$, an unweighted graph $G$ such that $\cc(G)<k$, 
%having strictly 
%less than $k$ connected components,
%There exists an algorithm that takes as input an unweighted graph $G$ with less than $k$ components , an integer $k>0$, 
and a real number 
$0< \epsilon \leq 1 $ and outputs a subgraph $G_1$ of $G$ with $V(G')=V(G)$, such that 
$|E(G)|-|E(G_1)|\leq 2\epsilon\cdot \opt(G,k)$. Further,  if $\cc(G_1)<k$,  then
 each non-trivial $2$-cut 
 %$C$ 
 of $G_1$ has weight at least $\frac{\epsilon \cdot  \opt(G,k)}{k-1}$.
 %$w(G_1,C)\geq \frac{\epsilon \cdot  \opt(G,k)}{k}$.  
\end{lemma}
\begin{proof}

The graph $G_1$ is essentially obtained by a greedy algorithm that removes all $2$-cuts of small weights. In particular,  
we obtain the graph $G_1$ by the following procedure. Obtain a $2$-approximate $k$-cut of $G$ using any of the known approximation algorithm~\cite{NaorR01,ravi2008approximating,SaranV95}, let $w_a$ be the weight of this cut.  Initialize $G'=G$ and as long as the graph has a non-trivial $2$-cut  of weight at most $\frac{\epsilon \cdot  w_a}{k-1}$ and $\cc(G')<k$, we remove all edges of this $2$-cut from $G'$. Return the resulting graph $G'$ as $G_1$.

%Observe that the number of times the above procedure will be repeated is less than $k-1$ times, since if it is
%repeated $k-1$ times we have cut the graph into at least $k$ connected components using at most  
%$ \frac{\epsilon \cdot  w_a}{k-1} \leq $
%%
%
%%(i) Obtain a $2$-approximate $k$-cut of $G$ using any of the known approximation algorithm~\cite{NaorR01,ravi2008approximating,SaranV95}, let $w_a$ be the weight of this cut. 
%%
%%
%%(ii) Initialize $G'=G$. (ii) Let $C$ be a minimum weight $2$-cut among all basic $2$-cuts of $G'$ and $E'$ be the set of edges in this cut. If $w(G',C)> \frac{\epsilon \cdot w_{a}}{k}$ or if $G'$ has $k$ connected components, go to step (iii). Otherwise, modify $E(G')=E(G')\backslash E'$ and repeat step (ii). (iii) Return $G'$ as $G_1$.

Observe that the number of times the above procedure will be repeated is at most $k-1$ times,  since if it is
repeated $k-1$ times we have cut the graph into at least $k$ connected components. Furthermore, in each iteration 
in which $E(G')$ is modified, the size of the set of edges $E'$ removed from $G'$ is at most 
$\frac{\epsilon\cdot w_{a}}{k-1}$.  
%In each iteration of step (ii) in which $E(G')$ is modified, the size of the set of edges $E'$ removed from $G'$ is at most $\frac{\epsilon\cdot w_{a}}{k}$ and step (ii) is repeated at most $k$-times. 
Thus $|E(G)|-|E(G_1)|$, the total number of edges removed from $G$ to obtain $G_1$ is less than or equal to $\epsilon\cdot w_{a}$. Since $w_{a}\leq 2\cdot \opt(G,k)$, it follows that $|E(G)|-|E(G_1)|\leq 2\epsilon \cdot \opt(G,k)$.  Also, if $\cc(G_1)<k$, then the  weight of a minimum non-trivial $2$-cut  of $G_1$ is more than $\frac{\epsilon\cdot w_{a}}{k-1}$. Since $w_{a}\geq \opt(G,k)$, it follows that $\frac{\epsilon\cdot w_{a}}{k-1} \geq \frac{\epsilon\cdot \opt(G,k)}{k-1}$. Thus, if $G_1$ has less than $k$ connected components, each non-trivial $2$-cut of $G_1$ has weight more than 
$\frac{\epsilon\cdot \opt(G,k)}{k-1}$. Since, we can find a non-trivial minimum $2$-cut in polynomial time, we have that the algorithm runs in polynomial time. 
This concludes the proof. 
%$w(G_1,C)\geq \frac{\epsilon\cdot OPT(G,k)}{k}$.
\end{proof}

Next given an unweighted graph $G$, an integer $k$ and an $\epsilon >0$, we give a sparsification procedure that outputs a subgraph $G'$ of $G$, and a number $p$ such that with high probability it holds that for {\em every} partition $\tilde{P}$ of $V(G) = V(G')$ into $k$ parts, the number of edges of $G'$ crossing the partition $\tilde{P}$ is within a $1 \pm \epsilon$ factor of $p$ times the number of edges crossing it in $G$. This procedure is very similar to that of  known randomized cut sparsifiers~\cite{BenczurK15,Karger94}).
% Thus, for purposes of $(1+\epsilon)$-approximation to \mkk we well find a $(1+\epsilon)$-approximate solution in $G'$. 

\begin{lemma}\label{RandLemma2}
There exists a polynomial time algorithm that takes as input  an integer $k>0$, an unweighted graph $G$ with $\cc(G)<k$, 
%having less than $k$ connected components,
 and a real number $\frac{1}{n}<\epsilon \leq 1$ and outputs a subgraph  $G_2$ with $V(G_2)=V(G)$, and a real number $r$  
 %having the same set of vertices as $G$ 
 such that with probability at least $\prup$, for all $k$-cuts $\tilde{P}$ in $G$, 
 $(1-\epsilon)\cdot w(G,\tilde{P})\leq w(G_2,\tilde{P})\cdot r \leq (1+\epsilon)\cdot w(G_2,\tilde{P})$. 
\end{lemma}
\begin{proof}
We obtain the graph $G_2$ by the following procedure: (i) Initialize $G'=G$. (ii) Fix ${p=\frac{100 \cdot \log n}{\epsilon^2\cdot \opt(G,2)}}$, where $\opt(G,2)$ denotes the weight of a non-trivial 
minimum $2$-cut of $G$. (iii) Keep each edge $e$ in $G'$ with probability $p$ and remove with probability $1-p$. (iv) Return $G'$ as $G_2$ and $r=1/p$.

\begin{claim}
\label{claim:singledeviate}
For every non-trivial $2$-cut $C$ in $G$, $|w(G_2,C)-w(G,C)\cdot p|\geq \epsilon\cdot w(G,C)\cdot p$ with probability at most $2\cdot n^{\frac{-100\cdot w(G,C)}{3\cdot \opt(G,2)}}$.
\end{claim}
\begin{proof}
Let us associate a random variable $X_e$ to each edge $e\in G$, $X_e=1$ if the edge $e$ is in graph $G_2$ and $0$ otherwise. The random variable $X_e = 1$ with probability $p$ and $X_e=0$ with probability $1-p$. Let $E'$ be the set of edges in the cut $C$ in $G$ and $X=\underset{e\in E'}{\sum}X_e.$
The expectation of $X$ (denoted by $\mu$), is $\underset{e\in E'}{\sum}p=w(G,C)\cdot p$.  
Applying Chernoff bound~(Lemma~\ref{lem:chernoff}) on $X$, we get 
%$\Pr[|X-\mu|\geq \epsilon\mu ] \leq 2\cdot e^{-\mu \epsilon^2/3}$. Thus, 
%we have
%\begin{eqnarray*}
%\Pr[|X-\mu|\geq \epsilon\mu ]&\leq 2\cdot e^{-\mu \epsilon^2/3}\\
$$\Pr[|w(G_2,C)-w(G,C)\cdot p]\geq \epsilon\cdot w(G,C)\cdot p|  \leq   2\cdot e^{-{\frac{100 \epsilon^2\cdot \log n \cdot  w(G,C)}{3\epsilon^2\cdot \opt(G,2)}}}
 \leq  2\cdot n^{\frac{-100\cdot w(G,C)}{3\cdot \opt(G,2)}}. $$
%\end{eqnarray*}
%   P[|w(G_2,C)-w(G_1,C)\cdot p]\geq \epsilon\cdot w(G_1,C)\cdot p] 
\end{proof}

Next we would like to simultaneously apply Claim~\ref{claim:singledeviate} over all non-trivial $2$-cuts in $G$. Towards this we will first bound the number of non-trivial $2$-cuts in $G$
\begin{claim}
The probability that for all non-trivial $2$-cuts $C$ in $G$, $|w(G_2,C)-w(G,C)\cdot p|\leq \epsilon\cdot w(G,C)\cdot p$ is at least $\prup$.
\end{claim}
\begin{proof}
First we show that the number of non-trivial $2$-cuts in $G$ of weight $i$, where $\opt(G,2)\leq i\leq |E(G)|$ is at most $n^{2i/\opt(G,2)}$.  For any $\alpha \geq 1$, a cut is called an {\em $\alpha$-minimum cut} if its number of edges is at most $\alpha$ times larger than a minimum cut. In any undirected graph, and for any real number $\alpha \geq 1$,  the number of $\alpha$-minimum cuts is at most $n^{2\alpha}$~\cite{Karger94}. By choosing $\alpha= i/\opt(G,2)$, we get that the number of non-trivial $2$-cuts in $G$ of weight $i$ is at most $n^{2i/\opt(G,2)}$.

Now we use union bound to show that the probability of all non-trivial $2$-cuts $C$ of $G$ satisfying $|w(G_2,C)-w(G,C)\cdot p|\leq \epsilon\cdot w(G,C)\cdot p$ is at least $\prup$. Let $A$ denote the event  that some non-trivial $2$-cut $C$ in $G$ does not satisfy the required inequality. An upper bound on $\Pr[A]$
 %The probability $Pr[bad]$ that there is some basic $2$-cut $C$ in $G_1$ does not satisfy the required inequality 
can be obtained by combining the Claim~\ref{claim:singledeviate}  along with the bound on the number of cuts of weight $i$ as follows. However, recall that the number of edges in $G$, say $m$, is upper bounded by  $m_{\sf orig}^2/\epsilon$, where $m_{\sf orig}$ is the number of edges in the input graph (before we applied the Knapsack trick to make it an  unweighted graph). This implies that $m\leq m_{\sf orig}^2/\epsilon \leq n^5$. 
%\begin{eqnarray*}
%\Pr[A] &\leq  &  \displaystyle\sum_{i=\opt(G_1,2)}^{m} 2 \cdot n^{\frac{-100\cdot w(G_1,C)}{3\cdot \opt(G_1,2)}} \cdot n^{\frac{2i}{\opt(G_1,2)}} =  \displaystyle\sum_{i=\opt(G_1,2)}^{m} 2 \cdot n^{\frac{-100\cdot w(G_1,C)}{3\cdot \opt(G_1,2)}} \cdot n^{\frac{2i}{\opt(G_1,2)}} =  \\
% &= & \displaystyle\sum_{i=\opt(G_1,2)}^{m} 2\cdot n^{-32i/OPT(G_1,2)}\\
% &\leq  &1/n^{30}.
%\end{eqnarray*}

\begin{eqnarray*}
\Pr[A] &\leq  &  \displaystyle\sum_{i=\opt(G,2)}^{m} 2 \cdot n^{\frac{-100 i}{3\cdot \opt(G,2)}} \cdot n^{\frac{2i}{\opt(G,2)}}  \\
 &\leq & \displaystyle\sum_{i=\opt(G,2)}^{m} 2\cdot n^{-31i/\opt(G,2)}\\
 &\leq & \frac{1}{n^{31}} \cdot  \frac{n^4}{\epsilon} \\
 & \leq & \pru.
\end{eqnarray*}

%\begin{align*}
%P[bad] &= \displaystyle\sum_{i=OPT(G_1,2)}^{i=m} 2 \cdot n^{\frac{-100\cdot w(G_1,C)}{3\cdot OPT(G_1,2)}} \cdot n^{i/OPT(G_1,2)} \\
%P[bad] &= \displaystyle\sum_{i=OPT(G_1,2)}^{i=m} 2\cdot n^{-32i/OPT(G_1,2)}\\
%P[bad] &\leq 1/n^{30}.
%\end{align*} 
\end{proof}
\begin{claim}
The probability that for all $k$-cuts $\tilde{P}$ in $G$, $(1-\epsilon)\cdot w(G,\tilde{P})\leq w(G_2,\tilde{P})\cdot r \leq (1+\epsilon)\cdot w(G,\tilde{P})$ is at least $\prup$. 
\end{claim}
\begin{proof}
Let $\tilde{P}=\{P_1,\ldots ,P_k\}$ be a $k$-cut in $G$, each part $P_i$ is in some connected component $X$ of $G$. Let $\cal C$ contain all the non-trivial $2$-cuts $C_i=(P_i,V(G_1)\backslash P_i)$ in $G$. 
 %let $C_i$ be the non-trivial $2$-cut $(P_i,X\backslash P_i)$ in $G$. 
 The weight of $\tilde{P}$ in $G$ is given by,
\begin{align*}
w(G,\tilde{P}) &= \frac{\displaystyle{\sum_{P_i\in \tilde{P}} w(G,(P_i,V(G_1\backslash P_i))}}{2} \\
 & = \frac{\displaystyle{\sum_{C_i\in {\cal C}} w(G,C_i)}}{2}
\end{align*}
Thus, using the previous claim, with probability at least $\prup$,
\begin{align*}
 |w(G_2,C_k)-w(G,C_k)\cdot p| &= \frac{ \bigg | \displaystyle\sum_{C_i\in {\cal C}}w(G_2,C_i) -  \displaystyle\sum_{C_i\in {\cal C}} w(G,C_i) \bigg | }{2} \\
 &= \frac{ \bigg | \displaystyle\sum_{C_i\in {\cal C}}
 w(G_2,C_i) - w(G,C_i) \bigg | }{2} \\
 &= \frac{ \bigg | \displaystyle\sum_{C_i\in {\cal C}}\epsilon\cdot w(G,C_i)\cdot p \bigg | }{2}  \\
 &=  \frac{ \epsilon\cdot p \cdot \bigg |\displaystyle\sum_{C_i\in {\cal C}} w(G,C_i)\bigg | }{2}  \\ 
 &= \epsilon\cdot w(G,\tilde{P})\cdot p
\end{align*} 
\end{proof}
The previous claim concludes the proof of the lemma.
\end{proof}
%We now use Claim~\ref{Randlemma1} and \ref{RandLemma2} to prove Theorem~\ref{Randalg}.
%\begin{proof}{Theorem~\ref{Randalg}:}
%We use Lemma~\ref{Randlemma1} to obtain graph $G_1$ from $G,k,\epsilon$. If $G_1$ has $k$ connected components, then we return $G_2
%$ as $G'$ and $r=1$. Otherwise use Lemma~\ref{RandLemma2} to obtain graph $G_2$ and $r$ from $G_1,k,\epsilon$ and return $G_2$ as $G'$ and $r$
%.
%
%Combining Lemma~\ref{Randlemma1} and Lemma~\ref{Randlemma1}, it is easy to see that $G'$ satisfies the properties claimed in Theorem~\ref{Randalg}. 
%\end{proof}
%%%%%%%%%%%%%%%%%%%%%%%%%%%%%%%%%%%%%%%%%%%%%%%%%%%%%%%%%%%%%%%%%%%%%%%%%%%%%%%%%%%%%%%%%%%%%%%%%%%%%%%%%%%

%!TEX root = main.tex

\section{Edge Unbreakable Decomposition Theorem in Polynomial Time}
\label{section:decomposition}
In this section we give our main decomposition theorem. That is, we give a tree decomposition where each bag is edge-unbreakable. In particular we will prove the following theorem.

\begin{theorem}\label{thm:decomp}\label{thm:decomposition}
Given an $n$-vertex graph $G$ and an integer $s$, one can in polynomial time compute a rooted compact tree decomposition $(\tau,\chi)$ of $G$ such that 
\begin{enumerate}
\setlength{\itemsep}{-2pt}
\item every adhesion of $(\tau,\chi)$ is of size at most $s$;
\item every bag of $(\tau,\chi)$ is $((s+1)^5,s)$-edge-unbreakable.
\end{enumerate}
\end{theorem}

The proof of Theorem~\ref{thm:decomp} closely follows the proof of Theorem~$1.2$ of Cygan et al.~\cite{DBLP:journals/corr/abs-1810-06864} with a few differences. The overall template of iterated refinement is {\em exactly} the same. 

The main difference is found in Section~\ref{subsec:flw} where we state and prove Lemma~\ref{lem:unbreakableTesting}. In Cygan et al.~\cite{DBLP:journals/corr/abs-1810-06864} the corresponding lemma has vertex separators instead of edge cuts,  condition $1$ has $|A \cap Q| > s$ and $|B \cap Q| > s$ instead of  $|A \cap Q| > (s+1)^5$ and $|B \cap Q| > (s+1)^5$ and their algorithm runs in $s^{\cO(s)}n^{\cO(1)}$ time instead of polynomial time.  The proofs of Lemma~\ref{lem:unbreakableTesting} and Lemma~$3.4$ of Cygan et al.~\cite{DBLP:journals/corr/abs-1810-06864} are entirely different. The remaining parts of the proof are so similar that we {\bf have copied their text verbatim and incorporated minor modifications to suit our need}. In particular {\em all of Subsections~\ref{sec:refine} and~\ref{sec:decoAlgo} are included solely for the convenience of the reader and ease of understanding}.

\subsection{Improving Tree Decomposition given a Lean Witness}
\label{sec:refine}
Our starting point for the proof of Theorem~\ref{thm:decomp} is the definition of a lean tree decomposition of Thomas~\cite{Thomas90}; we follow the formulation of~\cite{BellenbaumD02}.
\begin{definition}
A tree decomposition $(\tau,\chi)$ of a graph $G$ is called \emph{lean} if for every $t_1,t_2 \in V(\tau)$
and all sets $Z_1 \subseteq \chi(t_1)$ and $Z_2 \subseteq \chi(t_2)$ with $|Z_1| = |Z_2|$, either
$G$ contains $|Z_1|$ vertex-disjoint $Z_1-Z_2$ paths, or there exists an edge $e \in E(\tau)$ on the path from $t_1$ to $t_2$
such that $|\adh(e)| < |Z_1|$.
\end{definition}

For a graph $G$ and a tree decomposition $(\tau,\chi)$ that is not lean, a quadruple
$(t_1,t_2,Z_1,Z_2)$ for which the above assertion is not true is called a \emph{lean witness}.
Note that it may happen that $t_1 = t_2$ or $Z_1 \cap Z_2 \neq \emptyset$.
In particular $(t,t,Z_1,Z_2)$ is called a \emph{single bag lean witness}.
The \emph{order} of a lean witness is the minimum order of a separation $(X_1,X_2)$ such that $Z_i \subseteq X_i$ for $i=1,2$. Observe that by Menger's theorem the following conditions are equivalent.

%Bellenbaum and Diestel~\cite{BellenbaumD02} defined an improvement step, that, given a tree
%decomposition and a lean witness, refines the decomposition so that it is in some sense
%closer to being lean. We will use the same refinement step, but only for the special single bag case, and thus in subsequent sections
%focus on finding a lean witness in a current candidate tree decomposition. 

\begin{claim}
\label{claim:equiv}
For a tree decomposition $(\tau,\chi)$, a node $t \in V(\tau)$, and subsets $Z_1,Z_2 \subseteq \chi(t)$,
either both or none of the following two conditions are true:
\begin{itemize}
  \item $(t,t,Z_1,Z_2)$ is a single bag lean witness for $(\tau,\chi)$,
  \item there is a separation $(X_1,X_2)$ of $G$ with $Z_i \subseteq X_i$, where $X = X_1 \cap X_2$, such that $|X| < |Z_1| = |Z_2|$,
  and there is a set of vertex disjoint $Z_1-Z_2$ paths $\{P_x\}_{x \in X}$ 
  such that $x \in V(P_x)$ for every $x \in X$.
\end{itemize}
Moreover given a single bag lean witness $(t,t,Z_1,Z_2)$ one can find the above separation $(X_1,X_2)$ and set of paths $\{P_x\}_{x \in X}$ in polynomial time.
\end{claim}
A minimum order of a separation from the second point is called the {\em{order}}
of the single bag lean witness $(t,t,Z_1,Z_2)$.
To argue that the refinement process stops after a small number of steps, or that it stops
at all, we define the following potential for a graph $G$, a tree decomposition $(\tau,\chi)$, and an integer $s$:
$$\Pot_{G,s}(\tau, \chi) = \sum_{t \in \tau} \max(0, |\chi(t)|-2s-1).$$
%Note that the potential $\Pot$ is different than the one used in~\cite{BellenbaumD02}, as
%the one used in~\cite{BellenbaumD02} can be exponential in $n$ while being oblivious to the cut
%size $s$.
%
%Given a witness, a single refinement step we use is encapsulated in the following lemma, which is essentially a repetition
%of the refinement process of~\cite{BellenbaumD02} with the analysis of the new potential.
%We emphasize that in this part, all considered tree decompositions are unrooted.

\begin{lemma}\label{lem:vertex-refine}
Assume we are given a graph $G$, an integer $s$, a tree decomposition $(\tau, \chi)$ of $G$
with every adhesion of size at most $s$, one node $q \in \tau$ with $|\chi(q)| > 2s+1$, and
a single bag lean witness $(q,q,Z_1,Z_2)$ with $|Z_1| \le s+1$.
Then one can in polynomial time compute a tree decomposition $(\tau',\chi')$ of $G$
with every adhesion of size at most $s$
such that $\Pot_{G,s}(\tau,\chi) > \Pot_{G,s}(\tau',\chi')$.
\end{lemma}
\begin{proof}
Apply Claim~\ref{claim:equiv}, yielding a separation $(X_1,X_2)$ and a family $\{P_x\}_{x \in X}$ of vertex disjoint $Z_1-Z_2$ paths, where $X=X_1\cap X_2$. 
Note that $s+1 \ge |Z_1| > |X|$, hence the order of $(X_1,X_2)$ is at most $s$.

We construct a tree decomposition $(\tau',\chi')$ as follows.
First for every $i = 1,2$, we construct a decomposition $(\tau^i,\chi^i)$ of $G[X_i]$:
we start with $\tau^i$ being a copy of $\tau$, where a node $t \in V(\tau)$ corresponds
to a node $t^i \in V(\tau^i)$, and $\chi^i(t^i) := \chi(t) \cap X_i$ for every $t \in V(\tau)$.
Then for every $x \in X \setminus \chi(q)$ we take the node $t_x \in V(\tau)$ such that $x \in \chi(t_x)$ and $t_x$
is closest to $q$ among such nodes, and insert $x$ into every bag $\chi^i(t^i)$ for $t^i$
lying on the path between $q^i$ (inclusive) and $t_x^i$ (exclusive) in $\tau^i$.

Clearly, $(\tau^i,\chi^i)$ is a tree decomposition of $G[X_i]$ and $X \subseteq \chi^i(q^i)$.
We construct $(\tau',\chi')$ by taking $\tau'$ to be a disjoint union of $\tau^1$ and $\tau^2$, with
the copies of the node $q$ connected by an edge $q^1q^2$, and $\chi' := \chi^1 \cup \chi^2$.
Since $(X_1,X_2)$ is a separation and $X = X_1 \cap X_2$ is present in both bags $\chi^1(q^1)$, $\chi^2(q^2)$, we infer
that $(\tau',\chi')$ is a tree decomposition of $G$.

We now argue that every adhesion of $(\tau', \chi')$ is of size at most $s$. This is clearly
true for the edge $q^1q^2$ 
connecting $\tau^1$ and $\tau^2$, as the adhesion there is exactly $X$ and $|X| \leq s$.

Consider now a bag $t^i$ in a tree $(\tau^i,\chi^i)$. The set $\chi^i(t^i) \setminus \chi(t)$
consists of some vertices of $X$, namely those vertices $x \in X \setminus \chi(q)$ for which $t$
lies on the path between $q$ (inclusive) and $t_x$ (exclusive).
However, by the properties of a tree decomposition
and Menger's theorem, $\chi(t)$ contains at least one vertex of $P_x$ that lies between 
$x$ and the endpoint in $Z_{3-i}$, that is in $V(P_x) \cap (X_{3-i} \setminus X_i)$.
This vertex is not present in $\chi^i(t^i)$ and, consequently, $|\chi^i(t^i)| \leq |\chi(t)|$.
The same argumentation holds for every edge $e^i \in E(\tau^i)$ and adhesion of this edge.

We are left with analysing the potential decrease.
Fix $t \in V(\tau)$. We analyse the difference between the contribution to the potential
of $t$ in $(\tau,\chi)$ and the copies of $t$ in $(\tau',\chi'$).
First, by the analysis in the previous paragraph, we have $|\chi^i(t^i)| \leq |\chi(t)|$
for $i=1,2$. Consequently, if $|\chi(t)| \leq 2s+1$, then $|\chi^i(t^i)| \leq 2s+1$ for $i=1,2$
and the discussed contributions are all equal to $0$.
Furthermore, if $|\chi^i(t^i)| \leq 2s+1$ for some $i=1,2$, then 
$\max(|\chi(t)|-2s-1,0) \geq \max(|\chi^{3-i}(t^{3-i})|-2s-1,0)$ as
$|\chi^{3-i}(t^{3-i})| \leq |\chi(t)|$.

Otherwise, assume that $|\chi(t)| > 2s+1$ and $|\chi^i(t^i)| > 2s+1$ for $i=1,2$.
Note that $\chi^1(t^1) \cup \chi^2(t^2) \subseteq \chi(t) \cup X$
while $\chi^1(t^1) \cap \chi^2(t^2) \subseteq X$ and $|X| \leq s$.
Consequently, 
\begin{align*}
|\chi(t)|-2s-1 &\geq |\chi(t) \setminus X_2| + |\chi(t) \setminus X_1| - 2s-1 \\
    &\geq |(\chi(t) \setminus X_2) \cup X| + |(\chi(t) \setminus X_1) \cup X| - 4s-1 \\
    &> |\chi^1(t^1)| -2s-1 + |\chi^2(t^2)| -2s-1.
\end{align*}
We infer that for every bag $t \in V(\tau)$, the contribution of $t$ to the potential
$\Pot_{G,s}(\tau,\chi)$ is not smaller than the contribution of the two copies of $t$
in $\Pot_{G,s}(\tau',\chi')$. To prove strict inequality, we show that these contributions
are not equal for the bag $q$.

Recall that we assumed $|\chi(q)| > 2s+1$. By the previous argumentation,
the only chance for equal contributions of $q$ to $\Pot_{G,s}(\tau,\chi)$
and $q^1,q^2$ to $\Pot_{G,s}(\tau',\chi')$ is that $|\chi^i(s^i)| \leq 2s+1$ for some $i=1,2$.
However, note that $\chi^{3-i}(q^{3-i}) = (\chi(q) \setminus X_i) \cup X$. Consequently, as $Z_i \subseteq \chi(q)$ and $|Z_i| > |X|$, we have
$|\chi^{3-i}(q^{3-i})| < |\chi(q)|$ and hence $|\chi(q)|-2s-1 > \max(|\chi^{3-i}(q^{3-i})|-2s-1,0)$. This finishes the proof of the lemma.
\end{proof}

\subsection{Finding a Lean Witness}
\label{subsec:flw}

\begin{lemma}
\label{lem:unbreakableTesting}
There exists a polynomial time algorithm that takes as input an unweighted connected graph $G$, a set of terminals $Q\subseteq V(G)$ and an  integer $s$ and in polynomial time concludes either of the two statements. 
\begin{enumerate}
\setlength{\itemsep}{-3pt}
\item There is no edge cut $(A,B)$ of $G$ of order at most $s$ such that $|A\cap Q|\geq (s+1)^5$ and  $|B\cap Q|\geq (s+1)^5$.
\item Outputs a edge cut $(A,B)$ of $G$ of  order at most $s$  such that $|A\cap Q|\geq s+1$ and  
$|B\cap Q|\geq s+1$. 

\end{enumerate}
\end{lemma}
\begin{proof}
Towards the proof, we will show that if $G$ has a edge cut $(X,Y)$ of order at most $s$ such that $|X\cap Q|\geq (s+1)^5$ and  $|Y\cap Q|\geq (s+1)^5$, then we will always output a edge cut $(A,B)$ of $G$ of  order at most $s$  such that $|A\cap Q|\geq s+1$ and  $|B\cap Q|\geq s+1$.

For rest of the proof we assume that we have a edge cut $(X,Y)$ of order at most $s$ such that $|X\cap Q|\geq (s+1)^5$ and  $|Y\cap Q|\geq (s+1)^5$. Let $E(X,Y)$ be denoted by $S$. 
To design our algorithm in polynomial time we will output a polynomial size family $\cal F$ of {\em connected sets} that has following properties.  
\begin{enumerate}
\setlength{\itemsep}{-3pt}
\item For each set $C\in \cal F$, we have that $|C\cap Q|\geq s+1$. 
%\item Each set $C\in \cal F$  is either in $G[X]$ or $G[Y]$. That is each set is in $G-E(X,Y)$. 
\item There exists a set $C_1\in \cal F$ that belongs to $G[X]$ and there exists a set $C_2\in \cal F$ that belongs to $G[Y]$. 
\end{enumerate} 

Note that if we have $\cal F$ of polynomial size satisfying the above properties then the algorithm just iterates over pairs $(C_1,C_2)$ and computes a min-cut between $C_1$ and $C_2$ in polynomial time~\cite{DBLP:books/daglib/0015106}. If this cut is of size at most $s$, then we can get a edge cut $(A,B)$ of order at most $s$ such that $C_1 \subseteq A$ and 
$C_2\subseteq B$. Thus, we  have that $|A\cap Q|\geq s+1$ and  $|B\cap Q|\geq s+1$, concluding the proof. Next, we will show that if we do have a edge cut $(X,Y)$ of order at most $s$ such that $|X\cap Q|\geq (s+1)^5$ and  $|Y\cap Q|\geq (s+1)^5$, then we will always find such a pair of $(C_1, C_2)$. 

Now our job  reduces to finding $\cal F$ of polynomial size satisfying the above properties. To achieve this we start with a spanning tree $T$ of $G$ and use this to get the family $\cal F$. Let $T$ be rooted at some vertex. This immediately gives parent, child, ancestor, descendant relations among the vertices of the tree.  We add the following connected sets  to the  family ${\cal F}$. 
\begin{itemize}
\setlength{\itemsep}{-3pt}
\item For every vertex $v$ of the tree, all the vertices in the rooted subtree $T_v$, if $|V(T_v)\cap Q|\geq s+1$.
\item For every pair of vertices $(u,v)$ in the tree where one is descendant of other, the path between them $P_{uv}$, if 
$|V(P_{uv})\cap Q|\geq s+1$.
\item Let $P_{uv}$ be the path where one is descendant of other. Let $W$ denote the set of children of $P_{uv}$. That is all those vertices whose parents are on $P_{uv}$. We say $w\in W$  is {\em a good child} if $|V(T_w)\cap Q|\geq 1$. 
Let $W_{\sf good}$ be the subset of $W$ containing all good children. Let $|W_{\sf good}|\geq (s+1)^2$. We partition  
$|W_{\sf good}|$ into $(s+1)$ parts, say $W_1,\ldots, W_{s+1}$ such that each part contains at least $s+1$ vertices and 
no two sets have any element in common. Now, for all $1\leq i\leq s+1$, add the set containing the path $P_{uv}$ and  subtree rooted at vertices in $W_i$. We call these sets $H_{1},\ldots, H_{s+1}$. 
%Furthermore, there are at least $(s+1)^2$ children of 
%
% and the number of edges leaving this path $P_{uv}$  is at least $(s+1)^2$,  
%we have $s+1$ sets containing the path and $s+1$ subtrees rooted at vertices that are children of the path. 
\end{itemize}
By our construction each set in $\cal F$ is connected, contains at least $s+1$ terminals, and has size polynomial in $n$. 

%Let  $\cal F$ be the set of enumerated vertex sets that contain at least $s+1$ terminals. 
%
%and for every pair of vertices where one is dencendant of other and the number of edges out is >= (s+1)^2 we have s+1 sets containing the path and s+1 subtrees rooted at vertices that are children of the pa

%Let $C_x$ and $C_
%This will prove the result. To achieve this, we will make a family of connected sets of 
%
%so what algorithm does is as follows it picks a spanning tree, and then uses the spanning tree to compute poly(n)
% sets as follows:

%for every vertex of the tree, all the vertices in the rooted subtree
%
%for every pair of vertices in the tree where one is descendant of other, the path between them
%
%and for every pair of vertices where one is dencendant of other and the number of edges out is >= (s+1)^2 we have s+1 sets containing the path and s+1 subtrees rooted at vertices that are children of the path
%
%so let F be the set of enumerated vertex sets that contain >= s+1 terminals

To prove that the $\cal F$ has the desired property, all that remains to show is that there exists a set $C_1\in \cal F$ that belongs to $G[X]$ and there exists a set $C_2\in \cal F$  that belongs to $G[Y]$. We will only show that there exists a set $C_1\in \cal F$ that belongs to $G[X]$. The proof for the other case is analogous.   
%Let $C$ denote  connected components contained in $X$ and $Y$, respectively, that have maximum number of terminals.  In other words, look at the graph $G-E(X,Y)$ and pick connected components $C_x$ in $G[X]$ and $C_y$ in $G[Y]$, whose intersection with $T$ is maximized . 

Let $C$ be a connected component in $G[X]$ (in $T-S$)  with largest intersection with $Q$. Observe that $C$ has at least $(s+1)^4$ terminals as deleting at most $s$ edges can only make $s+1$ connected components. 
Since, $T$ is a a rooted tree,  $C$ itself is a rooted tree, let $r$ be the root of $C$. For every vertex $v \in C$ incident to a cut-edge (edge in $S$), mark the path from $v$ to $r$ in $C$). Let $M$ be the set of marked vertices in $C$. Observe that $M$ is a {\em sub-tree} of $C$. For every $u\in M$, define the weight of $u$ to be the number of terminals in the connected component of 
$C -(M -\{u\})$ that contains $u$. Note that  the sum of weights in $M$ is equal to number of terminal  in $C$. 
Let   $P$ be the root-leaf path in $M$ of largest weight. Since, $|S|\leq s$, we have that the weight of $P$ ($w(P)$) is at least $\frac{w(C)}{s}\geq \frac{(s+1)^4}{s}$. Here, $w(C)$ denotes the number of terminals in $C$. If $P$ contains at least $s+1$ terminals, then clearly $V(P)\in \cal F$. Thus, assume that there are at most $s$ terminals on $P$.  Since,  $w(P)$ is at least $\frac{(s+1)^4}{s}$ we have that either there is a child $u$ of a vertex in $P$, such that $T_u$ has at least $s+1$  terminals or $|W_{\sf good}|\geq (s+1)^2$. First let us handle the case when  $|W_{\sf good}|\geq (s+1)^2$. In this case, look at the sets $H_{1},\ldots, H_{s+1}$ that we added to $\cal F$. We know that any edge in $S$ can touch at most one set among  $H_{1},\ldots, H_{s+1}$. This implies that there exists a set $H_i$ that has no edge from $S$ incident to it and thus, 
 it is completely contained in $G[X]$. Indeed, it also contains at least $s+1$ terminals.

 Now we know that $|W_{\sf good}|< (s+1)^2$ and $P$ has at most $s$ terminals. Let $u$ be the vertex of $P$. 
 Then, the weight of a vertex $u$ on $P$ is given by the following. 
 
\begin{displaymath} 
    w(u) =
\left\{
	\begin{array}{ll}
		1+\sum\limits_{\mbox{$ux$ and $x$ is not marked}} |V(T_x)\cap Q|  & \mbox{if $u$ is a terminal}  \\ \\
		  \sum\limits_{\mbox{$ux$ and $x$ is not marked}} |V(T_x)\cap Q|     & \mbox{otherwise} 
	\end{array}
\right.
\end{displaymath}

Observe that if a vertex $w$ is not {\em marked} then the subtree $T_w$ is not incident to any edge in $S$ because of our construction of $M$. Recall that, $w(P)$ is at least $\frac{(s+1)^4}{s}$ . This implies that one of the edges that projected weight onto $P$ projected at least $1/(s+1)^2$ fraction of total weight. Let $ux$ be that edge. By the above definition of the weight of $u$ ,we have that $x$ is unmarked and the subtree of $T$ rooted at $x$ has many terminals and furthermore, $T_x$ is not incident to any edge in $S$ and hence lies in $G[X]$. Furthermore, $T_x$ has at least 
$\frac{\frac{(s+1)^4}{s}-s}{(s+1)^2} \geq (s+1)$. This concludes the proof. 
\end{proof}

\subsection{The Algorithm}
\label{sec:decoAlgo}
%
%In the subsequent two sections, we prove the following two lemmas that
%allow us to find a single-bag lean witness in case of a bag being either
%not $(k,k)$-edge unbreakable or not $(2^k k,k)$-unbreakable.
%The proofs follow the principles of the 
%algorithms for the high-connectivity phase in the technique
%of randomized contractions~\cite{randcontr}
%and the \textsc{Hypergraph Painting} subroutine of~\cite{minbisection-STOC}; 
%in particular, the main ingredient is color-coding.
%
%\begin{lemma}\label{lem:edge-cut}
%Given an $n$-vertex graph $G$, an integer $k$, and a set $S \subseteq V(G)$
%with the property that every connected component $D$ of $G-S$ satisfies $|N_G(D)| \leq k$,
%one can in $2^{O(k \log k)} n^{O(1)}$ time either
%find an edge cut $(A,B)$ in $G$ such that the order of $(A,B)$ is some $\ell \leq k$ and it holds that $|A \cap S| > \ell$ as well as $|B \cap S| > \ell$, or correctly conclude
%that no such edge cut exists.
%\end{lemma}
%
%\begin{lemma}\label{lem:separations}
%Given an $n$-vertex graph $G$, an integer $k$, and a set $S \subseteq V(G)$
%with the property that every connected component $D$ of $G-S$ satisfies $|N_G(D)| \leq k$,
%one can in $2^{O(k \log k)} n^{O(1)}$ time either
%find a separation $(A,B)$ in $G$ such that the order of $(A,B)$ is some $\ell \leq k$ and it holds that $|A \cap S| > \ell$ and $|B \cap S| > \ell$, or correctly conclude
%that no separation $(A,B)$ of order at most $k$ in $G$ has the property that $|A \cap S| > 2^k k$ and $|B \cap S| > 2^k k$.
%\end{lemma}

%\paragraph*{Tree decompositions: cleanup and compactification}
For our proof we will use the following \emph{cleanup procedure} on a tree decomposition $(\tau,\chi)$ of a graph $G$: as long as there exists
an edge $st \in E(\tau)$ with $\chi(s) \subseteq \chi(t)$, contract the edge $st$ in $\tau$, keeping the name $t$ and the bag $\chi(t)$
at the resulting vertex. We shall say that a node $s$ and bag $\chi(s)$ \emph{disappears} in a cleanup step.
Clearly, the final result $(\tau',\chi')$ of a cleanup procedure is a tree decomposition of $G$ and every adhesion of $(\tau',\chi')$ is also an adhesion of $(\tau,\chi)$. Observe that $|E(\tau')| \leq |V(G)|$: if one roots $\tau'$ at an arbitrary vertex,
 going from child to parent on every edge $\tau'$ we forget at least one vertex of $G$, and every vertex can be forgotten only once.

It is well-known that every rooted tree decomposition can be refined to a compact one; see e.g.~\cite[Lemma 2.8]{BojanczykP16}. We will need  the following for our proof.
%For convenience, we provide a formulation of this fact suited for our needs.
%a full proof can be found in Appendix~\ref{app:compactification}.

\begin{lemma}[\cite{DBLP:journals/corr/abs-1810-06864},Lemma~2.3]\label{lem:compactification}
Given a graph $G$ and its tree decomposition $(\tau,\chi)$, one can compute in polynomial time a compact tree decomposition $(\wh{\tau},\wh{\chi})$ of $G$ such that every bag of $(\wh{\tau},\wh{\chi})$ 
is a subset of some bag of $(\tau,\chi)$, and every adhesion of $(\wh{\tau},\wh{\chi})$ is a subset of some adhesion of $(\tau,\chi)$.
\end{lemma}

With Lemmas~\ref{lem:unbreakableTesting} and \ref{lem:compactification} in hand, we now formally prove
Theorem~\ref{thm:decomp}.

\begin{proof}[Proof of Theorem~\ref{thm:decomp}.]
First, we can assume that $G$ is connected, as otherwise we can compute a tree decomposition
for every component separately, and then glue them up in an arbitrary fashion.

We start with a naive unrooted tree decomposition $(\tau,\chi)$ that has a single bag
with the entire vertex set and iteratively improve it, using Lemma~\ref{lem:vertex-refine},
until it satisfies the conditions of Theorem~\ref{thm:decomp}, except for compactness,
      which we will handle in the end. 
We will maintain the invariant that every adhesion of $(\tau,\chi)$ is of size at most $s$.
At every step the potential $\Pot_{G,s}(\tau, \chi)$ will decrease, leading to at most
$n-2s-1$ steps of the algorithm.

Let us now elaborate on a single step of the algorithm.There is one reasons why $(\tau,\chi)$ may not satisfy
the conditions of Theorem~\ref{thm:decomp}:  it contains a bag that is not $((s+1)^5,s)$-edge unbreakable.  
%Note here that a bag that is not ((s+1)^5,s)$-edge unbreakable has more than $???+1$ vertices.
%
%For the first case, note that we can in time $2^{\Oh(k)} n^{\Oh(1)}$ check
%if there exists an adhesion that is not well-linked by brute-forcing over every partition
%of every adhesion in the decomposition $(T,\chi)$. Furthermore, a separation 
%that witnesses that some adhesion is not well-linked gives also ground to a single-bag
%lean witness of order at most $k$ that can be used by Lemma~\ref{lem:vertex-refine}.
Consider a bag $Q := \chi(t)$ that is not $((s+1)^5,s)$-edge unbreakable. 
Now using Lemma~\ref{lem:unbreakableTesting}, with graph $G$ and terminal set $Q$, in polynomial time we can 
find a a edge cut $(A,B)$ of $G$ of  order at most $s$  such that $|A\cap Q|\geq s+1$ and  
$|B\cap Q|\geq s+1$. Consequently, Lemma~\ref{lem:unbreakableTesting} allows us to find a single-bag lean witness of order at most $s$ for the node $t$ in polynomial time.

Thus, we uncovered a single-bag lean witness for a node $t\in V(\tau)$ satisfying $|\chi(t)|>2s+1$.
Hence, we may refine the decomposition by applying Lemma~\ref{lem:vertex-refine} and proceed iteratively with the refined decomposition. As asserted by Lemma~\ref{lem:vertex-refine}, the potential $\Pot_{G,k}(\tau,\chi)$ strictly decreases in each iteration.

We remark that between the refinement steps we need to reduce the number of edges in the decomposition to at most $n$ using the same reasoning as in the proof of Lemma~\ref{lem:compactification}. 
That is, as long as there exists
an edge $st \in E(\tau)$ with $\chi(s) \subseteq \chi(t)$, we can contract $s$ onto $t$ keeping $\chi(t)$ as the bag of the new node.
A direct check shows that this operation neither increases the sizes of adhesions nor the potential $\Pot_{G,k}$ 
of the decomposition, while, as argued in the proof of Lemma~\ref{lem:compactification}, it bounds the number of edges of $\tau$ by $n$.

Observe that the potential $\Pot_{G,k}(\tau,\chi)$ is bounded polynomially in $n$ and every iteration can be executed in polynomial time. Hence, we conclude that the refinement process finishes within polynomial time and outputs an unrooted tree decomposition $(\tau,\chi)$ that satisfies all the requirements of Theorem~\ref{thm:decomp}, except for being compact.
This can be remedied by applying the algorithm of Lemma~\ref{lem:compactification}.
Note that neither the edge-unbreakability of bags nor the upper bound on the sizes of adhesions can deteriorate as a result of applying the algorithm of Lemma~\ref{lem:compactification},
as every bag (resp. every adhesion) of the obtained tree decomposition is a subset of a bag (resp. an adhesion) of the original one. This concludes the proof. 
\end{proof}

%!TEX root = main.tex

\section{A $s^{\cO(k)}$ Time Algorithm for {\sc Min $k$-cut} on Unweighted Graphs}%(Proof of Theorem~\ref{DP-main-theorem}) Dynamic Programming over unbreakable tree decompositions to find minimum $k$-cut for graphs having bounded solution size.}
\label{section:DP}
In this section we give our new algorithm for \mkk parameterized by $s$ and $k$. It is a dynamic programming 
algorithm over an edge unbreakable tree decomposition (Theorem~\ref{thm:decomp})  
and uses $T$-trees introduced by Thorup~\cite{Thorup08} crucially. 

Let $G$, $k$, $s$ be input to \mkk. Towards our proof, in polynomial time, we first compute 
a rooted compact tree decomposition $(\tau,\chi)$ of $G$, using Theorem~\ref{thm:decomp}, such that 
\begin{enumerate}
\setlength{\itemsep}{-2pt}
\item every adhesion of $(\tau,\chi)$ is of size at most $s$;
\item every bag of $(\tau,\chi)$ is $((s+1)^5,s)$-edge-unbreakable.
\end{enumerate}

Let $\tilde{P}^*$ be an optimum $k$-cut of $G$. Secondly, we compute a family $\mathcal{T}$ of spanning trees of $G$ such that there exists a tree $T^*\in \mathcal{T}$ such that $E(T^*)$ crosses $\tilde{P}^*$ at most $2k-2$ times. Recall that $T^*$ is said to be a $T$-tree of $G$. Such a family can be computed in polynomial time using Thorup's tree packing algorithm~\cite{Thorup08}.

Finally, we assume that the input graph $G$ is connected. Else, if $\cc(G)\geq k$, we return the connected components of $G$ itself as the desired partition of $G$ of value $0$. Otherwise, if  $\cc(G)\leq k$, then at the expense of $k^ks^k$ in time we can guess the values of $k$ and $s$ for each connected component. 

Thus, throughout this section, we assume that we are given $(\tau,\chi)$ and a tree $T$ from $\mathcal{T}$ and remark that the notations $G,k,s,(\tau,\chi),T$ have fixed meanings. Also recall that in some parts, we refer to a $l$-cut of a graph $G'$ as a $l$-partition of $V(G')$ but it will be clear from the context. We begin by defining additional terminologies and proving some results that will helpful in defining the states of the dynamic programming.%(\todo{explain with example}).

%In this section, we design a dynamic programming algorithm that takes as input positive integers $k$, $s$, an unweighted  multigraph $G$, a $((s+1)^5,s)$-unbreakable tree decomposition $(\tau,\chi)$ having $|A_t|\leq s$ for all $t\in V(\tau)$ and spanning tree $T$ of $G$\todo{do we need this as an input} and checks whether $\opt(G,k)\leq s$.  
%thm:paraAlg

\begin{definition}
For a tree $T'$ and an integer $c$, a partition $\tilde{P}$ of $V(T')$ is $c$-{\em admissible} with respect to $T'$ if $E(T')$ crosses $\tilde{P}$ at most $c$ times. For a subset $X$ of $V(T')$ a partition $\tilde{P}$ of $X$ is $T'$-{\em feasible} if $\tilde{P}$ is the projection onto $X$ of a $2k-2$-admissible partition of $V(T')$.
\end{definition}
%
%\todo{not able to parse this line} 
We will denote the family of all $T'$-feasible partitions of $X$ by ${\cal F}^X_{T'}$. Further, if $X=\emptyset$ then ${\cal F}^X_{T'}$ will contain the partition $P_{\emptyset}$, where $P_{\emptyset}$ is the partition of $X$ that contains a single empty part.

Given a tree $T'$, we define the partition of $V(T')$ obtained by removing a subset $E'$ of edges from $T'$ as {\em the  partition of $V(T')$} in which each part is equal to the set of vertices in some tree in the forest obtained by deleting $E'$ from $T'$. 

Given a tree $T'$ and $X \subseteq V(T')$, we define the {\em projection of $T'$ onto $X$} as the tree ${\sf proj}(T',X)$ obtained by the following procedure. Initially, set $T'' = T'$, and exhaustively apply the following modifications to $T''$:
(a) Delete from $T''$ all leaves of $T''$ that are not in $X$.
(b) Smooth all degree $2$ vertices $v$ of $T''$ that are not in $X$. 
Once neither one of the operations $(a)$, $(b)$ can be applied the procedure returns $T''$ as the projection ${\sf proj}(T',X)$. An example of $\sf{proj}(T',X)$ is depicted in Figure~\ref{fig:projection}.
\begin{figure}[h]
    \centering
    \includegraphics[width=0.8\textwidth]{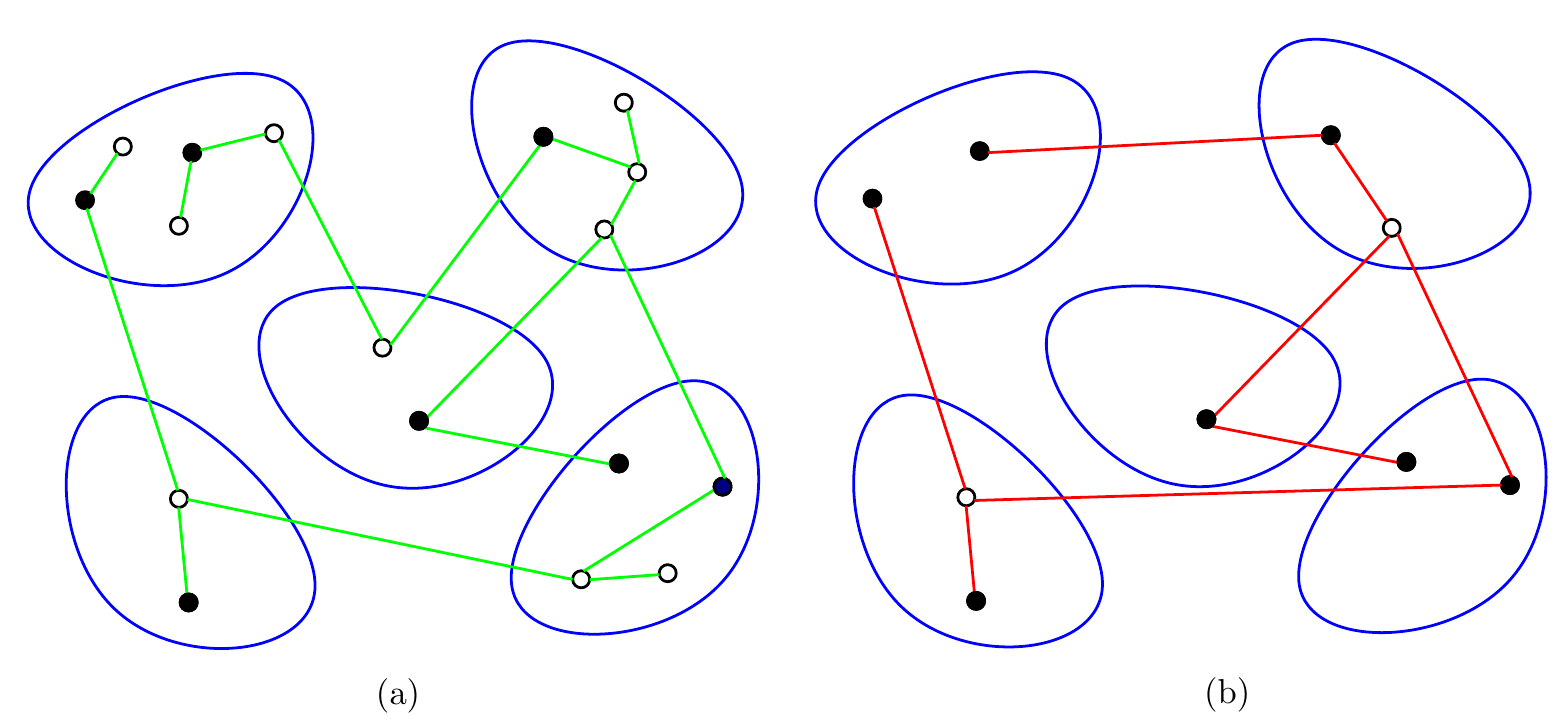}
    \caption{(a) Tree $T'$ is depicted by green edges and the set $X$ by black vertices. (b) $\sf{proj}(T',X)$ is depicted by red edges. }
    \label{fig:projection}
\end{figure}

For the projection ${\sf proj}(T',X)$, we will denote its vertices and edges by $V^{X}_{T'}$ and $E^{X}_{T'}$ respectively. We encapsulate some useful properties of ${\sf proj}(T',X)$ in the following observation.

\begin{observation}
\label{projPropertyObs}
All leaves and degree $2$ vertices of ${\sf proj}(T',X)$ are in $X$. Furthermore, the number of vertices in ${\sf proj}(T',X)$ is at most $2|X|$.
\end{observation}

The first part of the observation follows directly from the definition of the projection procedure, while the second follows from the well-known fact that in a tree, the number of vertices of degree at least three is at most the number of leaves. 
%If a vertex $v \notin X$ is a leaf, delete $v$ from $T$.  \item If a vertex $v \notin X$ has degree two in $T$, then
The following lemma shows that the set of $T$-feasible partitions of a set $X$ does not change when $T$ is projected on $X$. Thus, for enumerating $T$-feasible partitions of $X$ it will be sufficient to work with the projection of $T$ onto $X$ instead, which is much more efficient.
%The following lemma will let us assume that that the number of $T$-feasible partitions of $X\subseteq V(T)$ is small and will allow us to compute all $T$-feasible partitions of $X$ using only the smaller sized tree $\sf{proj}(T,X)$.
\begin{lemma}\label{sizeoffeasibleset} For all $X\subseteq V(T)$, $\mathcal{F}^X_T = \mathcal{F}^X_{\sf{proj}(T,X)}$, Further $\mathcal{F}^X_T$ is of size $\cO((4k|X|)^{2k})$ and can be computed in time proportional to its size. 
\end{lemma} 
\begin{proof}
For the proof, we show that every $T$-feasible partition of $X$ is a ${\sf proj}(T,X)$-feasible partition of $X$ and that every ${\sf proj}(T,X)$-feasible partition of $X$ is a $T$-feasible partition of $X$.

Let $\tilde{P}$ be a $T$-feasible partition of $X$, we show that $\tilde{P}$ is also a ${\sf proj}(T,X)$-feasible partition. Since $\tilde{P}$ is $T$-feasible, there exists a $2k-2$-admissible partition $\tilde{P}_{V(T)}$ of $V(T)$ whose projection on $X$ is $\tilde{P}$. Let $E_1$ be the set of edges from $E(T)$ that cross $\tilde{P}_{V(T)}$. We construct a set $E_2$ of edges from ${\sf proj}(T,X)$ of size at most $2k-2$ such that the projection of the partition obtained by removing $E_2$ from ${\sf proj}(T,X)$ on $X$ is $\tilde{P}$. For each pair $x,x'\in X$ in different parts in $\tilde{P}$, there is at least one edge in $E_1$ in the path $P(x,x')$ connecting $x$ and $x'$ in $T$, call one such arbitrary edge $e_{x,x'}$. For each $e_{x,x'}$, we will find an edge $e'_{x,x'}$ in $E({\sf proj}(T,X))$ such that $e'_{x,x'}$ lies in the path between $x$ and $x'$ in ${\sf proj}(T,X)$ and add it to $E_2$ and prove that none of the edges in $E_2$ are in any path $P(y,y')$ in $\sf{proj}(T,X)$ connecting two vertices $y,y'\in X$ that are in same component in $\tilde{P}$. %Since $E_2$ will have one edge for each pair of vertices in $x,x'\in X$ that are in different parts in $\tilde{P}$, 
Thus the projection of the partition obtained by removing $E_2$ from ${\sf proj}(T,X)$ on $X$ will be $\tilde{P}$ and hence $\tilde{P}$ will be ${\sf proj}(T,X)$-feasible. 
We now find the edge $e'_{x,x'}$. Let $e_{x,x'}=uv$, assume that $P(x,x')=P(x,u)\cup uv \cup P(v,x')$. Let $u'$ be the last vertex in $P(x,u)$ that has not been smoothened in ${\sf proj(T,X)}$ and $v'$ be the first vertex in $P(v,x')$ that has not been smoothened in ${\sf proj(T,X)}$. Note that since we never smooth out $x$ and $x'$, $u'$ and $v'$ must exist and can be equal to $x$ and $x'$ respectively. We set $e'_{x,x'}=u'v'$. 

We now prove that none of the edges $e'_{x,x'}\in E_2$ can lie in the path between two vertices $y,y'\in X$ that are in same components in $\tilde{P}$ in ${\sf proj}(T,X)$. Suppose not, assume some edge $e'_{x,x'}=u'v'\in E_2$ lies in the path between two vertices $y,y'\in X$ that are in the same component in $\tilde{P}$. Then, by the construction of ${\sf proj}(T,X)$, $e_{x,x'}=uv$ must lie in the path between $y,y'$ in $T$. Thus, $y,y'$ would not have been in the same components in $\tilde{P}$, this is a contradiction. 

For the converse, let $\tilde{P}$ be a ${\sf proj}(T,X)$-feasible partition of $X$, we show that $\tilde{P}$ is also a $T$-feasible partition. Since $\tilde{P}$ is ${\sf proj}(T,X)$-feasible, there exists a $2k-2$-admissible partition $\tilde{P}_{V^X_T}$ of $V^X_T$ whose projection on $X$ is $\tilde{P}$. Let $E_1$ be the set of edges from $E^X_T$ that cross $\tilde{P}_{V^X_T}$. We construct a set $E_2$ of edges from $T$ of size at most $2k-2$ such that the projection of the partition obtained by removing $E_2$ from $T$ on $X$ is $\tilde{P}$. For each pair $x,x'\in X$ in different parts in $\tilde{P}$, there is at least one edge in $E_1$ in the path $P(x,x')$ connecting $x$ and $x'$ in ${\sf proj}(T,X)$, call one such arbitrary edge $e_{x,x'}$. For each $e_{x,x'}$, we will find an edge $e'_{x,x'}$ in $E(T)$ such that $e'_{x,x'}$ lies in the path between $x$ and $x'$ in $T$ and add it to $E_2$ and prove that none of the edges in $E_2$ are in any path $P(y,y')$ in $T$ connecting two vertices $y,y'\in X$ that are in same component in $\tilde{P}$.
Thus the projection of the partition obtained by removing $E_2$ from $T$ on $X$ will be $\tilde{P}$ and hence $\tilde{P}$ will be a $T$-feasible partition of $X$. 
We now find the edge $e'_{x,x'}$. Let $e_{x,x'}=uv$, assume that $P(x,x')=P(x,u)\cup uv \cup P(v,x')$. Let $v'$ be the vertex adjacent to vertex $u$ in the path $P(u,v)$ between $u$ and $v$ in $T$. We set $e'_{x,x'}=uv'$. 
We now prove that none of the edges $e'_{x,x'}\in E_2$ can lie in the path between two vertices $y,y'\in X$ in $T$ that are in same components in $\tilde{P}$. Suppose not, assume some edge $e'_{x,x'}=uv'\in E_2$ lies in the path between two vertices $y,y'\in X$ that are in same components in $\tilde{P}$. All the interior vertices in path $P(u,v)$ were smoothened in ${\sf proj}(T,X)$. So, by construction of ${\sf proj}(T,X)$, none of the interior vertices are in $X$ nor have a path to $y$ or $y'$ in $T$ that doesn't include $u$ or $v$. Thus, $uv$ was an edge in the path between $y$ and $y'$ in ${\sf proj}(T,X)$. This implies that $y,y'$ would not have been in the same components in $\tilde{P}$, this is a contradiction.

Thus, we have proved that $\mathcal{F}^X_T = \mathcal{F}^{X}_{\sf{proj}(T,X)}$. To find the set $\mathcal{F}^X_T$, we do the following procedure. We initialize set $\mathcal{P}=\emptyset$. For each subset $E'$ of edges in $\sf{proj}(T,X)$ of size at most $2k-2$, we obtain the partition $\tilde{P}$ by removing the edges in $E'$ from $\sf{proj}(T,X)$ and add to $\mathcal{P}$ all partitions that $\tilde{P}$ refines. Finally we return $\mathcal{P}$ as $\mathcal{F}^X_{\sf{proj}(T,X)}$. Clearly, the set $\mathcal{F}^X_{\sf{proj}(T,X)}$ is of size $\cO((2|X|)^{2k}(2k)^{2k})$ which is $\cO((4k|X|)^{2k})$ and the procedure takes time corresponding to its size to compute it.
\end{proof}

{\em For the rest of this section, unless stated otherwise, we will assume that the weight of any partition of a set of vertices $X\subseteq V$ is computed with respect to $G[X]$}. 

We now define the states of the dynamic programming algorithm over the tree $\tau$ that we will use to compute the required value. For each node $t \in V(\tau)$ we define a function $f_t \colon \mathcal{F}^{A_t}_{T} \times \{ 1,\ldots ,k \} \longrightarrow \{ 0,\ldots , s \} \cup \{ \infty \}$ that our algorithm will compute. 
%
%\begin{equation*}
% $f \colon V(\tau) \times \mathcal{F}^{A_t}_T \times \{ 1,\cdots,k \} & \longrightarrow \{ 0,\cdots, s \} \cup \{ \infty\}$ 
%\end{equation*}
%
The domain of $f_t$ consists of all pairs $(\tilde{P}_{A_t}, i)$ where $\tilde{P}_{A_t}$ is a $T$-feasible (meaning crosses $T$ at most $ 2k-2$ times) partition of $A_t$, where $A_t$ is the set of vertices that are common between node $t$ and its parent in $\tau$, and $i$ is a positive integer less than or equal to $k$. On input $(\tilde{P}_{A_t}, i)$, $f_t$ returns the smallest possible weight %(with respect to $G_t$) 
of a $T$-feasible $i$-partition $\tilde{P}$ of $V(G_t$) such that $\tilde{P}_{A_t}$ is the projection of $\tilde{P}$ on $A_t$. However, if this weight is greater than $s$, or no such partition exists then $f_t(\tilde{P}_{A_t}, i)$ returns $\infty$.
The main step of the algorithm of Theorem~\ref{thm:paraAlg} is an algorithm that computes $f_t(\tilde{P}_{A_t}, i)$ for every node $t$ and pair $(\tilde{P}_{A_t}, i) \in \mathcal{F}^{A_t}_T \times \{ 1,\ldots,k \}$ assuming that the values of $f_{t'}$ have already been computed for all children $t'$ of $t$.  We now state that this step can be carried out efficiently.  

\begin{lemma} \label{dpstate_lemma}
There exists an algorithm that takes as input $(\tau,\chi)$, a node $t \in V(\tau)$, a $T$-feasible partition $\tilde{P}_{A_t}$ of $A_t$ and a positive integer $i \leq k$, together with the value of $f_{t'}(\tilde{P}_{A_{t'}}', i')$ for every child $t'$ of $t$, $T$-feasible partition $\tilde{P}_{A_{t'}}'$ of $A_{t'}$, and positive integer $i' \leq i$, and outputs $f_t(\tilde{P}_{A_t},i)$ in time $s^{\cO(k)}n^{\cO(1)}$.
\end{lemma}

Since a $T$-tree is a spanning tree of $G$ such that there exists an optimal $k$-cut of $G$ that is $T$-feasible, it follows that if $T$ is a $T$-tree of $G$ then $\opt(G,k)$ is given by $f_r(P_{\emptyset},k)$, where $r$ is the root of $\tau$, recall that $A_{r}=\emptyset$. If $T$ is not a $T$-tree of $G$, then $f_r(P_\emptyset,k)$ is the weight of some $k$-cut of $G$ of weight at most $s$ or $\infty$ and thus $f_r(P_\emptyset,k)\geq \opt(G,k)$. Therefore, we will now prove Theorem~\ref{thm:paraAlg} using Lemma~\ref{dpstate_lemma}.

\begin{proof}[Proof of Theorem \ref{thm:paraAlg}]
As explained in the beginning of the section, we assume that (a) the graph is connected, (b) we are given 
$(\tau,\chi)$, satisfying properties mentioned in Theorem~\ref{thm:decomp} and (c) a tree $T$ from $\mathcal{T}$~\cite{Thorup08}. 
We design an algorithm that takes all these inputs and works as follows. 
%mentioned in Theorem~\ref{thm:paraAlg} 
%and does the following. 
Starting from the leaf nodes, in a bottom up manner, for each node $t\in \tau$, the algorithm computes $\mathcal{F}^{A_t}_T$ using Lemma~\ref{sizeoffeasibleset} and obtains the value of $f_t(\tilde{P}_{A_t},i)$ for all pairs $(\tilde{P}_{A_t},i)\in \mathcal{F}^{A_t}_T \times \{1,\ldots,k\}$ using Lemma~\ref{dpstate_lemma}. Finally it returns the value $f_r(P_\emptyset,k)$, where $r$ is the root of $\tau$.

Since $f_r(P_\emptyset,k)\geq \opt(G,k)$ for any tree spanning tree $T$ of $G$ and $f_r(P_\emptyset,k)= \opt(G,k)$ if $T$ is a $T$-tree, the algorithm returns the desired output. 

It is easy to see that $\mathcal{F}^{A_t}_T$ is of size $\cO((4ks)^{2k})$ which is $s^{\cO(k)}$ and can be computed in time $s^{\cO(k)}$ from Lemma~\ref{sizeoffeasibleset} and the fact that $|A_t|\leq s$ for all $t\in \tau$. Thus the domain of $f_t$ is of size $s^{\cO(k)}\cdot k$ and each state can be computed in time $s^{\cO(k)}n^{\cO(1)}$ from Lemma~\ref{dpstate_lemma}. Thus the algorithm runs in time $s^{\cO(k)}n^{\cO(1)}$. This completes the proof.
%\cdot|V(\tau)|$.
\end{proof}
For proving Lemma~\ref{dpstate_lemma}, we will design an algorithm that takes all the inputs mentioned in Lemma~\ref{dpstate_lemma} and returns a positive integer $v$, which is the weight of some $T$-feasible $i$-partition $\tilde{P}$ of $V(G_t)$ having weight at most $s$ such that $\tilde{P}_{A_t}$ is the projection of $\tilde{P}$ on $A_t$. If no such partition exists, the algorithm will return $v = \infty$. This automatically ensures that $f_t(\tilde{P}_{A_t},i) \leq v$. The difficult part is to ensure that $f_t(\tilde{P}_{A_t},i) \geq v$. If $f_t(\tilde{P}_{A_t},i)=\infty$, this inequality trivially holds, and so it suffices to prove it for the cases when $f_t(\tilde{P}_{A_t},i)$ is finite and a $T$-feasible $i$-partition $\tilde{P}$ of $V(G_t)$ such that $\tilde{P}_{A_t}$ is the projection of $\tilde{P}$ on $A_t$ and having weight at most $s$ exists.
The proof of Lemma~\ref{dpstate_lemma} spans the remainder of this section and follows the route of  ``randomized contraction style Dynamic Programming'' \cite{randcontr,DBLP:journals/corr/abs-1810-06864}  but is somewhat more complicated because at every step we need to use the $T$-tree $T$ to speed up computations. 
We begin by giving names to some of the central objects of this proof. The notation for these tools will be used throughout the section and are summarized in Figure~\ref{fig:notations}.
%We will describe the algorithm in the later parts of this section, but the algorithm will have the property that $f_t(\tilde{P}_{A_t},i) \leq v$ just by the definition of what it returns. 

%Thus, the interesting case is when $f_t(\tilde{P}_{A_t},i) \neq \infty $, that is when a $T$-feasible $i$-partition $\tilde{P}$ of $V(G_t)$ such that $\tilde{P}_{A_t}$ is the projection of $\tilde{P}$ on $A_t$ and having weight at most $s$ exists. In this case, we will show that $f_t(\tilde{P}_{A_t},i) \geq v$. 
% %%%%%%%%%%%%%%%%%%%%%%%%%%%%%%%%%%%%%%%
%From now on, we will handle the case when the size of $\chi(t)$ is not of $O(s\log s)$. 
%We will now develop the necessary tools to prove Lemma~\ref{dpstate_lemma}, and postpone the actual proofs of Lemma~\ref{dpstate_lemma} to a later part of this Section.

Let $\tilde{P}(\tilde{P}_{A_t},i)$ be a smallest possible weight %(with respect to $G_t$)
$T$-feasible $i$-partition of $V(G_t)$ such that $\tilde{P}_{A_t}$ is the projection of $\tilde{P}(\tilde{P}_{A_t},i)$ on $A_t$ having weight at most $s$. We say that $\tilde{P}(\tilde{P}_{A_t},i)$ realises $f_t(\tilde{P}_{A_t},i)$. We introduce the following notations assuming such a partition exists. %These notations will be used throughout this section and hence are summarised in Figure~\ref{Fig:notations} for quick reference. 
\begin{figure}[h]
    \centering
    \includegraphics[width=1.0\textwidth]{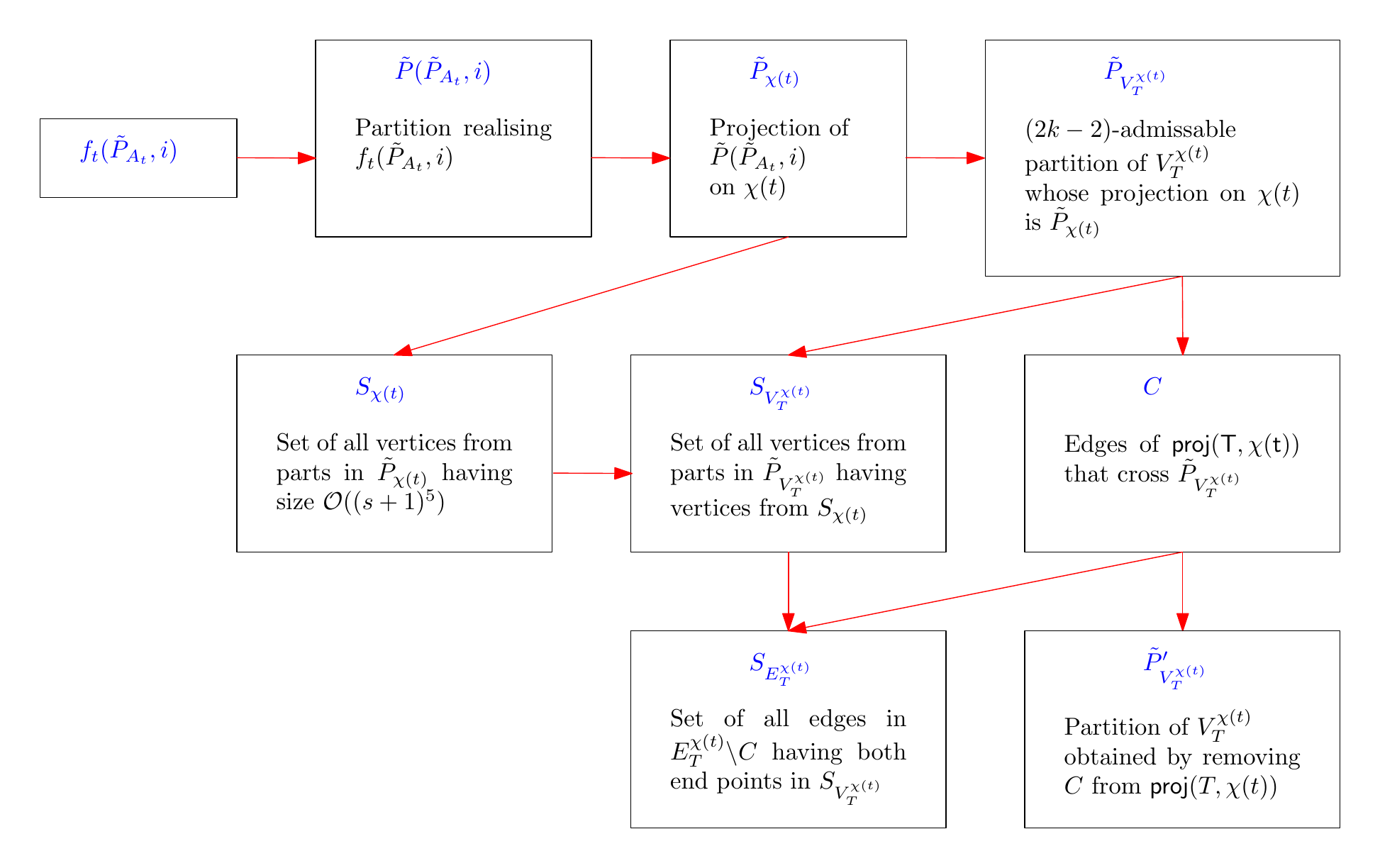}
    \caption{Various notations associated with the state $f_t(\tilde{P}_{A_t},i)$ along with their dependencies.}
    \label{fig:notations}
\end{figure}

Let $\tilde{P}_{\mathcal{\chi}(t)}$ be the projection of $\tilde{P}(\tilde{P}_{A_t},i)$ onto $\mathcal{\chi}(t)$.
%which is the partition of $\mathcal{\chi}(t)$ in $\tilde{P}(\tilde{P}_{A_t},i)$.
%
Since $\tilde{P}(\tilde{P}_{A_t},i)$ is a $T$-feasible $i$-partition of $V(G_t)$, $\tilde{P}_{\chi(t)}$ must be a $T$-feasible partition of $\chi(t)$. Recall that the vertices and edges of $\sf{proj}(T,\mathcal{\chi}(t))$ are denoted by $V^{\chi(t)}_T$ and $E^{\chi(t)}_T$ respectively. An example of $\chi(t)$ along with $A_t$ and ${\sf proj}(T,\chi)$ is illustrated in Figure~\ref{fig:projchi}. 
\begin{figure}[h]
    \centering
    \includegraphics[width=0.7\textwidth]{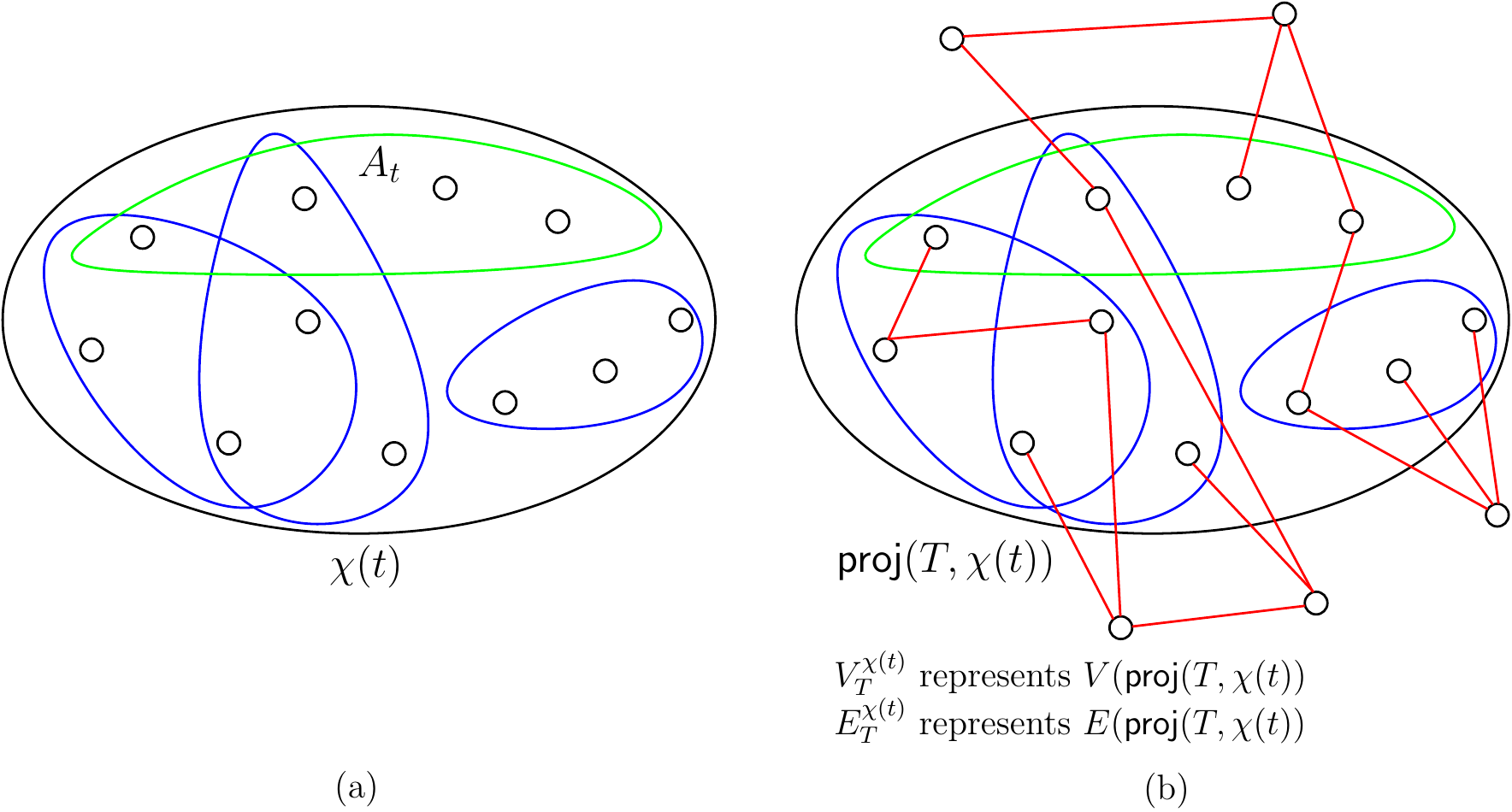}
    \caption{(a) An example of a bag $\chi(t)$ in $(\tau,\chi)$ along with the adhesions (blue) associated with it, the adhesion $A_t$ is in green. (b) An example of ${\sf proj}(T,\chi(t))$ for the bag $\chi(t)$ in $(a)$.}
    \label{fig:projchi}
\end{figure}
From Lemma~\ref{sizeoffeasibleset} it follows that $\tilde{P}_{\chi(t)}$ is the projection onto $\chi(t)$ of some $2k-2$-admissible partition $\tilde{P}_{V^{\chi(t)}_T}$ of $V^{\chi(t)}_T$. Let $C$ be the set of edges of $\sf{proj}(T,\mathcal{\chi}(t))$ that cross %$\tilde{P}_{\sf{proj}}$
$\tilde{P}_{V^{\chi(t)}_T}$, it is easy to see that $|C|\leq 2k-2$. Let $\tilde{P}'_{V^{\chi(t)}_T}$ be the partition of $V^{\chi(t)}_T$ obtained by deleting the edges in $C$ from $\sf{proj}(T,\chi(t))$. Observe that $\tilde{P}'_{V^{\chi(t)}_T}$ is a refinement of $\tilde{P}_{V^{\chi(t)}_T}$. We remark that $\tilde{P}'_{V^{\chi(t)}_T}$ and $\tilde{P}_{V^{\chi(t)}_T}$ are different objects, and that this difference is crucial to our arguments (we apologize for the notational similarity!). %Indeed, P is ... while P’ is ..., and so P’ is a refinement of P.
%
%Let $C$ be a set of edges in $\sf{proj}(T,\mathcal{\chi}(t))$ that when removed from $\sf{proj}(T,\mathcal{\chi}(t))$ \todo{Write what it means to remove edges from a tree} forms a partition whose projection on $\chi(t)$ is refined by $\tilde{P}_{\mathcal{\chi}(t)}$.
%
%Let $T'$ denote the intersection of $\sf{proj}(T,\mathcal{\chi}(t))$ with $\tau_t$ and $C'=E(T')\cap C$ be the set of edges from $C$ in $E(T')$.

Since $(\tau,\chi)$ is an $((s+1)^5,s)$-unbreakable tree decomposition, each part in $\tilde{P}_{\chi(t)}$ is of size $\cO((s+1)^5)$ except for at most one part. 
%We will first handle easy case when each part is of size at most $O(s\log s)$ first. In this case $\chi(t)$ is of size $O(s\log s)$. This case will give us an intuition as to how we plan to handle the harder case when the bag is not of size $O(s\log s)$. 
Let $S_{\chi(t)}$ denote the set of all vertices in parts having size $\cO((s+1)^5)$ in $\tilde{P}_{\chi(t)}$. % $S_{\chi(t)}$ denote the set of all other vertices in $\chi(t)$, that is $S_{\chi(t)}=\chi(t)\backslash B_{\chi(t)}$. Observe that if all parts have size at most $O(s^{1.5}\log s)$ vertices, then $B_{\chi(t)}=\emptyset$ and $S_{\chi(t)}=\chi(t)$.
%%%%%%%%%%%%%%%%%%
The set of vertices from $\chi(t)$ in each part in $\tilde{P}_{V^{\chi(t)}_T}$ are either all from $S_{\chi(t)}$ or from $\chi(t)\backslash S_{\chi(t)}$ as $\tilde{P}_{\chi(t)}$ is the projection of $\tilde{P}_{V^{\chi(t)}_T}$ on $\chi(t)$. Let the union of the vertices in parts in $\tilde{P}_{V^{\chi(t)}_T}$ that contain vertices from $S_{\chi(t)}$ be denoted by $S_{V^{\chi(t)}_T}$, that is $ S_{V^{\chi(t)}_T}= \bigcup_{P\in \tilde{P}_{V^{\chi(t)}_T}, P\cap \chi(t) \subseteq S_{\chi(t)}}P$. Observe that $S_{\chi(t)}\subseteq S_{V^{\chi(t)}_T}$.  
Let the set of edges in $E^{\chi(t)}_T \backslash C$ having both end points in $S_{V^{\chi(t)}_T}$ be denoted by $S_{E^{\chi(t)}_T}$. %and the set of vertices in $V^{\chi(t)}_T\backslash S_{V^{\chi(t)}_T}$ that are the endpoint of some edge in $C$ be denoted by $B$. Since the size of $C$ is at most $2k-2$, it follows that the size of $B$ is also at most $2k-2$. 
We now state a lemma that will help bound the sizes of the sets $S_{E^{\chi(t)}_T}$, and $S_{V^{\chi(t)}_T}$. 
%%%%%%%%%%%%%%%%%%%%%
%that lie inside the parts of $\tilde{P}_{\sf{proj}}$ having vertices from $S$ be called $E_S$ and the set of all vertices in these parts be called $V_S$. 
%%%%%%%%%%%%%%%%%%%%%%%%%%%%%%%%%%%%%%%%%%%%%%%%%%%%%%%%%%%%%%%%%%%%%%%%%%%%%%%%%%%%%%%%%%%%%%%%%%%%%%%%%%%%%%%%%%%%%%%%%%%%%%%%%%%%%%%%%%%%%%%%%%%
%%%%%%%%%%%%%%%%%%%%%%%%%%%%%%%%%%%%%%%%%%%%%%%%%%%%%%%%%%%%%%%%%%%%%%%%%%%%%%%%%%%%%%%%%%%%%%%%%%%%%%%%%%%%%%%%%%%%%%%%%%%%%%%%%%%%%%%%%%%%%%%%%%%
\begin{lemma}\label{size_comp}
The size of each part $P$ in the partition $\tilde{P}'_{V^{\chi(t)}_T}$ is at most $2(|P\cap \chi(t)| + 2k-2)$. Furthermore the sets $S_{E^{\chi(t)}_T}$ and $S_{V^{\chi(t)}_T}$ are of size at most $2\cdot(2k-1)(|S_{\chi(t)}|+2k-2)$. 
%Size of each connected component $X$ in $\sf{proj}(T,\chi(t))\backslash C$ \todo{Notation of T minus edges to denote partition} satisfies $|X|\leq 2|X\cap \chi(t)| + 2k$. 
\end{lemma}
\begin{proof}
%By Observation~\ref{projPropertyObs}, all leaves and degree two vertices of $\sf{proj}(T,\chi(t))$ are in $\chi(t)$ and the number of vertices in $\sf{proj}(T,\chi(t))$ is at most $2|\chi(t)|$. 
Since $\tilde{P}'_{V^{\chi(t)}_T}$ is the partition of $V^{\chi(t)}_T$ obtained by removing the edges in $C$ from $\sf{proj}(T,\chi(t))$, each part $P$ in it can be associated with a tree $T_P$ in the forest obtained by removing $C$ from $\sf{proj}(T,\chi(t))$. The number of vertices from $V^{\chi(t)}_T\backslash \chi(t)$ that are a leaf or a degree two vertex in $T_P$ is at most $2k-2$ because from Observation~\ref{projPropertyObs} they had degree more than two in $\sf{proj}(T,\chi(t))$ and thus must have been adjacent to an edge in $C$. The number of vertices in $P$ %from $V^{\chi(t)}_T\backslash \chi(t)$
having degree greater than two in $T_P$ is bounded by the number of leaves in $T_P$ which is bounded by $|P\cap \chi(t)|+2k-2$ because every leaf of $T_P$ is either in $\chi(t)$ or not in $\chi(t)$. Thus, the number of vertices from $V^{\chi(t)}_T\backslash \chi(t)$ in $P$ is bounded by $|P\cap \chi(t)|+2\cdot(2k-2)$ and hence the size of $P$ is at most $2\cdot(|P\cap \chi(t)|+2k-2)$.

Since the size of $C$ is at most $2k-2$, there are at most $2k-1$ parts in $\tilde{P}'_{V^{\chi(t)}_T}$ and thus the sizes of $S_{E^{\chi(t)}_T}$ and $S_{V^{\chi(t)}_T}$ are at most $2\cdot(2k-1)(|S_{\chi(t)}|+2k-2)$. 
\begin{comment}
Suppose an edge $e=(x,y)$ is removed from $\sf{proj}(T,\chi(t))$. Assume that $x$ is the parent of $y$ in the tree, then the component $X$ containing $x$ after removing $e$ from $\sf{proj}(T,\chi(t))$ will have at most $2|X\cap \chi(t)| + 1$ vertices. If we remove all the edges in $C$ iteratively, each edge will increase the number of vertices not in $\chi(t)$ in a component by at most one for exactly one component. Therefore, since each part $P$ in the partition $\tilde{P}'_{V^{\chi(t)}_T}$ is a connected component in the forest that remains after deleting the $C$ edges from $\sf{proj}(T,\chi(t))$, it follows that $|P|\leq 2|P\cap \chi(t)| + 2k-2$. In fact since each edge in $C$ contributes to the increase in the number of non $\chi(t)$ vertices in exactly one component $S_{V^{\chi(t)}_T}$, the total number of vertices in all parts in $P$ having vertices from $S$ is at most $2|S_{\chi(t)}|+2k$. The same argument holds $S_{E^{\chi(t)}_T}$, the number of edges from $E^{\chi(t)}_T$ that are within the parts in $P$ having vertices from $S$. 
\end{comment}
%Each vertex $v$ that is in $X$ but not in $\chi (t)$, $v \in X, v \notin \chi(t)$, can be associated with a vertex in $X\cap \chi (t)$ or an edge in $C$.\\
%Thus $|X| = |X\cap \chi}(t)| + |X\backslash \chi (t)| \leq 2|X \cap \chi (t)| +2k$.
\end{proof}
%If there exists a $i$-partition $P^*$ of $\mathcal{B}(T_t)$ of size at most $s$ that respects $P_A$. Since $T_d$ is $s$-unbreakable, the projection of $P^*$ on $\mathcal{B}(t)$ will partition $\mathcal{B}(t)$ into one big part say $B$ having more than $O(s\log s)$ vertices and all other parts will be small, having at most $O(s\log s)$ vertices. Let the vertices in these small parts be denoted by $S$. In $T'\backslash C'$, $S$ will be in multiple connected components. Let the edges in these components be denoted by $E_S$ and vertices in these components be denoted by $V_S$. By Lemma~\ref{size_comp}, $|E_S|\leq 2s+2k$. Let $V_B$ denote the vertices in the big component that are adjacent to edges in $C'$.
%\begin{corollary}\label{compsize_corollary}
%$|V_S| \leq 2|S|+2k$, and $|E_S|\leq 2|S|+2k$. 
%\end{corollary}
Since $S_{\chi(t)}$ is of size $\cO(k(s+1)^5)$, the following observation directly follows from Lemma~\ref{size_comp}
\begin{observation}
\label{sizeof4sets}
The set $C$ is of size at most $2k-2$ and the sets $S_{E^{\chi(t)}_T}$ and $S_{V^{\chi(t)}_T}$ are of size $\cO(k^2(s+1)^5)$.
\end{observation}
%%%%%%%%%%%%%%%%%%%%%%%%%%%%%%%%%%%%%%%%%%%%%%%%%%%%%%%%%%%%%%%%%%%%%%%%%%%%%%%%%%%%%%%%%%%%%%%%%%%%%%%%%%%%%%%%%%%%%%%%%%%%%%%%%%%%%%%%%%%%%%%%%%%%%%%%%%%%%%%%%%%%%%%%%%%%%%%%%%%%%%%%%%%%%%%%%%%%%%%%%%%%%%%%%%
In later parts of this section, we develop an algorithm that takes as input a subset $C'$ of edges of ${\sf proj}(T,\chi(t))$ and outputs a positive integer $v$ greater than or equal to $f_t(\tilde{P}_{A_t},i)$. Further if $C\subseteq C'$ and $S_{E^{\chi(t)}_T}\cap C' = \emptyset$, then $v$ will be equal to $f_t(\tilde{P}_{A_t},i)$. We state this as a result below in Lemma~\ref{DPSubLemma} and will be using it to prove Lemma~\ref{dpstate_lemma}.

\begin{lemma}\label{DPSubLemma}
There exists an algorithm that takes as input $(\tau,\chi)$, a node $t \in V(\tau)$, a $T$-feasible partition $\tilde{P}_{A_t}$ of $A_t$, a positive integer $i\leq k$, $\sf{proj}(T,\chi(t))$, a set of edges $C'\subseteq E^{\chi(t)}_T$, together with the value of $f_{t'}(\tilde{P}_{A_{t'}}', i')$ for every child $t'$ of $t$, $T$-feasible partition $\tilde{P}_{A_{t'}}'$ of $A_{t'}$, and positive integer $i' \leq i$, and outputs an integer $v\geq f_t(\tilde{P}_{A_t},i)$ in time $s^{\cO(k)}n^{\cO(1)}$. Further if $f_t(\tilde{P}_{A_t},i)\neq \infty$, $C\subseteq C'$, and $S_{E^{\chi(t)}_T}\cap C' = \emptyset$, then $v\leq f_t(\tilde{P}_{A_t},i)$.
%Given a partition $\{E'_{C} ,E'_{S}\}$ of $E(\sf{proj}(T,\chi(t)))$ such that $C\subseteq E'_{C}$ and $E_S\subseteq E'_S$ and a partition $\{V'_S,V'_B\}$ of $V(\sf{proj}(T,\chi(t)))$ such that $V_S \subseteq V'_S$ and $V_B \subseteq V'_B$, $f_t(\tilde{P}_{A_t},i)$ can be computed in time $O()$\todo{time}.
\end{lemma}
%\subsection{Proof of Lemma~\ref{dpstate_lemma}}
We will now prove that Lemma~\ref{dpstate_lemma} is true assuming Lemma~\ref{DPSubLemma}. In the proof, we will apply the well-known technique of splitters to obtain a family $\mathcal{C}$ of subsets of edges in $\sf{proj}(T,X)$ such that if $f_t(\tilde{P}_{A_t},i)\neq  \infty$, there is a set $C'$ in the family such that $C\subseteq C'$, and $S_{E^{\chi(t)}_T}\cap C'=\emptyset$. Then, we will apply Lemma~\ref{DPSubLemma} on each set in this family to obtain $f_t(\tilde{P}_{A_t},i)$. %The remainder of the section will be dedicated to the proof of Lemma~\ref{DPSubLemma}. 
\begin{proof}[Proof of Lemma~\ref{dpstate_lemma} assuming Lemma~\ref{DPSubLemma}] \label{Dpstate_lemma_proof}  
For the proof, we propose an algorithm that takes inputs as specified in Lemma~\ref{dpstate_lemma}, uses Lemma~\ref{DPSubLemma} and computes $f_t(\tilde{P}_{A_t},i)$. 

Firstly, the algorithm uses Lemma~\ref{splittersetlemma} with $S=E^{\chi(t)}_T$, $s_1=2k-2$, and $s_2=\cO(k^2(s+1)^5)$ to obtain a set $\mathcal{C}$ of subsets of $E^{\chi(t)}_T$. For each set $C'$ in $\mathcal{C}$, the algorithm calls the algorithm proposed in Lemma~\ref{DPSubLemma} with $C'$ along with the other inputs and obtains as output a value $v_{C'}$. Finally it returns the integer $v=\min\limits_{C'\in \mathcal{C}}v_{C'}$.

It follows from Lemma~\ref{DPSubLemma} that for each pair $C'$ in $\mathcal{C}$, $v_{C'}\geq f_t(\tilde{P}_{A_t},i)$. Thus, if $f_t(\tilde{P}_{A_t},i) = \infty$, then the algorithm will return the value $v=\infty$. If $f_t(\tilde{P}_{A_t},i) \neq \infty$, by Lemma~\ref{splittersetlemma} and Observation~\ref{sizeof4sets}, the set $\mathcal{C}$ contains a set $C'$ that satisfies $C\subseteq C'$, and $S_{E^{\chi(t)}_T}\cap C' = \emptyset$. By Lemma~\ref{DPSubLemma}, for that $C'$, $v_{C'} = f_t(\tilde{P}_{A_t},i)$. Thus if $f_t(\tilde{P}_{A_t},i) \neq \infty$, the algorithm will return the integer $v=f_t(\tilde{P}_{A_t},i)$.

The time taken by the algorithm is equal to the time taken to compute $\mathcal{C}$ plus the time taken to compute $v_{C'}$ for each $C'\in \mathcal{C}$. From Lemma~\ref{splittersetlemma} and Lemma~\ref{DPSubLemma}, it follows that the time taken by the algorithm is $s^{\cO(k)}n^{\cO(1)}$. This completes the proof.
\end{proof}

The remainder of this section will be dedicated to proving Lemma~\ref{DPSubLemma}, the crux of our result.
\subsection{Nice Decompositions of $\chi(t)$} %and their relation to $f_t$} %(Proof of Lemma~\ref{DPSubLemma})}
We first begin by defining {\em nice decompositions} of $\chi(t)$, a structure that will be helpful for proving Lemma~\ref{DPSubLemma}.
\begin{figure}[h]
    \centering
    \includegraphics[width=0.7\textwidth]{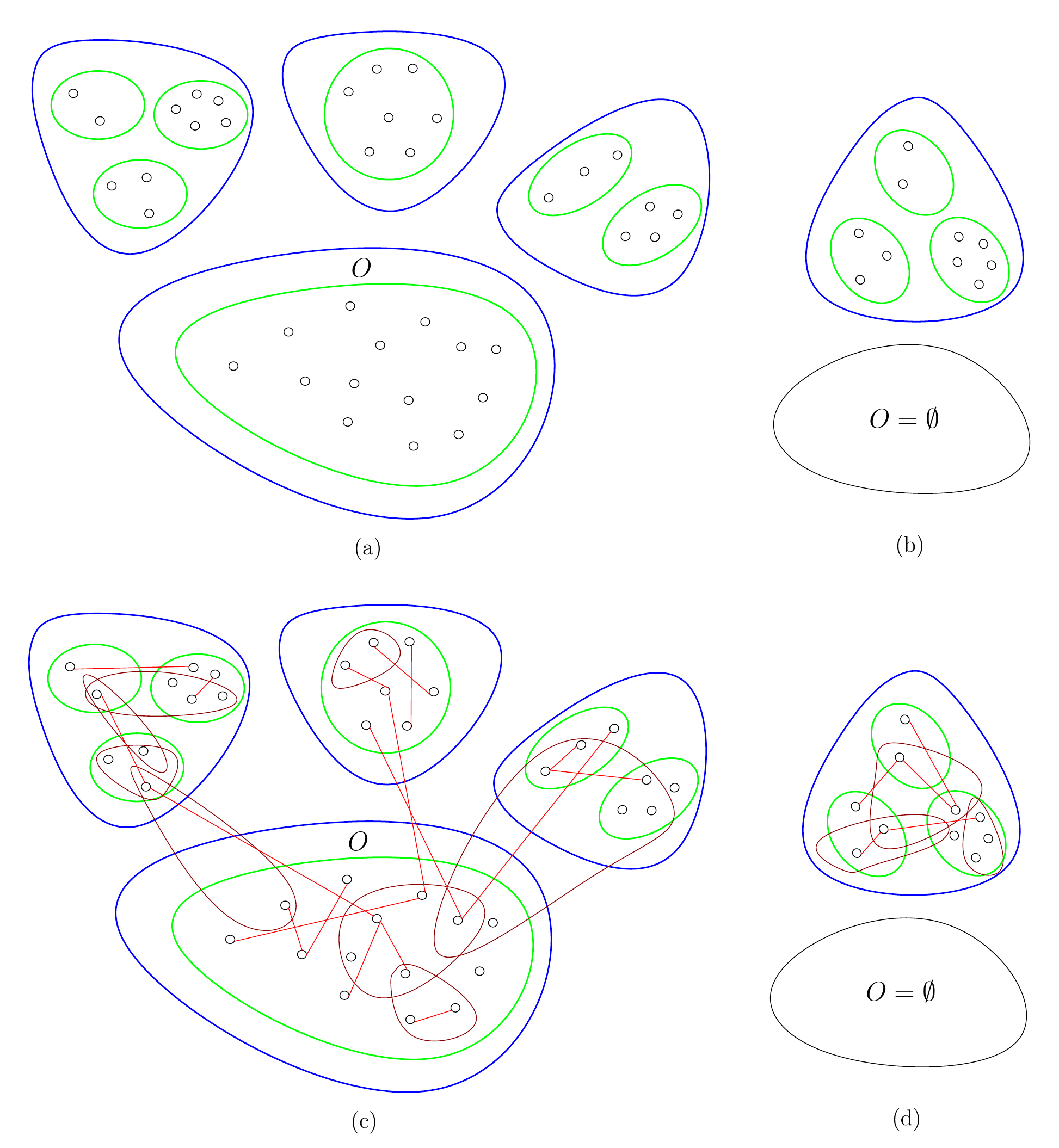}
    \caption{Examples of nice decomposition ($\tilde{P}'_{\chi(t)},\tilde{Q}_{\chi(t)},O)$ of $\chi(t)$. (a) $O\neq \emptyset$ (b) $O=\emptyset$.  The blue parts depict the parts in $\tilde{P}'_{\chi(t)}$ and the green parts depict those in $\tilde{Q}_{\chi(t)}$. In (c), (d) examples of possible edges and adhesions among the vertices of $\chi(t)$ in (a) and (b) respectively are depicted. Note that there are no adhesions and edges among different parts in  $\tilde{P}'_{\chi(t)}\backslash O$. }
    \label{fig:NiceDecomposition}
\end{figure}
\begin{definition}
\label{nicedecompositiondefinition}
Let $\tilde{P}'_{\chi(t)}$ be a partition of $\chi(t)$, $\tilde{Q}_{\chi(t)}$ be a refinement of $\tilde{P}'_{\chi(t)}$ and $O$ be a part in $\tilde{P}'_{\chi(t)}$ or an empty set. The triple $(\tilde{P}'_{\chi(t)},\tilde{Q}_{\chi(t)},O)$ is called a nice decomposition of $\chi(t)$ if it satisfies the following properties:
%(i) Part $P_1$ in $\tilde{P}'_{\chi(t)}$ is also a part in $\tilde{Q}_{\chi(t)}$, that is $P_1$ is not refined by $\tilde{Q}_{\chi(t)}$. 
%(i) There is a part $P$ in $\tilde{P}'_{\chi(t)}$ which is also a part in $\tilde{Q}_{\chi(t)}$, that is $P$ is not refined by $\tilde{Q}_{\chi(t)}$.
\begin{enumerate}\setlength\itemsep{-.7mm}
\item[(i)] If $O\neq \emptyset$, then the projection of $\tilde{Q}_{\chi(t)}$ on $O$ is $O$, that is $O$ is not refined by $\tilde{Q}_{\chi(t)}$ and if $O=\emptyset$, then $\tilde{P}'_{\chi(t)}$ has exactly one part.
\item[(ii)] The size of the projection of $\tilde{Q}_{\chi(t)}$ on each part $P \in \tilde{P}'_{\chi(t)}\backslash O$ has at most $2k-1$ parts. %If the size of $\chi(t)$ is of $O(s\log s)$, then there are at most $2k-2$ parts in the projection of $P_1$ on $\tilde{Q}_{\chi(t)}$.
%\\(iii) For all parts $P_i \in \tilde{P}'_{\chi(t)}, i>1$, each part $P$ in the projection of $\tilde{Q}_{\chi(t)}$ on $P_i$ has size of $O(s\log s)$.
\item[(iii)] There are no cross edges between the parts in the partition $\tilde{P}'_{\chi(t)}\backslash O$. 
%\todo{from G?} 
\item[(iv)] There are no adhesions in the set $\{A_{t'}: t'\in \chd(t)$ or $t'=t\}$ having vertices from more than one part in $\tilde{P}'_{\chi(t)}\backslash O$.
\end{enumerate} 
We refer to $O$ as the center of the nice decomposition.
\end{definition}
The structure of a nice decomposition of $\chi(t)$ is depicted by examples in Figure~\ref{fig:NiceDecomposition}.

In order to prove Lemma~\ref{DPSubLemma}, we first state and prove some results related to nice decompositions of $\chi(t)$ that will help us with the proof. Firstly, we develop an algorithm that computes a {\em small} family of nice decompositions of $\chi(t)$ containing a special nice decomposition in desirable cases. Then, we design another algorithm that computes a value greater than $f_t(\tilde{P}_{A_t},i)$ when input any nice decomposition and computes $f_t(\tilde{P}_{A_t},i)$ exactly if the input nice decomposition is special. 
%We now state the following results which we prove later in this subsection that will be helpful in proving Lemma~\ref{DPSubLemma}.
%%%%%%%%%%%%%%%%%%%%%%%%%%%%%%%%%%%%%%%%%%%%%%%%%%%%%%%%%%%%%%%%%%%%%%%%%%%%%%%%%%%%%%%%%%%%%%%%%%%%%%%%%%%%%%%%%%%%%%%%%%%%%%%%%%%%%%%%%%%%%%%%%%%%%%%%%%%%%%%%%%%%%%%%%%%%%%%%%%%%%%%%%%%%%%%%%%%%%%%%%%%%%%%%%%
\begin{lemma}
\label{NiceDecompositionPreprocessingLemma}
There exists an algorithm that takes as input $(\tau,\chi)$, a node $t \in V(\tau)$, $\sf{proj}(T,\chi(t))$, a set of edges $C'\subseteq E^{\chi(t)}_T$ and returns a set $\mathcal{D}$ of nice decompositions of $\chi(t)$ of size $\cO(s^{\cO(k)}\log n)$ in time $\cO(s^{\cO(k)}n^{\cO(1)})$. Furthermore, if %$C$, $B$ exists 
$f_t(\tilde{P}_{A_t},i)\neq \infty$, $C\subseteq C'$, and  $S_{E^{\chi(t)}_T}\cap C' = \emptyset$, then set $\mathcal{D}$ contains a nice decomposition $(\tilde{P}'_{\chi(t)},\tilde{Q}_{\chi(t)},O)$ such that $\tilde{Q}_{\chi(t)}$ is a refinement of $\tilde{P}_{\chi(t)}$.
\end{lemma}

%\subsubsection{(Proof  of Lemma~\ref{NiceDecompositionPreprocessingLemma}) Finding an useful set of Nice Decompositions of $\chi(t)$}
\begin{proof}
Given $G$, $k$, $(\tau,\chi)$, $t\in V(\tau)$, $\mathrm{proj}(T,\chi(t))$, and $C'\subseteq E_{T}^{\chi(t)}$ we design an algorithm that returns a set of nice decompositions by carrying out the following steps :
\begin{enumerate}\setlength\itemsep{-.7mm}
\item[(i)] Set $\tilde{R}$ to be the partition of $V^{\chi(t)}_T$ obtained by removing the edges in $C'$ from $proj(T,\chi(t))$.
\item[(ii)] If $\chi(t)$ is of size greater than $\cO((s+1)^5)$, go to step (iii). Else, let $\tilde{P}'_{\chi(t)}$ be the partition of $\chi(t)$ having all vertices in $\chi(t)$ in one part and let $\tilde{Q}_{\chi(t)}$ be the projection of $\tilde{R}$ on $\chi(t)$. If the number of parts in $\tilde{R}$ is less than or equal to $2k-1$, then return the set $\mathcal{D}$ of nice decompositions as $\{(\tilde{P}'_{\chi(t)},\tilde{Q}_{\chi(t)},\emptyset)\}$, else return $\mathcal{D}=
\emptyset$. 
%\item[(iii)] Initialize $B^*=B'$. For each part $R$ in $\tilde{R}$, if there is some vertex from  $B^*$ in $R$ then modify $B^*=B^*\cup R$.
%\item[(iv)] Combine all parts in $\tilde{R}$ that have vertices from $B^*$ and parts that have size greater than $O(s\log s)$ together in $\tilde{R}$. Let this part be $R_1$.
\item[(iii)] Let $S$ be the set of all parts in $\tilde{R}$. 
 Use Lemma~\ref{splittersetlemma} on $S$, $s1=2k-1$ and $s_2=(2k-1)(s^2+2s+k)$ and obtain a set $\mathcal{S}$ of subsets of $S$.
\item[(iv)] Construct a set $\mathcal{D}$ of nice decompositions using the following procedure: Initialize $\mathcal{D}=\emptyset$ and for each $X\in \mathcal{S}$, repeat the following steps: 
\begin{enumerate}\setlength\itemsep{-.7mm}
\item Initialize partition $\tilde{Q}=\tilde{R}$.
\item Combine all parts in $\tilde{Q}$ that are not in $X$. Let this part be $Q_1$.
\item Form a graph $G'$ with the parts in $\tilde{Q}$ except $Q_1$ as vertices. Add an edge between two parts in $G'$ if they have a cross edge between them or both parts have non empty intersection with some adhesion $A_{t'}$, $t'\in \chd(t)\cup t$. %Look at the graph $G'$ formed by the parts in $\tilde{Q}$ except $Q_1$, two parts are adjacent in $G'$ if they have a cross edge between them or both parts have non empty intersection with some adhesion $A_{t'}$, $t'\in \chd(t)\cup t$. 
\item Initialize $\tilde{Q}'=\tilde{Q}$, combine all parts in connected components in $G'$ having more than $2k-1$ parts with $Q_1$ in $\tilde{Q}'$. Call this part $O'$. %Now each connected component $O$ in $G'$ has at most $2k-2$ parts.
\item Initialize $\tilde{P}'=\tilde{Q}'$. For each connected component $Y$ in $G'$ having at most $2k-1$ parts, combine all parts in $Y$ in $\tilde{P}'$.  
\item Let $\tilde{Q}_\chi(t)$ be the projection of $\tilde{Q}'$ on $\chi(t)$ and let $\tilde{P}'_{\chi(t)}$ be the projection of $\tilde{P}'$ on $\chi(t)$. Let $O$ be the projection of $O'$ on $\chi(t)$.
\item Add the pair $(\tilde{P}'_{\chi(t)},\tilde{Q}_{\chi(t)},O)$ to $\mathcal{D}$ if $O\neq \emptyset$.
\end{enumerate}
\item[(v)] Return $\mathcal{D}$.
\end{enumerate}

\begin{claim}
\label{claim:nicedecompositionsetprop}
Each triple $(\tilde{P}'_{\chi(t)},\tilde{Q}_{\chi(t)},O) \in \mathcal{D}$ is a nice decomposition of $\chi(t)$. Further if $f_t(\tilde{P}_{A_t},i)\neq \infty$, $C\subseteq C'$, and $S_{E^{\chi(t)}_T}\cap C'=\emptyset$, then $\mathcal{D}$ contains a nice decomposition $(\tilde{P}'_{\chi(t)},\tilde{Q}_{\chi(t)},O)$ such that $\tilde{Q}_{\chi(t)}
$ is a refinement of $\tilde{P}_{\chi(t)}$.
\end{claim}
\begin{proof}
In the proof, will be using the notations summarized in Figure~\ref{fig:notations}. 
Firstly, we need to prove that every triple $(\tilde{P}'_{\chi(t)},\tilde{Q}_{\chi(t)},O)$ in $\mathcal{D}$ is a nice decomposition of $\chi(t)$. For this we need to show that $\tilde{Q}_{\chi(t)}$ is a refinement of $\tilde{P}'_{\chi(t)}$, $O$ is either a part in $\tilde{P}'_{\chi(t)}$ or an empty set and that the triple satisfies all the properties stated in Definition~\ref{nicedecompositiondefinition}. Then, to complete the proof we need to show that if $C\subseteq C'$, and  $S_{E^{\chi(t)}_T}\cap C'=\emptyset$, then $\mathcal{D}$ contains a triple $(\tilde{P}'_{\chi(t)},\tilde{Q}_{\chi(t)},O)$ such that $\tilde{Q}_{\chi(t)}$ is a refinement of $\tilde{P}_{\chi(t)}$. We divide the proof into two cases, one where the size of $\chi(t)$ is of $\cO((s+1)^5)$ and other where it is not. 

In the first case, $\mathcal{D}$ is either an empty set or has only one nice decomposition $\{(\tilde{P}'_{\chi(t)},\tilde{Q}_{\chi(t)},O)\}$, where, $\tilde{P}'_{\chi(t)}$ is the partition of $\chi(t)$ having all vertices in $\chi(t)$ in one part, $\tilde{Q}_{\chi(t)}$ is the projection of $\tilde{R}$ on $\chi(t)$ and $O=\emptyset$. If $\mathcal{D} \neq \emptyset$, then it is easy to see that $\tilde{Q}_{\chi(t)}$ is a refinement of $\tilde{P}'_{\chi(t)}$ and that all the four properties are satisfied since there is only one part in $\tilde{P}'_{\chi(t)}$. Thus, $\mathcal{D}$ is a set having zero or one nice decomposition of $\chi(t)$.
In this case, since the size of $\chi(t)$ is of $O((s+1)^5)$, it follows that by definition if $f_t(\tilde{P}_{A_t},i)\neq \infty$, then $S_{E^{\chi(t)}_T} = E^{\chi(t)}_T\backslash C$.  Thus, if $C\subseteq C'$, and $S_{E^{\chi(t)}_T}\cap C'=\emptyset$, then $C=C'$. Therefore, $\tilde{P}'_{\chi(t)}$ is the same as $\tilde{P}_{\chi(t)}$ and since $\tilde{Q}_{\chi(t)}$ is a refinement of $\tilde{P}'_{\chi(t)}$, it follows that it is a refinement of $\tilde{P}_{\chi(t)}$ also. This completes the proof for the first case.

We now argue for the case when the size of $\chi(t)$ is not of $\cO(s+1)^5$. By construction, every triple $(\tilde{P}'_{\chi(t)},\tilde{Q}_{\chi(t)},O)$ in $\mathcal{D}$ corresponds to some $X\in \mathcal{S}$ and has $O\neq \emptyset$. 

By construction, from step (iv), $\tilde{Q}'$ is a refinement of $\tilde{P}'$ since $\tilde{P}'$ is obtained by combining components in $\tilde{Q}'$. Since $\tilde{Q}_{\chi(t)}$ and $\tilde{P}'_{\chi(t)}$ are just projections of these partitions on $\chi(t)$, $\tilde{Q}_{\chi(t)}$ is a refinement of $\tilde{P}'_{\chi(t)}$. Since $O'$ is a part in $\tilde{P}'$ and $\tilde{Q}'$ and $O\neq \emptyset$, it easy to see that $O$ is part in $\tilde{P}'_{\chi(t)}$ and is not refined by $\tilde{Q}_{\chi(t)}$ as it is just the projection of $O'$ on $\chi(t)$. 
In step (d), we obtain $\tilde{Q}'$ from $\tilde{Q}$ by combining all parts in connected components of $G'$ that have more than $2k-1$ parts in them with $Q_1$. In step (e) we construct $\tilde{P}'$ from $\tilde{Q}'$ by combining the parts in each connected component $Y$ in $G'$ having at most $2k-1$ parts into one component in $\tilde{P}'$. Thus, the projection of $\tilde{Q}'$ on each part of $\tilde{P}'$ has at most $2k-1$ parts and hence the projection of $\tilde{Q}_{\chi(t)}$ on each part of $\tilde{P}'_{\chi(t)}$ has at most $2k-1$ parts.  
%
%We show that the properties required by $(\tilde{P}'_{\chi(t)},\tilde{Q}_{\chi(t)},O)$ to be a nice decomposition are satisfied by $\tilde{Q}'$ and $\tilde{P}'$ corresponding to $X$ and since $\tilde{Q}_{\chi(t)}$ and $\tilde{P}'_{\chi(t)}$ are just projections of these partitions on $\chi(t)$, they will also satisfy these properties. 
Since in $G'$ two parts from $\tilde{Q}\backslash Q_1$ are adjacent only if there is some cross edge between them or both have non empty intersection with some adhesion $A_{t'}$, $t'\in \chd(t)\cup t$, no parts in $\tilde{P}'\backslash O'$ and thus in $\tilde{P}'_{\chi(t)}\backslash O$ have any cross edge between them or have common adhesions. Therefore all the properties in Definition~\ref{nicedecompositiondefinition} are satisfied by $\tilde{P}'_{\chi(t)}$, $\tilde{Q}_{\chi(t)}$ and $O$ and thus $(\tilde{P}'_{\chi(t)},\tilde{Q}_{\chi(t)},O)$ is a nice decomposition of $\chi(t)$.

In this case, if $f_t(\tilde{P}_{A_t},i)\neq \infty$, $C\subseteq C'$, and $S_{E^{\chi(t)}_T}\cap C'=\emptyset$ then we need to show that there is a nice decomposition $(\tilde{P}'_{\chi(t)},\tilde{Q}_{\chi(t)},O)$ in $\mathcal{D}$ corresponding to some $X\in \mathcal{S}$ such that $\tilde{Q}_{\chi(t)}$ is a refinement of $\tilde{P}_{\chi(t)}$ and $O\neq \emptyset$. For this we first prove the following claim that talks about some properties of $\tilde{R}$  which will be helpful for identifying the required $X$.

%Recall that $\tilde{P}_{\chi(t)}$ is the projection of $\tilde{P}(\tilde{P}_{A_t},i)$ on  $\chi(t)$, where $\tilde{P}(\tilde{P}_{A_t},i)$ is the projection that realises $f_t(\tilde{P}_{A_t},i)$.\todo{?} 

\begin{claim}
\label{claim:preprocessingProp}
If $C\subseteq C'$, and $S_{E^{\chi(t)}_T}\cap C'=\emptyset$, then $\tilde{R}$ will satisfy the following properties. 
\begin{enumerate}\setlength\itemsep{-5pt}
\item[(P1)] $\tilde{R}$ is a refinement of $\tilde{P}'_{V^{\chi(t)}_T}$. 
\item[(P2)] Projection of $\tilde{R}$ on $S_{V^{\chi(t)}_T}$ will be the same as that of $\tilde{P}'_{V^{\chi(t)}_T}$ on $S_{V^{\chi(t)}_T}$. 
\item[(P3)]  Each part $R$ in $\tilde{R}$ has all vertices from either $S_{V^{\chi(t)}_T}$ or from $V^{\chi(t)}_T\backslash S_{V^{\chi(t)}_T}$, but not both.
%\item[(P4)] $R_1\in \tilde{R}$ contains vertices only from $V^{\chi(t)}_T\backslash S_{V^{\chi(t)}_T}$
\end{enumerate}
\end{claim}
\begin{proof}
Since $C\subseteq C'$ and $S_{E^{\chi(t)}_T}\cap C'=\emptyset$, $\tilde{R}$ is a refinement of $\tilde{P}'_{V^{\chi(t)}_T}$ and the projection of $\tilde{R}$ on $S_{V^{\chi(t)}_T}$ will be the same as that of $\tilde{P}'_{V^{\chi(t)}_T}$ on $S_{V^{\chi(t)}_T}$. Recall that the partition $\tilde{P}'_{V^{\chi(t)}_T}$ is the partition of $V^{\chi(t)}_T$ obtained by removing $C$ from ${\sf proj}(T,\chi(t))$. Since, each part in $\tilde{P}'_{V^{\chi(t)}_T}$ has vertices from either $S_{V^{\chi(t)}_T}$ or from $V^{\chi(t)}_T\backslash S_{V^{\chi(t)}_T}$, but not both, the same holds for the parts in $\tilde{R}$.
\end{proof}

Let $S'$ be the parts in $\tilde{R}$ having vertices from $S_{V^{\chi(t)}_T }$. By $(P2)$ and $(P3)$, $S'$ is a subset of $S$ which are the parts in $\tilde{R}$ and $S'$ is of size at most $2k-1$. Also, let $N$ be the parts in $\tilde{R}$ which are not in $S'$ and have cross edges to some part in $S'$ or have non empty intersection with an adhesion $A_{t'}$, $t'\in \chd(t)\cup t$ that also has non empty intersection with some part in $S'$. That is $N$ are the parts in $S\backslash S'$ that are adjacent to the parts in $S'$ in $G'$. We now prove that every part $P$ in $S'$ has at most $s^2+2s+k$ neighbours in $G'$. 

Each adhesion $A_{t'},t'\in \chd(t)$ that intersects different parts in $S$ can be associated with atleast one unique cross edge in $\tilde{P}(\tilde{P}_{A_{t}},i)$ from $E(G_{t'})\backslash E(G[A_{t'}])$ since $(\tau,\chi)$ is a compact tree decomposition and thus the graph obtained by removing the edges between any two vertices in $A_{t'}$ from $G_{t'}$ is connected. 
Since the number of vertices in an adhesion is at most $s$, the number of parts that $P$ can share an adhesion with is at most $s^2$. The number of cross edges in $\tilde{P}(\tilde{P}_{A_{t}},i)$ are atmost $s$, thus the number of parts that $P$ can share an edge in $E(G_t)$ with are at most $s$. In addition, the adhesion $A_{t}$ cannot be associated with a unique cross edges as it is the bag $t$ for which we are computing the state. Thus, $P$ can be adjacent to the parts that have vertices from $A_t$ and there are at most $s$ of these parts. Finally, since the small part that $P$ was contained in $\tilde{P}_{V^{\chi(t)}_T}$ is broken down into at most $k$ parts in $\tilde{P}'_{V^{\chi(t)}_T}$, at most $k$ edges can be saved so $P$ can have at most $k$ extra neighbours. Thus, $P$ has at most $s^2+2s+k$ neighbours in $G'$. Therefore, since $|S'|\leq 2k-1$, $N\leq (2k-1)(s^2+2s+k)$.

%shares an adhesion or has a cross edge with at most $s\cdot s\log s+s+k$ other parts. 
%Since every part in $P$ in $S'$ has size at most $O(s\log s)$ and the size of an adhesion in $(\tau,\chi)$ is at most $s$, it follows that $P$ will share an adhesion with at most $s\cdot O(s\log s)$ parts.
%
%Since the small part that $P$ was contained in $\tilde{P}_{V^{\chi(t)}_T}$ is broken down into at most $k$ parts in $\tilde{P}'_{V^{\chi(t)}_T}$ and since there are at most $s$ cross edges in the partition $\tilde{P}(\tilde{P}_{A_t},i)$, it follows that $P$ has cross edges to at most $s+k$ other parts. Also, as every adhesion can be associated with atleast one cross edge and since the size of an adhesion in $(\tau,\chi)$ is at most $s$, it follows that $P$ will share an adhesion with at most $s\times s$ parts. Thus, $P$ will have at most $s^2+s+k$ neighbours in $G'$. 

Thus from Lemma~\ref{splittersetlemma}, there exists $X\in \mathcal{S}$ such that $S'\subseteq X$ and $S'\cap N=\emptyset$. In $\tilde{Q}$, all parts that are not in $X$ are contained in $Q_1$ because of step (b). Therefore, all parts in $N$ are contained in $Q_1$ and the connected components in $G'$ either have all it's parts from $S'$ or no parts from $S'$. As $|S'|\leq 2k-1$, all connected components in $\tilde{Q}$ having more than $2k-1$ components contain no parts from $S'$. Thus $\tilde{Q}'$ has all the parts from $S'$ in it, and therefore $\tilde{Q}_{\chi(t)}$ %contain all vertices from $S_{\chi(t)}$ or none from it, $\tilde{Q}_{\chi(t)}$ 
is a refinement of $\tilde{P}_{\chi(t)}$. 

To complete the proof, we need to prove that $O\neq \emptyset$. Observe that each part in $\tilde{P}'_{\chi(t)}$ has all vertices from $S_{\chi(t)}$ or all from $\chi(t)\backslash S_{\chi(t)}$. Since we are in the case when the size of $\chi(t)$ is not of $\cO((s+1)^5)$, $\chi(t)\backslash S_{\chi(t)}\neq \emptyset$. Thus the number of parts in $\tilde{P}'_{\chi(t)}$ is more than one. If $O=\emptyset$, then since there are no edges or common adhesions between different parts in $\tilde{P}'_{\chi(t)}$, it means that $G_t$ is disconnected. But because $(\tau,\chi)$ is a compact tree decomposition, $G_t$ is connected, and therefore $O$ cannot be an empty set. This completes the proof of the claim.
\end{proof}
\begin{claim}
\label{claim:preprocessingSize}
$\mathcal{D}$ is of size $\cO(s^{O(k)}\log n)$ and is computed in time $\cO(s^{\cO(k)}n^{\cO(1)})$.
\end{claim}
\begin{proof}
In the case when the size of $\chi(t)$ is of $\cO((s+1)^5)$, $\mathcal{D}$ is of size at most one and is computed in polynomial time. In the other case, $\mathcal{S}$ is of size $\cO(s^{\cO(k)}\log n)$ from Lemma~\ref{splittersetlemma}. Since for each $X\in \mathcal{S}$ we add at most one nice decomposition to set $\mathcal{D}$ in step (vi), $\mathcal{D}$ is of size $\cO(ks)^{\cO(k)}\log n)$. Time taken to compute $\mathcal{S}$ in step (v) is $\cO(s^{\cO(k)}n^{\cO(1)})$ from Lemma~\ref{splittersetlemma}. All other steps take polynomial time, in particular, for each set $X\in \mathcal{S}$ it takes polynomial time to compute $\tilde{P}'_{\chi(t)}$, $\tilde{Q}_{\chi(t)}$ and $O$ and add it to $\mathcal{D}$. Thus, the time taken to compute $\mathcal{D}$ in this case is $\cO(s^{\cO(k)}n^{\cO(1)})$. 
\end{proof}
Claims~\ref{claim:nicedecompositionsetprop} and \ref{claim:preprocessingSize} complete the proof of Lemma~\ref{NiceDecompositionPreprocessingLemma}.
\end{proof}
%With this the proof of Theorem~\ref{DP-main-theorem} is complete.
%%%%%%%%%%%%%%%%%%%%%%%%%%%%%%%%%%%%%%%%%%%%%%%%%%%%%%%%%%%%%%%%%%%%%%%%%%%%%%%%%%%%%%%%%%%%%%%%%%%%%%%%%%%%%%%%%%%%%%%%%%%%%%%%%%%%%%%%%%%%%%%%%%%%%%%%%%%%%%%%%%%%%%%%%%%%%%%%%%%%%%%%%%%%%%%%%%%%%%%%%%%%%%%%%%%%%%%%%%%%%%%%%%%%%%%%%%%%%%%%%%%%%%%%%%%%%%%%%%%%%%%%%%%%%%%%%%%%%%%%%%%%%%%%%%

%In the remainder of this subsection, we will prove Lemma~\ref{FinalKnapsackStyleDP} and Lemma~\ref{NiceDecompositionPreprocessingLemma} to complete the proof of Lemma~\ref{DPSubLemma}.
With this, we start the second part of our section, an algorithm that takes as input a nice decomposition of $\chi(t)$ along with other inputs and outputs a value greater than or equal to $f_t(\tilde{P}_{A_t},i)$. Further if the partition $\tilde{Q}_{\chi(t)}$ in the nice decomposition satisfies a special property, then the algorithm will output $f_t(\tilde{P}_{A_t},i)$. We now state this as a Lemma. 
\begin{lemma}
\label{FinalKnapsackStyleDP}
There exists an algorithm that takes as input $(\tau,\chi)$, a node $t \in V(\tau)$, a $T$-feasible partition $\tilde{P}_{A_t}$ of $A_t$, a positive integer $i\leq k$, a nice decomposition $(\tilde{P}'_{\chi(t)},\tilde{Q}_{\chi(t)},O)$ of $\chi(t)$, together with the value of $f_{t'}(\tilde{P}_{A_{t'}}', i')$ for every child $t'$ of $t$, $T$-feasible partition $\tilde{P}_{A_{t'}}'$ of $A_{t'}$, and positive integer $i' \leq i$, and returns a positive integer $v$ such that $f_t(\tilde{P}_{A_t},i)\leq v$ or $v=\infty $ in time $k^{\cO(k)}n^{\cO(1)}$. Furthermore, if $\tilde{Q}_{\chi(t)}$ is a refinement of $\tilde{P}_{\chi(t)}$, then $v\leq f_t(\tilde{P}_{A_t},i)$.
\end{lemma}
%\subsubsection{ (Proof of Lemma~\ref{FinalKnapsackStyleDP}) Dynamic Programming over a Nice Decomposition of $\chi(t)$. }

To prove Lemma~\ref{FinalKnapsackStyleDP}, we design a dynamic programming algorithm on a given nice decomposition $(\tilde{P}'_{\chi(t)},\tilde{Q}_{\chi(t)},O)$ of $\chi(t)$. Let $p$ denote the number of parts in $\tilde{P}'_{\chi(t)}\backslash O$, it is easy to see that by the definition of a nice decomposition, if $O=\emptyset$, then $p=1$ and if not then $p=|\tilde{P}'_{\chi(t)}|-1$. We arbitrarily order all the parts in $\tilde{P}'_{\chi(t)}$ except $O$ and denote by $P_l$ the $l^{th}$ part in this order, where $l\geq 1$ and denote by $P_{\leq l}=\underset{x\leq l}\bigcup P_x$, the union of all parts $P_x$, $x\leq l$.  

Given a non negative integer $l$ that is less than or equal to $p$, we define the set $\mathcal{A}(l)$ to be the set of all adhesions in the set $\{A_{t'}:t'\in \chd(t)\}$ that have vertices only from $P_l\cup O$ and have non empty intersection with $P_l$. %Also, we define the set $\mathcal{A}(0)$ to be the set of all adhesions in the set $\{A_{t'}:t'\in \chd(t)\}$ that have vertices only from $O$. 
Also, we denote by $\mathcal{A}_{\leq}(l)=\underset{0\leq l'\leq l}\bigcup \mathcal{A}(l')$, the union of all sets $\mathcal{A}(l')$ where $0\leq l'\leq l$. Further, we define the graph $G(l)=G[O\cup P_l] \cup \underset{A_{t'}\in \mathcal{A}(l)}\bigcup G_{t'}$, to be the subgraph of $G$ induced by all the vertices in $O\cup P_{l} \cup \underset{A_{t'}\in \mathcal{A}(l)}\bigcup V(G_{t'})$ and define the graph $G_{\leq}(l)=G[O\cup P_{\leq l}] \cup \underset{A_{t'}\in \mathcal{A}_{\leq}(l)}\bigcup G_{t'}$, to be the subgraph of $G$ induced by all the vertices in $O\cup P_{\leq l }\cup \underset{A_{t'}\in \mathcal{A}_{\leq}(l)}\bigcup V(G_{t'})$.

%%%%%%%%%%%%%%%%%%%%%%%%%%%%%%%%%%%%%%%%%%%%%%%%%%%%%%%%%%%%%%%%%%%%%
We define two function $h,g: \{0,\ldots,p\} \times \{1,\ldots,i\}\longrightarrow \{0,\ldots,s\}\cup \{\infty\}$ that our algorithm will compute. The domain of $h$ and $g$ consists of all pairs $(l,j)$ where $l$ is a non negative integer less than or equal to $p$ and $j$ is a positive integer less than or equal to the input integer $i$. On input $(l,j)$, $h$ returns the smallest possible weight of a $j$-partition $\tilde{P}$ of $G(l)$ such that the projection of $\tilde{P}$ on $O\cup P_{l}$ is refined by the projection of $\tilde{Q}_{\chi(t)}$ on $O\cup P_l$ and if $A_t \subseteq O\cup P_l$ then $\tilde{P}_{A_t}$ is the projection of $\tilde{P}$ on $A_t$ and $g$ returns the smallest possible weight of a $j$-partition $\tilde{P}$ of $G_{\leq}(l)$ such that the projection of $\tilde{P}$ on $O\cup P_{\leq l}$ is refined by the projection of $\tilde{Q}_{\chi(t)}$ on $O\cup P_{\leq l}$ and if $A_t \subseteq O\cup P_{\leq l}$ then $\tilde{P}_{A_t}$ is the projection of $\tilde{P}$ on $A_t$. However, if this weight is greater than $s$, or no such partition exists then both $h(l,j)$ and $g(l,j)$ return $\infty$. %If $g(l,j)\neq \infty$, then we say that $\tilde{P}$ realises $g(l,j)$. 

The main step of an algorithm for Lemma~\ref{FinalKnapsackStyleDP} is an algorithm that computes $g(l,j)$ for every pair $(l,j)\in \{0,\ldots,p\} \times \{1,\ldots,i\}$ assuming that the value of $g(l',j')$ has been computed for every pair $(l',j')\in \{0,\ldots,l-1\} \times \{1,\ldots,i\}$. We now state that this step can be carried out.
%%%%%%%%%%%%%%%%%%%%%%%%%%%%%%%%%%%%%%%%%%%%%%%%%%
\begin{lemma}
\label{KnapsackDPstateLemma}
There exists an algorithm that takes all the inputs of Lemma~\ref{FinalKnapsackStyleDP}, along with a non negative integer $l\leq p$, a positive integer $j\leq i$ and if $l\geq 1$ the value of $g(l-1,j')$ for every positive integer $j'\leq j$ and returns a positive integer $v = g(l,j)$ in time $k^{\cO(k)}n^{\cO(1)}$.
\end{lemma}
Observe that $g(p,i)$ is the smallest possible weight of an $i$-partition $\tilde{P}$ of $G_t$ such that $\tilde{Q}_{\chi(t)}$ is a refinement of the projection of $\tilde{P}$ on $\chi(t)$ and the projection of $\tilde{P}$ on $A_t$ is $\tilde{P}_{A_t}$. By definition, $f_t(\tilde{P}_{A_t},i)$ is the weight of the smallest possible $i$-partition whose projection on $A_t$ is $\tilde{P}_{A_t}$, thus it follows that $g(p,i)\geq f_t(\tilde{P}_{A_t},i)$. If $f_t(\tilde{P}_{A_t},i)\neq \infty$ and $\tilde{Q}_{\chi(t)}$ is a refinement of $\tilde{P}_{\chi(t)}$ then $g(p,i)\leq f_t(\tilde{P}_{A_t},i)$ since by definition, $\tilde{P}(\tilde{P}_{A_t},i)$ is a $i$-partition of $G_t$ that has weight $f_t(\tilde{P}_{A_t},i)$, whose projection on $\chi(t)$ is $\tilde{P}_{\chi(t)}$, and whose projection on $A_t$ is $\tilde{P}_{A_t}$. We state these relations between $g(p,i)$ and  $f_t(\tilde{P}_{A_t},i)$ in the following observation.
\begin{observation}
\label{knapsackstateobservation}
$g(p,i)$ is always greater than or equal to $f_t(\tilde{P}_{A_t},i)$ and if $f_t(\tilde{P}_{A_t},i)\neq\infty$ and $\tilde{Q}_{\chi(t)}$ is a refinement of $\tilde{P}_{\chi(t)}$ then $g(p,i)$ is less than or equal to $f_t(\tilde{P}_{A_t},i)$.
\end{observation}
We are now ready to prove Lemma~\ref{FinalKnapsackStyleDP} using Lemma~\ref{KnapsackDPstateLemma} and Observation~\ref{knapsackstateobservation}.
\begin{proof}[Proof of Lemma~\ref{FinalKnapsackStyleDP} assuming Lemma~\ref{KnapsackDPstateLemma}]
For proving Lemma~\ref{FinalKnapsackStyleDP} we propose the following algorithm. The algorithm will take inputs as specified in the Lemma, compute the function $g$ using Lemma~\ref{KnapsackDPstateLemma} and finally returns an integer $v=g(p,i)$. From Observation~\ref{knapsackstateobservation}, it follows that the algorithm outputs a value $v\geq f_t(\tilde{P}_{A_t},i)$. Further if $f_t(\tilde{P}_{A_t},i)\neq \infty$ and $\tilde{Q}_{\chi(t)}$ is a refinement of $\tilde{P}_{\chi(t)}$ then $v\leq f_t(\tilde{P}_{A_t},i)$. %satisfies the specifications of the output stated in Lemma~\ref{FinalKnapsackStyleDP} and 
Since the domain of $g$ has $p\cdot i$ values, $p\leq n$ and $i\leq k$, from Lemma~\ref{KnapsackDPstateLemma} it is easy to see that the algorithm runs in time $k^{\cO(k)} n^{\cO(1)}$. 
\end{proof}
Since, $G_{\leq }(l)=G_{\leq}(l-1) \cup G(l)$ and all common edges of $G_{\leq}(l-1)$ and $G(l)$ are in $G(0)$, it follows that for any pair $(l,j)\in \{0,\ldots,p\} \times \{1,\ldots,i\}$ the following equation holds if $O\neq \emptyset$.
\begin{align}
g(l,j)=\begin{cases}
\underset{1\leq j'\leq j}\min g(l-1,j')+h(l,j-j'+1) & \mbox{if } l\geq 1\\
h(l,j) & \mbox{if } l=0\\
\end{cases}
\end{align}

If $O=\emptyset$, then there is only one part in $\tilde{P}'_{\chi(t)}$ of the nice decomposition and thus we can compute $g(1,j)=h(1,j)$, for all $j\leq i$ directly.

We are now ready to prove Lemma~\ref{KnapsackDPstateLemma}, the proof will have an algorithm that computes $h(l,j')$ efficiently for all $j'\leq j$.
\begin{proof}[Proof of Lemma~\ref{KnapsackDPstateLemma}]
Let $\mathcal{R}(l,j)$ be the set of partitions $\tilde{R}$ of $O\cup P_l$ having at most $j$ components such that the projection of $\tilde{Q}_{\chi(t)}$ on $O\cup P_l$ is a refinement of $\tilde{R}$ and if $A_t \subseteq  P_{\leq l}$, then $\tilde{P}_{A_t}$ is the projection of $\tilde{R}$ on $A_t$.

We define the function $h':\{1,\ldots,j\}\times \mathcal{R}(l,j)\longrightarrow \{0,\ldots,s\}\cup \{\infty\}$ that our algorithm will compute. On input $(j',\tilde{R})$, where $1\leq j'\leq j$ and $\tilde{R}\in \mathcal{R}(l,j)$, $h'$ returns the smallest possible weight of a $j'$-partition $\tilde{P}$ of $G(i)$ such that the projection of $\tilde{P}$ on $O\cup P_l$ is $\tilde{R}$. 

Clearly, for all positive integers $j'\leq j$, $h(l,j')$ satisfies the following equation.
\begin{equation}
h(l,j')=\underset{\tilde{R} \in \mathcal{R}(l,j)}\min h'(j',\tilde{R})
\end{equation}
Thus, to complete the proof, it is sufficient to compute $h'(j',\tilde{R})$ for all $j'\leq j$ given $\tilde{R}$. For this, we design a knapsack style dynamic programming algorithm on the set $\mathcal{A}(l)$ of adhesions associated with $P_l$. Arbitrarily order the adhesions in the set $\mathcal{A}(l)$ and let $a$ be a positive integer less than or equal to $|\mathcal{A}(l)|$, let $t_a$ denote the child of $t$ such that $A_{t_a}$ is the $a^{th}$ element in $\mathcal{A}(l)$. Let $G(l,a)$ denote the graph obtained by removing the edges among the vertices in $A_{t_a}$ from the graph $G[O\cup P_l] \cup G_{t_a}$ and let $G(l,0)$ denote the graph $G[O\cup P_l]$. Also, let $G_{\leq}(l,a)=\underset{0\leq a'\leq a}\bigcup G(l,a')$. 
We define two function $\zeta,\nu: \{0,\ldots,|\mathcal{A}(l)|\} \times \{1,\ldots,j\}\longrightarrow \{0,\ldots,s\}\cup \{\infty\}$ that our algorithm will compute. The domain of $\zeta$ and $\nu$ consists of all pairs $(a,r)$ where $a$ is a non negative integer less than or equal to $|\mathcal{A}(l)|$ and $r$ is a positive integer less than or equal to the input integer $j$. On input $(a,r)$, $\zeta$ returns the smallest possible weight of a $r$-partition $\tilde{P}$ of $G(l,a)$ such that the projection of $\tilde{P}$ on $O\cup P_{l}$ is $\tilde{R}$ and $\nu$ returns the smallest possible weight of a $r$-partition $\tilde{P}$ of $G_{\leq}(l,a)$ such that the projection of $\tilde{P}$ on $O\cup P_l$ is $\tilde{R}$. However, if this weight is greater than $s$, or no such partition exists then both $\zeta(a,r)$ and $\nu(a,r)$ return $\infty$.

It is easy to see that $h'(j',\tilde{R})$ satisfies the following equation just by the definition of $\nu$.
\begin{equation}
h'(j',\tilde{R}) = \nu(|\mathcal{A}(l)|,j')
\end{equation} 
 
Since $G_{\leq}(l,a)=G_{\leq}(l,a-1)\cup G(l,a)$ and there are no common edges between $G_{\leq}(l,a-1)$ and $G(l,a)$, for any pair $(a,r)\in \{0,\ldots,|\mathcal{A}(l)|\} \times \{1,\ldots,j\}$, the following equation holds, where $\tilde{R}_{A_{t_a}}$ denotes the projection of $\tilde{R}$ on the adhesion $A_{t_a}$.
\begin{align}
\nu(a,r)= \begin{cases}
\underset{0\leq r'\leq r}\min \nu(a-1,r') + \zeta(a,r-r'+|\tilde{R}_{A_{t_a}}|) & \mbox{if } a\geq 1\\
w(\tilde{R}) & \mbox{if } a=0 \mbox{ and }|\tilde{R}|=r\\
\infty & \mbox{otherwise}\\
\end{cases}
\end{align} 
Observe that $\zeta(a,r-r'+|\tilde{R}_{A_{t_a}}|)$ can be computed by the following equation just by the definition of $\zeta$ and $f$,
\begin{equation}
\zeta(a,r-r'+|\tilde{R}_{A_{t_a}}|)=f_{t_a}(\tilde{R}_{A_{t_a}},r-r'+|\tilde{R}_{A_{t_a}}|)-w(\tilde{R}_{A_{t_a}})
\end{equation}
Combining equations~$(3)$, $(4)$ and $(5)$, it is easy to see that given $j'\leq j$ and $\tilde{R}\in \mathcal{R}(l,j)$, $h'(j',\tilde{R})$ can be computed in $\cO(nk)$ time. The size of $\mathcal{R}(l,j)$ is at most $j^{2k}$ since from properties (i) and (ii) of Definition~\ref{nicedecompositiondefinition}, it follows that the projection of $\tilde{Q}_{\chi(t)}$ on $P_l\cup O$ has at most $2k$ parts. Thus, from equation~$(2)$, it follows that $h(l,j')$ can be computed in time $k^{\cO(k)}n^{\cO(1)}$. Combining this with equation $(1)$, it is clear that $g(l,j)$ can be computed in time $k^{\cO(k)}n^{\cO(1)}$. This completes the proof.
\end{proof}
We now are ready to complete the proof of Lemma~\ref{DPSubLemma} for which we obtained results throughout this section.

%%%%%%%%%%%%%%%%%%%%%%%%%%%%%%%%%%%%%%%%%%%%%%%%%%%%%%%%%%%%%%%%%%%%%%%%%%%%%%%%%%%%%%%%%%%%%%%%%%%%%%%%%%%%%%%%%%%%%%%%%%%%%%%%%%%%%%%%%%%%
\begin{proof}[Proof of Lemma~\ref{DPSubLemma}]
For the proof, we propose the following algorithm. Firstly, the algorithm will take inputs as specified in Lemma~\ref{DPSubLemma} and obtain a set of nice decompositions $\mathcal{D}$ from Lemma~\ref{NiceDecompositionPreprocessingLemma}. Then for each nice decomposition $D\in \mathcal{D}$, the algorithm obtains a value $v_D$ from Lemma~\ref{FinalKnapsackStyleDP}. Finally it returns the positive integer $v=\underset{D\in \mathcal{D}}{\min}\hspace{2mm} v_D$. If $\mathcal{D}=\emptyset$, then return $v=\infty$. Combining Lemma~\ref{NiceDecompositionPreprocessingLemma} and Lemma~\ref{FinalKnapsackStyleDP}, it is easy to see that the algorithm will return $v\geq f_t(\tilde{P}_{A_t},i)$. Further if $f_t(\tilde{P}_{A_t},i)\neq \infty$, $C\subseteq C'$, and $S_{E^{\chi(t)}_T}\cap C' = \emptyset$, then it will return $v= f_t(\tilde{P}_{A_t},i)$. The time taken by the algorithm is $\cO((s)^{\cO(k)}n^{\cO(1)})$ since $\mathcal{D}$ has at most $\cO((s)^{\cO(k)}\log n)$ sets and each set can be computed in $k^{\cO(k)}n^{\cO(1)}$ time from Lemma~\ref{NiceDecompositionPreprocessingLemma} and Lemma~\ref{FinalKnapsackStyleDP}. 
This proves Lemma~\ref{DPSubLemma}
\end{proof} 

%!TEX root = main.tex
\section{Combining all the Pieces: Proof of Theorem~\ref{thm:mainFPTAS} }
\label{section:combining}
In this section we conjure all the pieces we obtained so far and give a proof of our main result (Theorem~\ref{thm:mainFPTAS}).
\begin{proof}[{Proof of Theorem~\ref{thm:mainFPTAS}}]
Let $G, k, \epsilon$ be the input to \mkk. Also let $|V(G)|=n$. 
 We will output a partition $\tilde{P}$ with weight $v$  such that $v\leq (1+ \epsilon) \opt(G,k)$ with probability at least $\prup$. We first check if $\cc(G)\geq k$. If this is true then we know that $G$ has optimum $k$-cut of value $0$ and we can return $G$ itself. If $\epsilon <\frac{1}{n}$ then we run the  best 
known exact algorithm on $(G,k)$ running in time $n^{\cO(k)}$ and return the exact answer~\cite{KargerS96,Thorup08,DBLP:journals/corr/abs-1906-00417}. Clearly, 
$n^{\cO(k)}$ is $(k/\epsilon)^{\cO(k)} n^{\cO(1)}$ in this case.  So from now onwards we assume that 
$\cc(G)\leq k$ and $\epsilon \geq \frac{1}{n}$. Also, we set $\epsilon' = \frac{\epsilon}{10}$.

Now we apply a standard rounding procedure (similar to the well-known Knapsack PTAS \cite{DBLP:books/daglib/0015106}) that  reduces the problem in a $(1+\epsilon')$-approximation preserving manner to an unweighted multi-graph with at most $m^2/\epsilon$ edges. This implies that now the graph has at most $n^5$ edges (counting multiplicities). Let this graph be $G^\star$. 
Next we apply Lemma~\ref{RandLemma1} and obtain a subgraph $G_1$ of $G^\star$ with $V(G')=V(G^\star)$, such that 
$q=|E(G^\star)|-|E(G_1)|\leq 2\epsilon'\cdot \opt(G^\star,k)$. Further,  if $\cc(G_1)<k$,  then
 each non-trivial $2$-cut  of $G_1$ has weight at least $\frac{\epsilon' \cdot  \opt(G^\star,k)}{k-1}$. If $\cc(G_1)=k$, then we return the connected components as a partition  $\tilde{P}$. The cost of returned solution is 
 
 $$v \leq 2\epsilon'\cdot \opt(G^\star,k) \leq  2\epsilon'\cdot  (1+\epsilon')\opt(G,k) \leq (1+\epsilon) \opt(G,k).$$

Thus, we assume that $\cc(G_1)<k$ and each non-trivial $2$-cut  of $G_1$ has weight at least $\frac{\epsilon' \cdot  \opt(G^\star,k)}{k-1}$. 
 Now we apply the sparsification procedure described in Lemma~\ref{RandLemma2} and obtain a subgraph  $G_2$ with $V(G_2)=V(G_1)$, and a real number $r$  
 %having the same set of vertices as $G$ 
 such that with probability at least $\prup$, for all $k$-cuts $\tilde{P}$ in $G_1$, 
 $(1-\epsilon')\cdot w(G_1,\tilde{P})\leq w(G_2,\tilde{P})\cdot r \leq (1+\epsilon')\cdot w(G_1,\tilde{P})$. Here, $r=1/p$, where $p=\frac{100 \cdot \log n}{\epsilon'^2\cdot \opt(G_1,2)}$
 However, in $G_2$ we know that 
\begin{eqnarray*}
\opt(G_2,k) & \leq & (1+\epsilon') \cdot p \cdot \opt(G_1,k) \\
& = &  \frac{(1+\epsilon')100 \log n}{\opt(G_1,2) \cdot \epsilon'^2} \cdot \opt(G_1,k) \\
& \leq &  \frac{(1+\epsilon')100 \log n}{\opt(G_1,2) \cdot \epsilon'^2} \cdot \opt(G^\star,k)\\
%& \leq &  \frac{(1+\epsilon')^2100 \log n}{\opt(G_1,2) \cdot \epsilon'^2} \cdot \opt(G,k)\\
& \leq &  \frac{(1+\epsilon') \cdot k \cdot 100 \log n}{\epsilon'^3}\\
& \leq & \frac{\beta \cdot k \cdot 100 \log n}{\epsilon^3}.
\end{eqnarray*}
Here, $\beta$ is a fixed constant. Further,  in the second last transition we used the assumption that $\opt(G_1,2) \geq \epsilon' \cdot \opt(G^\star,k)/k-1$. Now we apply the Theorem~\ref{thm:paraAlg} on $G_2$ and for each 
$s\in \left\{1,\ldots,  \frac{\beta \cdot k \cdot 100 \log n}{\epsilon^3} \right\}$ and solve the problem exactly in time $s^{\cO(k)}n^{\cO(1)}$. 
Thus the running time of the algorithm is upper bounded by 
\begin{eqnarray*}
\left(\frac{\beta \cdot k \cdot 100 \log n}{\epsilon^3}\right)^{\cO(k)}n^{\cO(1)}& = & (k/\epsilon)^{\cO(k)} (\log n)^{\cO(k)}n^{\cO(1)}
 \leq  (k/\epsilon)^{\cO(k)} (k^{\cO(k)}+n) n^{\cO(1)}\\
& \leq &2^{\cO(k \log (\frac{k}{\epsilon}) ) }}n^{\cO(1). 
\end{eqnarray*}
This completes the running time analysis. All that remains is to show that what we get is an $(1+\epsilon)$-approximation algorithm. Let $s$ be the minimum value for which the algorithm described in Theorem~\ref{thm:paraAlg} returns yes. This implies that $s=\opt(G_2,k)$. Let $\tilde{P}$ be the corresponding partition. We return $s\cdot r+q$ as value of the cut and 
$\tilde{P}$ as a solution. Here, the value is the sum of edges deleted when we obtained $G_1$ from $G^\star$ and the value returned by Theorem~\ref{thm:paraAlg}  when ran on $(G_2,k,s)$. Now we have that 

\begin{eqnarray*}
\opt(G_2,k) \cdot r +q& \leq &   (1+\epsilon')\cdot \opt(G_1,k) + 2\epsilon'\cdot \opt(G^\star,k) \\
 &\leq & (1+\epsilon')\cdot \opt(G^\star,k) + 2\epsilon'\cdot \opt(G^\star,k) \\
  &\leq & (1+3\epsilon' )\cdot \opt(G^\star,k)  \\
  & \leq & (1+3\epsilon') \cdot (1+\epsilon') \cdot  \opt(G,k) \\
  & \leq & (1+\epsilon)\cdot \opt(G,k) 
\end{eqnarray*}

The correctness of the algorithm follows from Lemmas~\ref{RandLemma1},~\ref{RandLemma2} and Theorem~\ref{thm:paraAlg}. This concludes the proof. 
\end{proof}
%!TEX root = main.tex
\section{Conclusion}
\label{section:conclusion}
In this paper we gave a parameterized approximation algorithm with best possible approximation guarantee, and best possible running time dependence on said guarantee (upto \ETH and constants in the exponent) for \mkk. 
In particular, for every 
$\epsilon > 0$, the algorithm computes a $(1 +\epsilon)$-approximate solution in time $(k/\epsilon)^{\cO(k)}n^{\cO(1)}$. Along the way we also obtained a new exact algorithm with running time $s^{\cO(k)}n^{\cO(1)}$  on 
%connected
unweighted (multi-) graphs, where $s$ denotes the number of edges in a minimum $k$-cut. For an even more complete understanding of the parameterized apperoximation complexity of \mkk one could explore the possibility of lossy kernels of polynomial size in $s$ or $k$, and ratios between $2-\epsilon$ and $1+\epsilon$. Finally, an intriguing open problem about the parameterized complexity of \mkk is whether the problem admits an algorithm with running time $2^{\cO(s)}n^{\cO(1)}$.

%We would like to conclude with some interesting questions. Here, $\epsilon>0$ is a fixed constant. 

%\begin{itemize}
%\item Does \mkk parameterized by $s$ admit a lossy kernel of polynomial size with ratio $(1+\epsilon)$? 
%We refer the reader to the article of Lokshtanov et al.~\cite{DBLP:conf/stoc/LokshtanovPRS17} for the notion of lossy kernels. 
%\item Does \mkk parameterized by $k$ admit a lossy kernel of polynomial size with ratio $(1+\epsilon)$ or even $(2-\epsilon)$?  
%\item Does \mkk parameterized by $s$ admit an algorithm with running time $2^{\cO(s)}n^{\cO(1)}$?
%\end{itemize}

%\newpage
\bibliographystyle{alpha}
\bibliography{IEEEabrv,lean}

\newcommand{\etalchar}[1]{$^{#1}$}
\begin{thebibliography}{CCH{\etalchar{+}}16}

\bibitem[AWW14]{AbboudWW14}
Amir Abboud, Virginia~Vassilevska Williams, and Oren Weimann.
\newblock Consequences of faster alignment of sequences.
\newblock In Javier Esparza, Pierre Fraigniaud, Thore Husfeldt, and Elias
  Koutsoupias, editors, {\em Automata, Languages, and Programming - 41st
  International Colloquium, {ICALP} 2014, Copenhagen, Denmark, July 8-11, 2014,
  Proceedings, Part {I}}, volume 8572 of {\em Lecture Notes in Computer
  Science}, pages 39--51. Springer, 2014.

\bibitem[BD02]{BellenbaumD02}
Patrick Bellenbaum and Reinhard Diestel.
\newblock Two short proofs concerning tree-decompositions.
\newblock {\em Combinatorics, Probability {\&} Computing}, 11(6):541--547,
  2002.

\bibitem[BK15]{BenczurK15}
Andr{\'{a}}s~A. Bencz{\'{u}}r and David~R. Karger.
\newblock Randomized approximation schemes for cuts and flows in capacitated
  graphs.
\newblock {\em {SIAM} J. Comput.}, 44(2):290--319, 2015.

\bibitem[BP16]{BojanczykP16}
Miko{\l}aj Boja{\'{n}}czyk and Michal Pilipczuk.
\newblock Definability equals recognizability for graphs of bounded treewidth.
\newblock {\em CoRR}, abs/1605.03045, 2016.
\newblock Extended abstract appeared in the proceedings of LICS 2016.

\bibitem[BT17]{DBLP:conf/icml/BackursT17}
Arturs Backurs and Christos Tzamos.
\newblock Improving viterbi is hard: Better runtimes imply faster clique
  algorithms.
\newblock In Doina Precup and Yee~Whye Teh, editors, {\em Proceedings of the
  34th International Conference on Machine Learning, {ICML} 2017, Sydney, NSW,
  Australia, 6-11 August 2017}, volume~70 of {\em Proceedings of Machine
  Learning Research}, pages 311--321. {PMLR}, 2017.

\bibitem[CCH{\etalchar{+}}16]{randcontr}
Rajesh Chitnis, Marek Cygan, MohammadTaghi Hajiaghayi, Marcin Pilipczuk, and
  Micha\l{} Pilipczuk.
\newblock Designing {FPT} algorithms for cut problems using randomized
  contractions.
\newblock {\em {SIAM} J. Comput.}, 45(4):1171--1229, 2016.

\bibitem[CFK{\etalchar{+}}15]{DBLP:books/sp/CyganFKLMPPS15}
Marek Cygan, Fedor~V. Fomin, Lukasz Kowalik, Daniel Lokshtanov, D{\'{a}}niel
  Marx, Marcin Pilipczuk, Michal Pilipczuk, and Saket Saurabh.
\newblock {\em Parameterized Algorithms}.
\newblock Springer, 2015.

\bibitem[CKL{\etalchar{+}}18]{DBLP:journals/corr/abs-1810-06864}
Marek Cygan, Pawel Komosa, Daniel Lokshtanov, Michal Pilipczuk, Marcin
  Pilipczuk, Saket Saurabh, and Magnus Wahlstrom.
\newblock Randomized contractions meet lean decompositions.
\newblock {\em CoRR}, abs/1810.06864, 2018.

\bibitem[CQX20]{DBLP:journals/siamdm/ChekuriQX20}
Chandra Chekuri, Kent Quanrud, and Chao Xu.
\newblock {LP} relaxation and tree packing for minimum k-cut.
\newblock {\em {SIAM} J. Discret. Math.}, 34(2):1334--1353, 2020.

\bibitem[DEF{\etalchar{+}}03]{DBLP:journals/entcs/DowneyEFPR03}
Rodney~G. Downey, Vladimir Estivill{-}Castro, Michael~R. Fellows, Elena
  Prieto{-}Rodriguez, and Frances~A. Rosamond.
\newblock Cutting up is hard to do: the parameterized complexity of k-cut and
  related problems.
\newblock {\em Electron. Notes Theor. Comput. Sci.}, 78:209--222, 2003.

\bibitem[Die12]{DBLP:books/daglib/0030488}
Reinhard Diestel.
\newblock {\em Graph Theory, 4th Edition}, volume 173 of {\em Graduate texts in
  mathematics}.
\newblock Springer, 2012.

\bibitem[GH94]{GoldschmidtH94}
Olivier Goldschmidt and Dorit~S. Hochbaum.
\newblock A polynomial algorithm for the {$k$}-cut problem for fixed {$k$}.
\newblock {\em Math. Oper. Res.}, 19(1):24--37, 1994.

\bibitem[GLL18a]{DBLP:conf/focs/GuptaLL18}
Anupam Gupta, Euiwoong Lee, and Jason Li.
\newblock Faster exact and approximate algorithms for {$k$}-cut.
\newblock In Mikkel Thorup, editor, {\em 59th {IEEE} Annual Symposium on
  Foundations of Computer Science, {FOCS} 2018, Paris, France, October 7-9,
  2018}, pages 113--123. {IEEE} Computer Society, 2018.

\bibitem[GLL18b]{DBLP:conf/soda/GuptaLL18}
Anupam Gupta, Euiwoong Lee, and Jason Li.
\newblock An {FPT} algorithm beating {$2$}-approximation for {$k$}-cut.
\newblock In Artur Czumaj, editor, {\em Proceedings of the Twenty-Ninth Annual
  {ACM-SIAM} Symposium on Discrete Algorithms, {SODA} 2018, New Orleans, LA,
  USA, January 7-10, 2018}, pages 2821--2837. {SIAM}, 2018.

\bibitem[GLL19]{DBLP:conf/stoc/GuptaLL19}
Anupam Gupta, Euiwoong Lee, and Jason Li.
\newblock The number of minimum {$k$}-cuts: improving the karger-stein bound.
\newblock In Moses Charikar and Edith Cohen, editors, {\em Proceedings of the
  51st Annual {ACM} {SIGACT} Symposium on Theory of Computing, {STOC} 2019,
  Phoenix, AZ, USA, June 23-26, 2019}, pages 229--240. {ACM}, 2019.

\bibitem[GLL20]{DBLP:journals/corr/abs-1906-00417}
Anupam Gupta, Euiwoong Lee, and Jason Li.
\newblock The number of minimum {$k$}-cuts: Improving the karger-stein bound.
\newblock {\em To appear in STOC 2020}, abs/1906.00417, 2020.

\bibitem[Kar94]{Karger94}
David~R. Karger.
\newblock Using randomized sparsification to approximate minimum cuts.
\newblock In Daniel~Dominic Sleator, editor, {\em Proceedings of the Fifth
  Annual {ACM-SIAM} Symposium on Discrete Algorithms. 23-25 January 1994,
  Arlington, Virginia, {USA}}, pages 424--432. {ACM/SIAM}, 1994.

\bibitem[KL20]{DBLP:conf/soda/KawarabayashiL20}
Ken{-}ichi Kawarabayashi and Bingkai Lin.
\newblock A nearly 5/3-approximation {FPT} algorithm for min-\emph{k}-cut.
\newblock In Shuchi Chawla, editor, {\em Proceedings of the 2020 {ACM-SIAM}
  Symposium on Discrete Algorithms, {SODA} 2020, Salt Lake City, UT, USA,
  January 5-8, 2020}, pages 990--999. {SIAM}, 2020.

\bibitem[KS96]{KargerS96}
David~R. Karger and Clifford Stein.
\newblock A new approach to the minimum cut problem.
\newblock {\em J. {ACM}}, 43(4):601--640, 1996.

\bibitem[KT06]{DBLP:books/daglib/0015106}
Jon~M. Kleinberg and {\'{E}}va Tardos.
\newblock {\em Algorithm design}.
\newblock Addison-Wesley, 2006.

\bibitem[KT11]{KT11}
Ken{-}ichi Kawarabayashi and Mikkel Thorup.
\newblock The minimum $k$-way cut of bounded size is fixed-parameter tractable.
\newblock In {\em 52nd Annual {IEEE} Symposium on Foundations of Computer
  Science, FOCS, 2011, Palm Springs, CA, USA, October 22-25, 2011}, pages
  160--169. {IEEE} Computer Society, 2011.

\bibitem[Li19]{DBLP:conf/focs/Li19}
Jason Li.
\newblock Faster minimum {$k$}-cut of a simple graph.
\newblock In David Zuckerman, editor, {\em 60th {IEEE} Annual Symposium on
  Foundations of Computer Science, {FOCS} 2019, Baltimore, Maryland, USA,
  November 9-12, 2019}, pages 1056--1077. {IEEE} Computer Society, 2019.

\bibitem[Man18]{Manurangsi18}
Pasin Manurangsi.
\newblock Inapproximability of maximum biclique problems, minimum {$k$}-cut and
  densest at-least-\emph{k}-subgraph from the small set expansion hypothesis.
\newblock {\em Algorithms}, 11(1):10, 2018.

\bibitem[Mar06]{DBLP:journals/tcs/Marx06}
D{\'{a}}niel Marx.
\newblock Parameterized graph separation problems.
\newblock {\em Theor. Comput. Sci.}, 351(3):394--406, 2006.

\bibitem[MU05]{DBLP:books/daglib/0012859}
Michael Mitzenmacher and Eli Upfal.
\newblock {\em Probability and Computing: Randomized Algorithms and
  Probabilistic Analysis}.
\newblock Cambridge University Press, 2005.

\bibitem[NR01]{NaorR01}
Joseph Naor and Yuval Rabani.
\newblock Tree packing and approximating {$k$}-cuts.
\newblock In S.~Rao Kosaraju, editor, {\em Proceedings of the Twelfth Annual
  Symposium on Discrete Algorithms, January 7-9, 2001, Washington, DC, {USA}},
  pages 26--27. {ACM/SIAM}, 2001.

\bibitem[NSS95]{DBLP:conf/focs/NaorSS95}
Moni Naor, Leonard~J. Schulman, and Aravind Srinivasan.
\newblock Splitters and near-optimal derandomization.
\newblock In {\em 36th Annual Symposium on Foundations of Computer Science,
  Milwaukee, Wisconsin, USA, 23-25 October 1995}, pages 182--191. {IEEE}
  Computer Society, 1995.

\bibitem[RS08]{ravi2008approximating}
R~Ravi and Amitabh Sinha.
\newblock Approximating {$k$}-cuts using network strength as a lagrangean
  relaxation.
\newblock {\em European Journal of Operational Research}, 186(1):77--90, 2008.

\bibitem[SV95]{SaranV95}
Huzur Saran and Vijay~V. Vazirani.
\newblock Finding {$k$} cuts within twice the optimal.
\newblock {\em {SIAM} J. Comput.}, 24(1):101--108, 1995.

\bibitem[Tho90]{Thomas90}
Robin Thomas.
\newblock A {M}enger-like property of tree-width: The finite case.
\newblock {\em J. Comb. Theory, Ser. {B}}, 48(1):67--76, 1990.

\bibitem[Tho08]{Thorup08}
Mikkel Thorup.
\newblock Minimum {$k$}-way cuts via deterministic greedy tree packing.
\newblock In Cynthia Dwork, editor, {\em Proceedings of the 40th Annual {ACM}
  Symposium on Theory of Computing, Victoria, British Columbia, Canada, May
  17-20, 2008}, pages 159--166. {ACM}, 2008.

\end{thebibliography}
\end{document}